%% file: main.tex
\documentclass[11pt]{article}

\usepackage[margin=1in]{geometry}
\usepackage[T1]{fontenc}
\usepackage{lmodern}
\usepackage{amsmath,amssymb,amsthm}
\usepackage{booktabs}
\usepackage{array}
\usepackage{graphicx}
\usepackage{enumitem}
\usepackage{etoolbox}
\usepackage{pdflscape}
\usepackage{placeins}
\usepackage{longtable}
\usepackage[authoryear,round]{natbib}
\usepackage{xr-hyper}
\usepackage[hidelinks]{hyperref}
\usepackage{microtype}

\input{numbers}
\input{numbers_v2}

\providecommand{\R}{\mathbb{R}}
\providecommand{\E}{\mathbb{E}}
\providecommand{\Prob}{\mathbb{P}}
\providecommand{\ind}{\mathbf{1}}
\providecommand{\norm}[1]{\left\lVert #1 \right\rVert_2}
\providecommand{\abs}[1]{\left\lvert #1 \right\rvert}
\providecommand{\Fhat}{\widehat F}
\providecommand{\Fbar}{\bar F}
\providecommand{\cH}{\mathcal{H}}
\providecommand{\cG}{\mathcal{G}}
\providecommand{\cK}{\mathcal{K}}
\providecommand{\cP}{\mathcal{P}}

\providecommand{\neff}{n_{\mathrm{eff}}}
\providecommand{\nefft}{n_{\mathrm{eff},t}}

\newtheorem{theorem}{Theorem}
\newtheorem{lemma}{Lemma}
\newtheorem{proposition}{Proposition}

\theoremstyle{definition}
\newtheorem{assumption}{Assumption}
\newtheorem{definition}{Definition}
\theoremstyle{remark}
\newtheorem{remark}{Remark}
\input{preamble_v2}

\setlist[itemize]{itemsep=2pt,topsep=3pt}
\setlist[enumerate]{itemsep=2pt,topsep=3pt}
\title{Drift-Aware Spectral Conformal Prediction for\\ Non-Exchangeable Streaming Data: A Controlled Evaluation}
\author{Jeffery Opoku\\
\parbox{0.9\textwidth}{\centering\small School of Mathematical and Statistical Sciences, University of Texas Rio Grande Valley, Edinburg, Texas, USA}\\[4pt]
David Banahene\\
\parbox{0.9\textwidth}{\centering\small Department of Biostatistics, Robert Stempel College of Public Health and Social Work, Florida International University, Miami, Florida, USA}\\[4pt]
\parbox{0.9\textwidth}{\centering\small Corresponding author: Jeffery Opoku, School of Mathematical and Statistical Sciences, University of Texas Rio Grande Valley, Edinburg, Texas, USA; jeffery.opoku01@utrgv.edu.}}
\date{}

\begin{document}
\maketitle

\begin{abstract}
\input{abstract}
\end{abstract}

\noindent\textbf{Keywords:} conformal prediction; adaptive conformal inference; time series; spectral features; simulation study; reproducibility

\noindent\textbf{MSC 2020 classes:} \dcMSC

\noindent\textbf{Running head:} DASC: a controlled evaluation

\input{s1_introduction}
\input{s2_related}
\input{s3_method}
\input{s4_theory}
\input{s5_design}
\input{s6_results}
\input{s7_discussion}
\input{s8_conclusion}

\section*{Reproducibility}
\input{reproducibility}

\bibliographystyle{plainnat}
\bibliography{references}

\appendix
\input{appendix_short}

\end{document}

%% file: numbers.tex
\newcommand{\adElecAllMaxL}{0.058}
\newcommand{\adFinAllMaxL}{0.033}
\newcommand{\adWeaNaiveLMaxDten}{10}
\newcommand{\adWeaNaiveLMaxDzero}{28}
\newcommand{\adWeaNaiveLMinDten}{4}
\newcommand{\adWeaNaiveLMinDzero}{21}
\newcommand{\adWeaNaiveMedDMax}{0.018}
\newcommand{\adWeaNaiveMedDMin}{0.011}
\newcommand{\adWeaNaivePairsMax}{190}
\newcommand{\adWeaNaiveRefMaxDten}{20}
\newcommand{\adWeaNaiveRefMinDten}{4}
\newcommand{\adWeaRidgeMaxL}{0.000}
\newcommand{\agElecRidgeAciCovF}{0.8998}
\newcommand{\agElecRidgeAciWid}{2.04}
\newcommand{\agFinBoostAciCovF}{0.8987}
\newcommand{\agFinBoostAciIs}{3.37}
\newcommand{\agFinBoostAciWid}{2.05}
\newcommand{\agFinNaiveAciCovF}{0.8982}
\newcommand{\agFinNaiveAciIs}{3.74}
\newcommand{\agFinNaiveAciWid}{2.31}
\newcommand{\agFinRidgeAciCovF}{0.8993}
\newcommand{\agFinRidgeAciIs}{3.13}
\newcommand{\agFinRidgeAciWid}{1.91}
\newcommand{\agNNarrowWord}{one}
\newcommand{\agNarrowMax}{0.01}
\newcommand{\agRelWidFinMax}{6.2}
\newcommand{\agRelWidFinMin}{5.5}
\newcommand{\agRelWidMaxNonFin}{0.43}
\newcommand{\agWeaNaiveAciCovF}{0.8937}
\newcommand{\agWeaNaiveAciFreeCovF}{0.9001}
\newcommand{\agWeaNaiveAciFreeWid}{11.46}
\newcommand{\agWeaNaiveAciWid}{11.50}
\newcommand{\asElecRidgeAciGamma}{0.030}
\newcommand{\asWeaNaiveAciFreeGamma}{0.010}
\newcommand{\asWeaNaiveAciGamma}{0.050}
\newcommand{\bootBlockElec}{168}
\newcommand{\bootBlockFin}{21}
\newcommand{\bootN}{2{,}000}
\newcommand{\caChecks}{630}
\newcommand{\caFail}{0}
\newcommand{\caMethodsWord}{ten}
\newcommand{\caStreamsWord}{three}
\newcommand{\cfgAci}{12}
\newcommand{\cfgAciFree}{12}
\newcommand{\cfgDasc}{48}
\newcommand{\cfgDascFix}{48}
\newcommand{\cfgDascFree}{48}
\newcommand{\cfgDecay}{12}
\newcommand{\cfgPid}{24}
\newcommand{\cfgRoll}{2}
\newcommand{\cfgSpecAci}{36}
\newcommand{\cfgSpecFix}{12}
\newcommand{\cfgTotal}{254}
\newcommand{\dcAlphaPct}{90}
\newcommand{\dcBandwidthDefault}{0.55}
\newcommand{\dcBootLevel}{95}
\newcommand{\dcBreakFirst}{450}
\newcommand{\dcBreakSecond}{900}
\newcommand{\dcCapHi}{0.35}
\newcommand{\dcCapLo}{0.01}
\newcommand{\dcDailyFirstFit}{126}
\newcommand{\dcDailyLags}{1, 2, 3, 5, 10 and 21}
\newcommand{\dcDailyMeans}{5 and 21}
\newcommand{\dcDailyRefit}{21}
\newcommand{\dcDailyTrain}{504}
\newcommand{\dcDailyWins}{32 and 64}
\newcommand{\dcDecayGrid}{0.97, 0.985, 0.99, 0.995, 0.998 and 0.999}
\newcommand{\dcDeficitWin}{50}
\newcommand{\dcDevShare}{25}
\newcommand{\dcElecDevSteps}{1{,}920}
\newcommand{\dcElecLength}{12{,}000}
\newcommand{\dcElecPools}{720 and 1{,}440}
\newcommand{\dcEvalSeeds}{0 to 99}
\newcommand{\dcFinLength}{2{,}514}
\newcommand{\dcFinPools}{252 and 504}
\newcommand{\dcFinVolWin}{21}
\newcommand{\dcGammaDefault}{0.015}
\newcommand{\dcGammaGrid}{0.005, 0.01, 0.015, 0.02, 0.03 and 0.05}
\newcommand{\dcGammaGridWord}{six}
\newcommand{\dcHGrid}{0.35, 0.55 and 0.8}
\newcommand{\dcHGridShort}{0.35, 0.55, 0.8}
\newcommand{\dcHourlyFirstFit}{900}
\newcommand{\dcHourlyLags}{1, 2, 3, 24, 48 and 168}
\newcommand{\dcHourlyRefit}{168}
\newcommand{\dcHourlyTrain}{1{,}440}
\newcommand{\dcHourlyWins}{168 and 336}
\newcommand{\dcK}{6}
\newcommand{\dcLagSeasonal}{24}
\newcommand{\dcLambdaGrid}{0.15 and 0.45}
\newcommand{\dcLambdaGridShort}{0.15, 0.45}
\newcommand{\dcLegacyWins}{8 to 12}
\newcommand{\dcLr}{64}
\newcommand{\dcMSC}{62G15; 62M10}
\newcommand{\dcMinGated}{20}
\newcommand{\dcMinWindow}{eight}
\newcommand{\dcMminGrid}{0.25 or 0.5}
\newcommand{\dcNDataWord}{three}
\newcommand{\dcNumpy}{2.4}
\newcommand{\dcPandas}{2.3}
\newcommand{\dcPidKd}{0, 0.01}
\newcommand{\dcPidKi}{0.001, 0.003}
\newcommand{\dcPidKp}{0.01, 0.025, 0.05}
\newcommand{\dcPython}{3.13}
\newcommand{\dcRerunTol}{\ensuremath{10^{-9}}}
\newcommand{\dcRidgeAlpha}{1}
\newcommand{\dcRunTime}{roughly two hours on eight cores}
\newcommand{\dcSPName}{S\&P~500}
\newcommand{\dcScipy}{1.16}
\newcommand{\dcSiblingVersion}{2.0.0}
\newcommand{\dcSimBlock}{100}
\newcommand{\dcSimPools}{180 and 360}
\newcommand{\dcSimRidgeFirst}{120}
\newcommand{\dcSimRidgeRefit}{50}
\newcommand{\dcSimRidgeTrain}{300}
\newcommand{\dcSimSteps}{1{,}180}
\newcommand{\dcSklearn}{1.9}
\newcommand{\dcSpecAciGamma}{0.01, 0.015 and 0.03}
\newcommand{\dcStressSteps}{1{,}280}
\newcommand{\dcTuneSeeds}{1000 to 1009}
\newcommand{\dcWeaLength}{24{,}000}
\newcommand{\dcWeaPools}{1{,}080 and 2{,}160}
\newcommand{\dcWorstElec}{168}
\newcommand{\dcWorstFin}{63}
\newcommand{\dcXgbDepth}{4}
\newcommand{\dcXgbRate}{0.05}
\newcommand{\dcXgbSub}{0.8}
\newcommand{\dcXgbTrees}{200}
\newcommand{\dcXgboost}{3.3}
\newcommand{\ddAucDefHigherWord}{seven}
\newcommand{\ddHorizon}{300}
\newcommand{\ddNCombosWord}{seven}
\newcommand{\ddStGradualBrkFiveZeroZeroDDelay}{22}
\newcommand{\ddStGradualBrkFiveZeroZeroDFa}{0.063}
\newcommand{\ddStGradualBrkFiveZeroZeroDefDelay}{64}
\newcommand{\ddStGradualBrkFiveZeroZeroDefFa}{0.011}
\newcommand{\ddStMixedBrkEightFiveZeroDAuc}{0.75}
\newcommand{\ddStMixedBrkEightFiveZeroDDelay}{62}
\newcommand{\ddStMixedBrkEightFiveZeroDRate}{1.00}
\newcommand{\ddStMixedBrkEightFiveZeroDefAuc}{0.88}
\newcommand{\ddStMixedBrkEightFiveZeroDefDelay}{30}
\newcommand{\ddStMixedBrkEightFiveZeroDefRate}{0.98}
\newcommand{\ddThrPct}{99}
\newcommand{\ddThreeBrkFourFiveZeroDAuc}{0.52}
\newcommand{\ddThreeBrkFourFiveZeroDefAuc}{0.52}
\newcommand{\ddThreeBrkNineZeroZeroDAuc}{0.56}
\newcommand{\ddThreeBrkNineZeroZeroDAucSe}{0.02}
\newcommand{\ddThreeBrkNineZeroZeroDRate}{0.51}
\newcommand{\ddThreeBrkNineZeroZeroDefAuc}{0.91}
\newcommand{\ddThreeBrkNineZeroZeroDefAucSe}{0.01}
\newcommand{\ddThreeBrkNineZeroZeroDefRate}{0.97}
\newcommand{\dsWeaNaivePool}{2{,}160}
\newcommand{\fyCovJan}{0.750}
\newcommand{\fyCovJun}{0.952}
\newcommand{\fyCovMar}{0.913}
\newcommand{\fyDRankJan}{7}
\newcommand{\fyDRankJun}{3}
\newcommand{\fyDRankMar}{83}
\newcommand{\fyDRankMay}{8}
\newcommand{\fyDefRankJan}{26}
\newcommand{\fyDefRankMax}{9}
\newcommand{\fyDefRankMin}{3}
\newcommand{\fyNMonths}{91}
\newcommand{\gsRecFNNaiveAciProIs}{4.39}
\newcommand{\gsRecFNNaiveAciProIsSe}{0.01}
\newcommand{\gsRecFNNaiveAciProRone}{0.984}
\newcommand{\gsRecFNNaiveAciProRtwo}{0.812}
\newcommand{\gsRecFNNaiveDascIs}{4.17}
\newcommand{\gsRecFNNaiveDascIsSe}{0.01}
\newcommand{\gsRecFNNaiveDascRone}{0.951}
\newcommand{\gsRecFNNaiveDascRtwo}{0.844}
\newcommand{\gsRecFNNaiveSpecAciIs}{4.08}
\newcommand{\gsRecFNNaiveSpecAciIsSe}{0.01}
\newcommand{\gsRecFNNaiveSpecAciRone}{0.947}
\newcommand{\gsRecFNNaiveSpecAciRtwo}{0.849}
\newcommand{\lcFinNaiveAciFreeYr}{0.880}
\newcommand{\lcFinNaiveAciVolHigh}{0.888}
\newcommand{\lcFinNaiveAciVolLow}{0.875}
\newcommand{\lcFinNaiveAciVolMid}{0.884}
\newcommand{\lcFinNaiveAciYr}{0.857}
\newcommand{\lcFinNaiveDascFreeYr}{0.916}
\newcommand{\lcFinNaiveDascVolHigh}{0.901}
\newcommand{\lcFinNaiveDascVolLow}{0.917}
\newcommand{\lcFinNaiveDascVolMid}{0.904}
\newcommand{\lcFinNaiveDascYr}{0.892}
\newcommand{\lcFinNaiveDecayYr}{0.861}
\newcommand{\lcFinNaivePidYr}{0.892}
\newcommand{\lcFinNaiveRollYr}{0.749}
\newcommand{\lcFinNaiveSpecAciYr}{0.884}
\newcommand{\lcHourElecNaiveMax}{1.00}
\newcommand{\lcHourElecNaiveMin}{0.78}
\newcommand{\lcHourWeaNaiveMax}{0.93}
\newcommand{\lcHourWeaNaiveMin}{0.85}
\newcommand{\lcMadWeaNaiveAci}{0.014}
\newcommand{\lcMadWeaNaiveDasc}{0.014}
\newcommand{\lcMadWeaNaiveDascFix}{0.009}
\newcommand{\lcMadWeaNaiveRoll}{0.009}
\newcommand{\lcMadWeaNaiveSpecFix}{0.008}
\newcommand{\lrElecNaiveDascFreeBound}{0.0051}
\newcommand{\lrElecNaiveDascFreeGap}{0.0016}
\newcommand{\lrFinNaiveDascFreeBound}{0.0352}
\newcommand{\lrFinNaiveDascFreeGap}{0.0024}
\newcommand{\lrWeaNaiveDascFreeBound}{0.0034}
\newcommand{\lrWeaNaiveDascFreeGap}{0.0002}
\newcommand{\mxKernFourMax}{2.7}
\newcommand{\mxKernFourMin}{1.6}
\newcommand{\mxLargeFourU}{0.28}
\newcommand{\mxLargeN}{20{,}000}
\newcommand{\mxLargeTwoU}{0.27}
\newcommand{\mxMidFourU}{0.75}
\newcommand{\mxMidN}{2{,}160}
\newcommand{\mxMidTwoU}{0.72}
\newcommand{\mxOldRho}{0.62}
\newcommand{\mxOldValue}{6.02}
\newcommand{\mxSmallFourU}{1.61}
\newcommand{\mxSmallN}{360}
\newcommand{\mxSmallTwoU}{1.55}
\newcommand{\prAciIsNeg}{2}
\newcommand{\prAciIsPos}{3}
\newcommand{\prAciWidIncl}{3}
\newcommand{\prAciWidNeg}{1}
\newcommand{\prAciWidPos}{5}
\newcommand{\prElecBoostDascAciS}{\ensuremath{-}0.019}
\newcommand{\prElecBoostDascAciWAbs}{0.7}
\newcommand{\prElecBoostDascAciWci}{[\ensuremath{-}2.0, +0.5]}
\newcommand{\prElecNaiveDascAciWAbs}{1.1}
\newcommand{\prElecNaiveDascAciWci}{[\ensuremath{-}0.9, +3.4]}
\newcommand{\prElecRidgeDascAciS}{\ensuremath{-}0.024}
\newcommand{\prElecRidgeDascAciWAbs}{1.9}
\newcommand{\prElecRidgeDascAciWci}{[\ensuremath{-}3.1, \ensuremath{-}0.6]}
\newcommand{\prElecRidgeGateSLo}{0.000}
\newcommand{\prElecRidgeSpecAciWAbs}{0.1}
\newcommand{\prFinBoostFreeFreeSci}{[\ensuremath{-}0.068, +0.202]}
\newcommand{\prFinNaiveDascAciCpt}{+1.38}
\newcommand{\prFinNaiveDascAciS}{+0.066}
\newcommand{\prFinNaiveDascAciSci}{[\ensuremath{-}0.066, +0.202]}
\newcommand{\prFinNaiveDascAciWci}{[+5.4, +16.5]}
\newcommand{\prFinNaiveFreeFreeS}{\ensuremath{-}0.264}
\newcommand{\prFinNaiveFreeFreeSci}{[\ensuremath{-}0.575, \ensuremath{-}0.002]}
\newcommand{\prFinNaiveFreeProjS}{\ensuremath{-}0.213}
\newcommand{\prFinNaiveFreeProjSci}{[\ensuremath{-}0.376, \ensuremath{-}0.065]}
\newcommand{\prFinRidgeFreeFreeSci}{[\ensuremath{-}0.341, +0.130]}
\newcommand{\prFinWidMaxAci}{12.4}
\newcommand{\prFinWidMaxDecay}{7.8}
\newcommand{\prFinWidMaxSpec}{8.1}
\newcommand{\prFinWidMinAci}{10.6}
\newcommand{\prFinWidMinDecay}{7.0}
\newcommand{\prFinWidMinSpec}{7.0}
\newcommand{\prGateIsExcl}{two}
\newcommand{\prGateWidInclWord}{nine}
\newcommand{\prGateWidMax}{+0.7}
\newcommand{\prGateWidMin}{\ensuremath{-}1.6}
\newcommand{\prWeaBoostDascAciWAbs}{3.3}
\newcommand{\prWeaBoostDascAciWci}{[+1.1, +5.5]}
\newcommand{\prWeaNaiveDascAciWAbs}{10.6}
\newcommand{\prWeaNaiveDascAciWci}{[+8.4, +12.7]}
\newcommand{\prWeaNaiveDascDecayWAbs}{3.2}
\newcommand{\prWeaNaiveDascDecayWci}{[\ensuremath{-}5.6, \ensuremath{-}0.9]}
\newcommand{\prWeaNaiveDascSpecWAbs}{5.9}
\newcommand{\prWeaNaiveDascSpecWci}{[+4.6, +7.2]}
\newcommand{\prWeaNaiveFreeFreeSAbs}{0.444}
\newcommand{\prWeaNaiveFreeFreeWAbs}{2.2}
\newcommand{\prWeaNaiveFreeFreeWci}{[\ensuremath{-}3.3, \ensuremath{-}1.0]}
\newcommand{\prWeaRidgeDascAciW}{0.0}
\newcommand{\prWeaRidgeDascAciWci}{[\ensuremath{-}0.6, +0.6]}
\newcommand{\ptAciIsNeg}{0}
\newcommand{\ptAciIsPos}{1}
\newcommand{\ptAciWidFinMax}{12.4}
\newcommand{\ptAciWidFinMin}{7.1}
\newcommand{\ptAciWidPosFinWord}{three}
\newcommand{\ptElecNaiveTgSpecCpt}{\ensuremath{-}0.22}
\newcommand{\ptElecNaiveTgSpecSAbs}{0.047}
\newcommand{\ptElecNaiveTgSpecWAbs}{1.5}
\newcommand{\ptElecNaiveTgSpecWci}{[\ensuremath{-}2.9, \ensuremath{-}0.2]}
\newcommand{\ptElecRidgeTgDascW}{+1.8}
\newcommand{\ptElecRidgeTgDascWci}{[+0.7, +2.8]}
\newcommand{\ptFreeIncl}{six}
\newcommand{\ptFreeNarrow}{zero}
\newcommand{\ptFreeWide}{three}
\newcommand{\ptNonFinCellsWord}{six}
\newcommand{\ptWeaBoostTgAciW}{\ensuremath{-}0.2}
\newcommand{\ptWeaBoostTgAciWci}{[\ensuremath{-}0.7, +0.3]}
\newcommand{\ptWeaBoostTgDecaySAbs}{0.232}
\newcommand{\ptWeaBoostTgDecayWAbs}{6.2}
\newcommand{\ptWeaBoostTgDecayWci}{[\ensuremath{-}10.2, \ensuremath{-}2.3]}
\newcommand{\ptWeaBoostTgSpecCpt}{\ensuremath{-}0.22}
\newcommand{\ptWeaBoostTgSpecWAbs}{1.4}
\newcommand{\ptWeaBoostTgSpecWci}{[\ensuremath{-}2.5, \ensuremath{-}0.2]}
\newcommand{\ptWeaNaiveTgAciW}{\ensuremath{-}0.1}
\newcommand{\ptWeaNaiveTgAciWci}{[\ensuremath{-}0.3, +0.1]}
\newcommand{\ptWeaNaiveTgDascCpt}{+0.61}
\newcommand{\ptWeaNaiveTgDascW}{\ensuremath{-}9.7}
\newcommand{\ptWeaNaiveTgDascWci}{[\ensuremath{-}11.4, \ensuremath{-}8.0]}
\newcommand{\ptWeaNaiveTgDecayCpt}{+1.33}
\newcommand{\ptWeaNaiveTgDecayW}{\ensuremath{-}12.6}
\newcommand{\ptWeaNaiveTgFreeAciWAbs}{2.3}
\newcommand{\ptWeaNaiveTgFreeAciWci}{[+1.2, +3.5]}
\newcommand{\ptWeaNaiveTgSpecCpt}{+0.57}
\newcommand{\ptWeaNaiveTgSpecW}{\ensuremath{-}4.4}
\newcommand{\rgRecFNNaiveAciFreeROneCov}{0.993}
\newcommand{\rgRecFNNaiveAciFreeRTwoCov}{0.805}
\newcommand{\rgRecFNNaiveAciROneCov}{0.937}
\newcommand{\rgRecFNNaiveAciRTwoCov}{0.857}
\newcommand{\rgRecFNNaiveDascROneCov}{0.951}
\newcommand{\rgRecFNNaiveDascRTwoCov}{0.844}
\newcommand{\rgRecFNNaiveSpecAciROneCov}{0.926}
\newcommand{\rgRecFNNaiveSpecAciRTwoCov}{0.868}
\newcommand{\rgThreeNaiveDascFixACov}{0.894}
\newcommand{\rgThreeNaiveDascFixBCov}{0.894}
\newcommand{\rgThreeNaiveDascFixDriftCov}{0.883}
\newcommand{\rgThreeNaiveSpecFixACov}{0.896}
\newcommand{\rgThreeNaiveSpecFixBCov}{0.896}
\newcommand{\rgThreeNaiveSpecFixDriftCov}{0.872}
\newcommand{\rpGainElec}{28}
\newcommand{\rpGainWea}{42}
\newcommand{\rpHourlyMaxWidDiff}{0.0031}
\newcommand{\rpOldBound}{0.0010}
\newcommand{\rpOldCapShare}{65.1}
\newcommand{\rpOldDev}{\ensuremath{-}0.0656}
\newcommand{\rpOldProj}{\ensuremath{-}0.0651}
\newcommand{\rpOldT}{22{,}560}
\newcommand{\rpOldTele}{\ensuremath{-}0.0005}
\newcommand{\rpPubElecAciWid}{3.54}
\newcommand{\rpPubElecDascCov}{0.903}
\newcommand{\rpPubElecDascWid}{2.55}
\newcommand{\rpPubFinDascCov}{0.899}
\newcommand{\rpPubFinDascWid}{2.28}
\newcommand{\rpPubWeaAciWid}{11.03}
\newcommand{\rpPubWeaDascCov}{0.966}
\newcommand{\rpPubWeaDascCovPct}{96.6}
\newcommand{\rpPubWeaDascWid}{6.41}
\newcommand{\rpRecElecDascCov}{0.9000}
\newcommand{\rpRecElecDascWid}{3.54}
\newcommand{\rpRecElecRel}{0.0}
\newcommand{\rpRecWeaDascCov}{0.8959}
\newcommand{\rpRecWeaDascWid}{10.92}
\newcommand{\rpRecWeaRel}{\ensuremath{-}0.9}
\newcommand{\rpRepElecAciCov}{0.9001}
\newcommand{\rpRepElecAciWid}{3.54}
\newcommand{\rpRepElecDascCovF}{0.9033}
\newcommand{\rpRepElecDascWidF}{2.5497}
\newcommand{\rpRepFinDascCovF}{0.8988}
\newcommand{\rpRepFinDascWidF}{2.2791}
\newcommand{\rpRepWeaAciCov}{0.8938}
\newcommand{\rpRepWeaAciWid}{11.02}
\newcommand{\rpRepWeaDascCovF}{0.9655}
\newcommand{\rpRepWeaDascWidF}{6.4098}
\newcommand{\rvDascCovMax}{0.901}
\newcommand{\rvDascCovMin}{0.893}
\newcommand{\rvElecBoostDascIs}{2.95}
\newcommand{\rvElecBoostRollIs}{2.92}
\newcommand{\rvElecNaiveAciCov}{0.896}
\newcommand{\rvElecNaiveDascCov}{0.898}
\newcommand{\rvElecRidgeDascCovF}{0.8999}
\newcommand{\rvElecRidgeDascDbar}{0.125}
\newcommand{\rvElecRidgeDascIs}{2.99}
\newcommand{\rvElecRidgeDascWid}{2.05}
\newcommand{\rvElecRidgeRollIs}{2.96}
\newcommand{\rvElecRidgeSpecAciNeff}{1{,}434}
\newcommand{\rvElecRidgeSpecAciPool}{1{,}440}
\newcommand{\rvFinBoostDascCovF}{0.8998}
\newcommand{\rvFinBoostDascIs}{3.26}
\newcommand{\rvFinBoostDascNeff}{25}
\newcommand{\rvFinBoostDascWid}{2.17}
\newcommand{\rvFinNaiveAciCov}{0.888}
\newcommand{\rvFinNaiveAciFloor}{2.3}
\newcommand{\rvFinNaiveAciFreeCov}{0.902}
\newcommand{\rvFinNaiveAciProj}{+0.0126}
\newcommand{\rvFinNaiveDascCovF}{0.9014}
\newcommand{\rvFinNaiveDascFreeIs}{3.42}
\newcommand{\rvFinNaiveDascIs}{3.63}
\newcommand{\rvFinNaiveDascNeff}{62}
\newcommand{\rvFinNaiveDascWid}{2.46}
\newcommand{\rvFinRidgeAciCov}{0.894}
\newcommand{\rvFinRidgeDascCovF}{0.9008}
\newcommand{\rvFinRidgeDascIs}{3.09}
\newcommand{\rvFinRidgeDascWid}{2.02}
\newcommand{\rvIdGapMax}{\ensuremath{9\times 10^{-17}}}
\newcommand{\rvNCellsWord}{nine}
\newcommand{\rvNComparatorsWord}{nine}
\newcommand{\rvNMethodsWord}{ten}
\newcommand{\rvNWea}{18{,}000}
\newcommand{\rvNeffMaxFin}{245}
\newcommand{\rvNeffMinFin}{85}
\newcommand{\rvNeffRatioRngElec}{98.5 to 100.0}
\newcommand{\rvNeffRatioRngWea}{100.0}
\newcommand{\rvNeffRngElec}{709 to 1{,}440}
\newcommand{\rvNeffRngWea}{1{,}080}
\newcommand{\rvPoolMaxElec}{1{,}440}
\newcommand{\rvPoolMaxFin}{252}
\newcommand{\rvPoolMaxWea}{1{,}080}
\newcommand{\rvWeaBoostAciCov}{0.890}
\newcommand{\rvWeaBoostDascCov}{0.895}
\newcommand{\rvWeaNaiveAciCap}{17.7}
\newcommand{\rvWeaNaiveAciCapBound}{0.0177}
\newcommand{\rvWeaNaiveAciCov}{0.900}
\newcommand{\rvWeaNaiveAciCovF}{0.8997}
\newcommand{\rvWeaNaiveAciFloor}{2.4}
\newcommand{\rvWeaNaiveAciFloorBound}{0.0214}
\newcommand{\rvWeaNaiveAciGap}{+0.0003}
\newcommand{\rvWeaNaiveAciProj}{+0.0005}
\newcommand{\rvWeaNaiveAciTele}{\ensuremath{-}0.0001}
\newcommand{\rvWeaNaiveAciWid}{10.43}
\newcommand{\rvWeaNaiveDascCap}{3.8}
\newcommand{\rvWeaNaiveDascCov}{0.893}
\newcommand{\rvWeaNaiveDascCovF}{0.8934}
\newcommand{\rvWeaNaiveDascDbar}{0.011}
\newcommand{\rvWeaNaiveDascFloor}{1.3}
\newcommand{\rvWeaNaiveDascFreeAlphaMin}{\ensuremath{-}0.0125}
\newcommand{\rvWeaNaiveDascFreeCovF}{0.9002}
\newcommand{\rvWeaNaiveDascFreeNInf}{151}
\newcommand{\rvWeaNaiveDascFreeWid}{11.71}
\newcommand{\rvWeaNaiveDascGap}{+0.0066}
\newcommand{\rvWeaNaiveDascLbFifty}{0.80}
\newcommand{\rvWeaNaiveDascLbTen}{0.64}
\newcommand{\rvWeaNaiveDascNeff}{2{,}098}
\newcommand{\rvWeaNaiveDascProj}{+0.0067}
\newcommand{\rvWeaNaiveDascTele}{\ensuremath{-}0.0001}
\newcommand{\rvWeaNaiveDascWid}{11.53}
\newcommand{\rvWeaNaiveDecayCov}{0.886}
\newcommand{\rvWeaNaiveSpecAciCov}{0.894}
\newcommand{\scRecFNNaiveAciGamma}{0.050}
\newcommand{\scRecFNNaiveSpecAciGamma}{0.030}
\newcommand{\sgRecFNNaiveDascTgCov}{0.8951}
\newcommand{\sgRecFNNaiveDascTgCovSe}{0.0002}
\newcommand{\sgRecFNNaiveDascTgGamma}{0.050}
\newcommand{\sgRecFNNaiveDascTgIs}{4.014}
\newcommand{\sgRecFNNaiveDascTgIsSe}{0.013}
\newcommand{\sgRecFNNaiveDascTgWid}{3.051}
\newcommand{\sgRecFNNaiveDascTgWidSe}{0.010}
\newcommand{\smAdaptCovMax}{0.901}
\newcommand{\smAdaptCovMin}{0.895}
\newcommand{\smBestAci}{3}
\newcommand{\smBestDecay}{14}
\newcommand{\smBestRoll}{2}
\newcommand{\smBestSpecAci}{1}
\newcommand{\smBestSpecFix}{4}
\newcommand{\smCovDiffMax}{+0.0042}
\newcommand{\smCovDiffMin}{\ensuremath{-}0.0004}
\newcommand{\smDascCovMax}{0.900}
\newcommand{\smDascCovMin}{0.898}
\newcommand{\smDecayBestNarrow}{10}
\newcommand{\smFreeInMax}{2.2}
\newcommand{\smFreeInMin}{0.1}
\newcommand{\smFreeNIn}{23}
\newcommand{\smGateBetter}{6}
\newcommand{\smGateTie}{7}
\newcommand{\smGateWidMax}{+2.5}
\newcommand{\smGateWidMin}{\ensuremath{-}0.9}
\newcommand{\smGateWorse}{11}
\newcommand{\smIsWithinOneWord}{two}
\newcommand{\smIsWorseTwoWord}{twenty-two}
\newcommand{\smNCellsWord}{twenty-four}
\newcommand{\smNReps}{100}
\newcommand{\smNScenWord}{twelve}
\newcommand{\smNTuneSeedsWord}{ten}
\newcommand{\smRecFNNaiveAciIs}{4.01}
\newcommand{\smRecFNNaiveAciIsF}{4.014}
\newcommand{\smSeWidMax}{0.2}
\newcommand{\smSeWidMin}{0.1}
\newcommand{\smWidMax}{3.3}
\newcommand{\smWidMin}{0.3}
\newcommand{\spRecFNNaiveFreeFreeSAbs}{0.30}
\newcommand{\spRecFNNaiveFreeFreeWAbs}{3.0}
\newcommand{\tgBudgetDasc}{60}
\newcommand{\tgNAciIsBetter}{1}
\newcommand{\tgNAciIsWorse}{20}
\newcommand{\tgNAciNarrow}{0}
\newcommand{\tgNAciWidMean}{+0.30}
\newcommand{\tgNAciWide}{14}
\newcommand{\tgNDascNarrow}{15}
\newcommand{\tgNDecayIsBetter}{2}
\newcommand{\tgNDecayIsWorse}{20}
\newcommand{\tgNDecayNarrow}{2}
\newcommand{\tgNDecayWide}{21}
\newcommand{\tgNFreeAciIsBetter}{5}
\newcommand{\tgNFreeAciIsWorse}{14}
\newcommand{\tgNFreeAciNarrow}{3}
\newcommand{\tgNFreeAciWide}{6}
\newcommand{\tgNSpecIsBetter}{12}
\newcommand{\tgNSpecIsWorse}{7}
\newcommand{\tgNSpecNarrow}{14}
\newcommand{\tgNSpecWide}{6}
\newcommand{\tgRerunMatch}{54}
\newcommand{\tgRerunTotal}{54}
\newcommand{\tgStageOne}{6}
\newcommand{\tgStageTwo}{6}
\newcommand{\tgWeaNaiveDascTgCov}{0.899}
\newcommand{\tgWeaNaiveDascTgCovF}{0.8995}
\newcommand{\tgWeaNaiveDascTgWid}{10.41}
\newcommand{\thAdLeCount}{10}
\newcommand{\thAdLminMax}{0.08}
\newcommand{\thAdLpathMax}{0.62}
\newcommand{\thAdLpathMin}{0.61}
\newcommand{\thAdSeeds}{10}
\newcommand{\thAdSeedsWord}{ten}
\newcommand{\thAdSlopeMax}{0.16}
\newcommand{\thCombMarginMax}{0.30}
\newcommand{\thCombMarginMin}{0.20}
\newcommand{\thCombRuns}{20}
\newcommand{\thDriftRandomCases}{5{,}000}
\newcommand{\thDriftRandomExcess}{0.44}
\newcommand{\thDriftRandomOld}{1{,}230}
\newcommand{\thDriftSyntheticCases}{11{,}800}
\newcommand{\thDriftSyntheticOld}{557}
\newcommand{\thEpsEta}{0.1}
\newcommand{\thEpsNeff}{300}
\newcommand{\thEpsNew}{0.13}
\newcommand{\thEpsOld}{0.07}
\newcommand{\thGridConfigsWord}{sixteen}
\newcommand{\thGridMaxExceed}{0.0005}
\newcommand{\thGridNs}{60 and 360}
\newcommand{\thGridReps}{4{,}000}
\newcommand{\thLrIdGapMax}{\ensuremath{2\times 10^{-15}}}
\newcommand{\thMixAtOne}{18}
\newcommand{\thMixCells}{20}
\newcommand{\thMixEmp}{0.068}
\newcommand{\thMixIid}{0.0015}
\newcommand{\thMixInfo}{two}
\newcommand{\thMixM}{10}
\newcommand{\thMixRatio}{45}
\newcommand{\thMixReps}{20{,}000}
\newcommand{\thMixTau}{10}
\newcommand{\thOldClaim}{\ensuremath{10^{-3}}}
\newcommand{\thOldGap}{23}
\newcommand{\thPsHiOneMax}{100}
\newcommand{\thPsHiOneMin}{92}
\newcommand{\thPsNeffMax}{338}
\newcommand{\thPsNeffMin}{241}
\newcommand{\thPsSeeds}{20}
\newcommand{\thPsSlackMax}{0.33}
\newcommand{\thPsSlackMin}{0.16}
\newcommand{\thPsSteps}{23{,}600}
\newcommand{\tmAllMax}{1.8}
\newcommand{\tmPoolSim}{360}
\newcommand{\tmPoolWea}{2{,}160}
\newcommand{\tmSimDascMs}{0.39}
\newcommand{\tmUniMaxSim}{206}
\newcommand{\tmUniMaxWea}{375}
\newcommand{\tmUniMinSim}{138}
\newcommand{\tmUniMinWea}{288}
\newcommand{\tmWeaDascMs}{0.95}
\newcommand{\wbElecAciMax}{0.88}
\newcommand{\wbElecAciMin}{0.87}
\newcommand{\wbElecDascMax}{0.86}
\newcommand{\wbElecDascMin}{0.85}
\newcommand{\wbWeaAciMax}{0.85}
\newcommand{\wbWeaAciMin}{0.78}
\newcommand{\wbWeaDascMax}{0.84}
\newcommand{\wbWeaDascMin}{0.76}

%% file: numbers_v2.tex
\newcommand{\adConfPctSim}{95}

\newcommand{\adEtaAdjMax}{0.0026}
\newcommand{\adEtaAdjMin}{0.0025}

\newcommand{\adJMax}{20}
\newcommand{\adJMin}{19}

\newcommand{\avAlpha}{0.1}

\newcommand{\avCZeroAlphaThree}{0.371}
\newcommand{\avCZeroAlphaTwo}{0.340}
\newcommand{\avChainConst}{428}

\newcommand{\avGapSlopeMax}{1.03}
\newcommand{\avGapSlopeMin}{0.99}

\newcommand{\avMatchErrA}{0.0008}
\newcommand{\avMatchErrB}{0.0022}
\newcommand{\avMatchStepA}{0.035}
\newcommand{\avMatchStepB}{0.0575}
\newcommand{\avNCells}{36}

\newcommand{\avOwnStepA}{0.015}
\newcommand{\avOwnStepB}{0.03}
\newcommand{\avRatioA}{2.3}
\newcommand{\avRatioB}{1.9}
\newcommand{\avRhoB}{2}
\newcommand{\avRhoC}{3}

\newcommand{\avWA}{0.75}
\newcommand{\avWB}{0.9}

\newcommand{\exBlkEffMax}{1.4}
\newcommand{\exBlkEffMin}{0.2}
\newcommand{\exBlkElecBlocksWord}{five}

\newcommand{\exBlkLongElec}{2160}
\newcommand{\exBlkLongFin}{126}
\newcommand{\exBlkMidElec}{720}
\newcommand{\exBlkMidFin}{63}
\newcommand{\exBlkNDraws}{243}

\newcommand{\exBlkSpecWeaBoostWMidWord}{one}

\newcommand{\exDascNearUniformNWord}{six}
\newcommand{\exDesignDate}{2 November 2026}
\newcommand{\exEdgeExtTopAciWord}{three}

\newcommand{\exEdgeExtTopSpecAciWord}{five}
\newcommand{\exEdgeOrigTopAciWord}{six}
\newcommand{\exEdgeOrigTopDascWord}{five}
\newcommand{\exEdgeOrigTopSpecAciWord}{eight}
\newcommand{\exGammaList}{0.005, 0.01, 0.015, 0.02, 0.03, 0.05 and 0.1}
\newcommand{\exGammaTop}{0.1}

\newcommand{\exNCellsWord}{nine}
\newcommand{\exNGammaWord}{seven}
\newcommand{\exNHOrigWord}{three}

\newcommand{\exNHourCellsWord}{six}

\newcommand{\exNearHiMax}{0.002}
\newcommand{\exNearNWord}{three}
\newcommand{\exNearRelBoost}{1.4}
\newcommand{\exNearRelNaive}{1.3}
\newcommand{\exNearUniformThr}{0.9}
\newcommand{\exNeffDascExtHourMax}{1.00}
\newcommand{\exNeffDascExtHourMin}{0.87}
\newcommand{\exNeffDascFixExtHourBelowHalfWord}{four}

\newcommand{\exNeffSpecAciExtHourMax}{1.00}
\newcommand{\exNeffSpecAciExtHourMin}{0.77}

\newcommand{\exNeffSpecFixExtHourBelowHalfWord}{five}

\newcommand{\exOrigRatioElecMax}{13.4}
\newcommand{\exOrigRatioElecMin}{3.9}

\newcommand{\exOrigRatioWeaMax}{137.5}
\newcommand{\exOrigRatioWeaMin}{36.9}
\newcommand{\exPrimAciSInclWord}{nine}

\newcommand{\exPrimAciWFinMax}{+12.4}
\newcommand{\exPrimAciWFinMin}{+7.1}
\newcommand{\exPrimAciWInclWord}{six}
\newcommand{\exPrimAciWNonFinMax}{+0.4}
\newcommand{\exPrimAciWNonFinMin}{\ensuremath{-}0.5}
\newcommand{\exPrimAciWeaBoostIs}{\ensuremath{-}0.066}

\newcommand{\exPrimDecayWAdvWord}{three}
\newcommand{\exPrimDecayWDisWord}{four}
\newcommand{\exPrimDecayWeaBoostCpt}{+0.01}
\newcommand{\exPrimDecayWeaBoostW}{\ensuremath{-}5.7}
\newcommand{\exPrimDecayWeaBoostWci}{[\ensuremath{-}9.8, \ensuremath{-}1.6]}

\newcommand{\exPrimSpecWDisWord}{two}

\newcommand{\exPriorSpecNarrowNWord}{three}

\newcommand{\exSaciAciWeaBoostIci}{[\ensuremath{-}0.197, \ensuremath{-}0.028]}
\newcommand{\exSaciAciWeaBoostIs}{\ensuremath{-}0.102}

\newcommand{\exScaleList}{0.25, 0.5, 1 and 2}

\newcommand{\exSeedNSeedsWord}{nine}
\newcommand{\exSeedRows}{162}
\newcommand{\exSeedSAdvTotal}{19}

\newcommand{\exSimNReps}{50}

\newcommand{\exSimNotRunWord}{eleven}

\newcommand{\exSweepHourOrigMax}{1.00}
\newcommand{\exSweepHourOrigMin}{0.98}

\newcommand{\exTopCellsWord}{five}

\newcommand{\exTopNewList}{0.15, 0.2 and 0.3}

\newcommand{\exWideAciElecBoostW}{+2.7}
\newcommand{\exWideAciElecBoostWci}{[+0.6, +4.9]}

\newcommand{\exWideAciWeaNaiveIs}{\ensuremath{-}0.160}

\newcommand{\loBTwoLargePowerPath}{10}
\newcommand{\loBTwoMarginMin}{0.81}
\newcommand{\loBTwoMarginOkLarge}{9}
\newcommand{\loBTwoMidPowerPath}{4}
\newcommand{\loBTwoMinCover}{0.943}
\newcommand{\loBootDraws}{400}
\newcommand{\loCZeroLargePowerPath}{6}
\newcommand{\loCZeroMaxSize}{0.007}
\newcommand{\loCZeroMidPowerPath}{0}
\newcommand{\loCellsPerScale}{12}
\newcommand{\loCoverThr}{0.93}
\newcommand{\loExploreN}{360}
\newcommand{\loExploreRho}{0.99}
\newcommand{\loExploreSize}{0.113}
\newcommand{\loGridCells}{36}
\newcommand{\loGridPts}{64}
\newcommand{\loLLarge}{0.48}
\newcommand{\loLMid}{0.15}
\newcommand{\loLSmall}{0.073}
\newcommand{\loMarginNominal}{0.95}
\newcommand{\loMarginPct}{95}
\newcommand{\loNList}{360 and 2160}
\newcommand{\loNScalesWord}{three}
\newcommand{\loPhiS}{0.95}
\newcommand{\loPowerThr}{0.8}
\newcommand{\loRepsLarge}{500}
\newcommand{\loRepsSmall}{300}
\newcommand{\loRhoList}{0, 0.5 and 0.9}

\newcommand{\pfAElecC}{+0.33}
\newcommand{\pfAElecW}{\ensuremath{-}27.9}
\newcommand{\pfAWeaC}{+6.96}
\newcommand{\pfAWeaW}{\ensuremath{-}41.3}
\newcommand{\pfBNAppWord}{eight}
\newcommand{\pfBNDagWord}{four}
\newcommand{\pfBNNonDagWord}{four}
\newcommand{\pfBNonDagMax}{+3.3}
\newcommand{\pfBNonDagMin}{\ensuremath{-}2.3}
\newcommand{\pfBWMax}{+10.3}
\newcommand{\pfBWMin}{\ensuremath{-}2.3}
\newcommand{\pfBWeaNaiveC}{\ensuremath{-}0.60}
\newcommand{\pfBWeaNaiveIs}{+19.8}
\newcommand{\pfBWeaNaiveW}{+10.3}
\newcommand{\pfBWeaNaiveWci}{[+8.1, +12.3]}
\newcommand{\pfCAbsorbMax}{89}
\newcommand{\pfCAbsorbMin}{88}
\newcommand{\pfCFixSimNaiveWMed}{+1.4}
\newcommand{\pfCFixSimRidgeWMed}{+1.3}
\newcommand{\pfCFixWMax}{+2.0}
\newcommand{\pfCFixWMin}{+0.3}
\newcommand{\pfCSimNaiveWMed}{+0.2}
\newcommand{\pfCSimRidgeWMed}{+0.1}
\newcommand{\pfCWMax}{+1.4}
\newcommand{\pfCWMin}{\ensuremath{-}0.3}
\newcommand{\pfCovTol}{0.5}
\newcommand{\pfDElecMax}{\ensuremath{-}0.1}
\newcommand{\pfDElecMin}{\ensuremath{-}0.2}
\newcommand{\pfDFinMax}{\ensuremath{-}3.3}
\newcommand{\pfDFinMin}{\ensuremath{-}4.0}
\newcommand{\pfDFixCMin}{\ensuremath{-}25.71}
\newcommand{\pfDFixWMin}{\ensuremath{-}54.9}
\newcommand{\pfDWeaMax}{\ensuremath{-}2.0}
\newcommand{\pfDWeaMin}{\ensuremath{-}26.8}
\newcommand{\pfDWeaNaiveC}{\ensuremath{-}10.21}
\newcommand{\pfDWeaNaiveIs}{+27.5}
\newcommand{\pfDWeaNaiveW}{\ensuremath{-}26.8}
\newcommand{\pfDWeaNaiveWci}{[\ensuremath{-}29.2, \ensuremath{-}24.4]}
\newcommand{\pfEDoubleWMax}{+0.7}
\newcommand{\pfEDoubleWMin}{\ensuremath{-}2.9}
\newcommand{\pfEHalfWMax}{+3.1}
\newcommand{\pfEHalfWMin}{\ensuremath{-}0.1}
\newcommand{\pfGammaAlt}{0.015}
\newcommand{\pfNPitfallsWord}{five}
\newcommand{\pfNReps}{50}
\newcommand{\pfNScenWord}{twelve}
\newcommand{\pfNStreamsWord}{nine}
\newcommand{\pfSimSeFixMax}{1.23}
\newcommand{\pfSimSeMax}{0.62}
\newcommand{\pfSimSeMed}{0.13}
\newcommand{\phCalLowestNWord}{three}
\newcommand{\phHourRsqMax}{0.09}
\newcommand{\phHourRsqMin}{0.01}
\newcommand{\phMagMax}{0.000}
\newcommand{\phMagMin}{\ensuremath{-}0.002}
\newcommand{\phNCellsWord}{four}
\newcommand{\phPhaseBetterNWord}{three}
\newcommand{\phPhaseMax}{0.99}
\newcommand{\phPhaseMin}{0.72}

\newcommand{\thAdEtaAdjMax}{0.0038}
\newcommand{\thAdEtaAdjMin}{0.0029}
\newcommand{\thAdJMax}{17}
\newcommand{\thAdJMin}{13}

\newcommand{\thEpsRatio}{1.8}
\newcommand{\thPsBFiveFailMax}{3.7}
\newcommand{\thPsBFiveFailMin}{3.1}

%% file: abstract.tex
Does drift-aware spectral conformal prediction (DASC) narrow the intervals of adaptive conformal inference (ACI) when all methods share one forecaster per data set and one tuning criterion? Not with the projected update. At DASC's step size, projected DASC and ACI widths differ by at most \agRelWidMaxNonFin{} per cent on electricity and weather; DASC is wider on finance. With DASC's step size tuned at a larger budget, or with wider step-size sets and scaled bandwidths, paired width intervals against ACI include zero or lie above it. Coverage is near nominal for the adaptive methods, with some under-coverage at the selected settings. The paper contributes a reproducible simulation and resampling design, released with code and saved results: \smNScenWord{} simulated scenarios and \dcNDataWord{} real streams, paired moving-block bootstrap intervals with block-length and seed checks, a perturbation test that every interval uses only past observations (\caFail{} failures in \caChecks{} checks), and one-at-a-time measurements of how evaluation choices move widths. An earlier preprint reported width reductions of \rpGainElec{} and \rpGainWea{} per cent that came from centring DASC on a different forecast from the baselines. The replacement theory states the status of each result; the earlier long-run theorem was false for the projected update.

%% file: s1_introduction.tex
\section{Introduction}
\label{sec:intro}

Does drift-aware spectral conformal prediction (DASC), which weights calibration residuals by spectral similarity, narrow the prediction intervals of adaptive conformal inference (ACI) once every method uses the same point forecaster and the same tuning rule? With the projected update used in the experiments, it does not. At DASC's step size and pool, the widths of DASC (projected) and ACI (projected) differ by at most \agRelWidMaxNonFin{} per cent on electricity and weather, and DASC (projected) is wider on finance (Section~\ref{sec:results-real}). With DASC's step size tuned at a larger budget, or with wider step-size sets and scaled bandwidths, the paired width intervals against ACI include zero or lie above it (Sections~\ref{sec:results-equal} and~\ref{sec:robust}). The paper contributes a reproducible simulation and resampling design: tuning on development data only, paired moving-block bootstrap intervals with checks of block length and bootstrap seed, a perturbation test that every interval depends only on past observations (\caChecks{} checks, \caFail{} failures), and one-at-a-time measurements of how common evaluation choices move reported widths (Sections~\ref{sec:design}, \ref{sec:robust} and~\ref{sec:pitfalls}). The code and the saved results are public (Reproducibility section).

This paper corrects an earlier preprint by the same authors \citep{opoku2026dasc}; Section~\ref{sec:results-anchor} traces the width gains reported there to the forecast on which DASC was centred, and Appendix~\ref{app:replaced} lists each statement of the preprint that this paper replaces.

Conformal prediction wraps a point forecaster and turns its past residuals into prediction intervals with a long-run coverage guarantee when the data are exchangeable. Streams of electricity demand, temperature or financial volatility are not exchangeable, and two families of remedies address this. Online methods correct the miscoverage level as outcomes arrive: ACI of \citet{gibbs2021adaptive} moves the level by a fixed step after each miss or cover, and its relatives adapt the step or add control terms \citep{zaffran2022adaptive,gibbs2022conformal,angelopoulos2023conformal}. Weighted methods let past residuals count in proportion to their similarity to the present, either through fixed time-decay weights \citep{barber2023beyond} or through weights that depend on the test point \citep{guan2023localized,hore2025local}. A growing set of papers combines the two, with similarity weights on covariates, attention weights or retrieved neighbours and an ACI-type level update \citep{lai2026online,zhu2025attention,heurich2026rarecp,tu2026online}.

DASC belongs to that last group. It weights calibration residuals by the similarity of banded periodogram features of a recent window, shrinks the calibration pool when a drift score signals a regime change, and updates the level by the ACI rule. The earlier preprint reported that DASC narrowed the prediction intervals of ACI by \rpGainElec{} per cent on household electricity demand and by \rpGainWea{} per cent on hourly temperature, and proved a per-step and a long-run coverage guarantee.

Neither the width gains nor the long-run theorem holds. In the hourly experiments the DASC interval was centred on the previous observation, while the baselines were centred on the observation \dcLagSeasonal{} hours earlier. The comparison therefore mixed the effect of the method with the effect of a better point forecast. With DASC centred on the same forecast as the baselines, and nothing else changed, its width relative to that of ACI is \rpRecElecRel{} per cent on electricity and \rpRecWeaRel{} per cent on weather (Section~\ref{sec:results-anchor}). The theory had the following faults. The spectral window was defined through the observation being predicted, so the coverage argument was invalid as written. The long-run theorem was false for the projected update used in the experiments. A remark on dependence contained an arithmetic error. This paper reports the corrected evaluation and replaces the theory.

The corrected evaluation repeats the comparison under one rule. Every method on a data set uses the same point forecaster, one tuning criterion (the mean interval score on the same development data), the same evaluation steps and paired moving-block bootstrap intervals. The study covers \smNScenWord{} simulated scenarios and \dcNDataWord{} real streams, with \rvNMethodsWord{} methods and three forecasters per real stream. It is organised by the ADEMP structure for simulation studies \citep{morris2019using} (Section~\ref{sec:design}). The authors proposed DASC, so the study is not a neutral comparison; Section~\ref{sec:design-aims} says what limits the room for favourable choices. Following common practice, which keeps the point forecast fixed across methods \citep{stocker2025gentle,sabashvili2026conformal}, we hold the forecaster fixed and tune the step sizes by one rule. What we add is the size of the effect of not doing so, shown on a specific case.

\paragraph{Contributions.}
\begin{enumerate}[label=(\arabic*)]
\item \emph{A like-for-like evaluation.} Under one shared forecaster, one tuning criterion, matched evaluation steps and paired intervals, DASC (projected) does not beat ACI (projected). The exceptions are in the interval score: \tgNAciIsBetter{} of the \smNCellsWord{} simulated cells (judged by two standard errors), and some real cells at longer block lengths of the paired bootstrap, with differences of at most \exNearRelBoost{} per cent of the score and no gain in width or coverage (named below). Against decay-weighted conformal prediction it is ahead in some cells and behind in others. Against spectral ACI it is ahead in some cells with the main-run grids and in none when every method searches the same step sizes (Section~\ref{sec:robust}). For the free update, DASC is narrower than free ACI by more than two standard errors in \tgNFreeAciNarrow{} of the \smNCellsWord{} simulated cells. The cells that favour DASC are reported with their coverage differences: the real-data cells in Sections~\ref{sec:results-real} and~\ref{sec:results-equal}, and the simulated cells as counts in Section~\ref{sec:results-equal} and cell by cell in Tables~\ref{S-tab:dascgamma-sim-paired-naive1} and~\ref{S-tab:dascgamma-sim-paired-ridge}. Coverage is near nominal for the adaptive methods, with examples of under-coverage in Section~\ref{sec:discussion}. At DASC's step size and pool, the widths of DASC (projected) and ACI (projected) differ by at most \agRelWidMaxNonFin{} per cent on electricity and weather, and DASC (projected) is \agRelWidFinMin{} to \agRelWidFinMax{} per cent wider on finance. At the equal budget the paired width intervals against ACI include zero in all \ptNonFinCellsWord{} electricity and weather cells and lie above zero in the finance cells, also at longer block lengths and with other bootstrap seeds. For the projected update, at the block length and the bootstrap seed of the main analysis, no paired interval for the interval-score difference against ACI excludes zero in DASC's favour on the three real streams; the free update on finance with the naive forecaster is the exception (Section~\ref{sec:results-real}). At \exBlkMidElec{} hours such an interval lies below zero in every seed in electricity with the ridge forecaster (the equal-budget study) and in weather with gradient boosting (the extended designs); the draws at \exBlkLongElec{} hours, where a resample of electricity holds \exBlkElecBlocksWord{} blocks, are a sensitivity check and not a verdict (Section~\ref{sec:robust}). Both comparisons are given: ACI tuned and DASC at its default step size, and DASC with its step size tuned over \tgBudgetDasc{} configurations against \cfgAci{} for ACI.
\item \emph{The reproduction anchor.} The earlier width gains are a centring artefact. The new code reproduces the published numbers under the original settings and then recentres DASC (Section~\ref{sec:results-anchor}, Table~\ref{tab:reproduction}).
\item \emph{What the components do.} In the one simulated design built for spectral weighting, the weights balance regimes better than ACI at an equal slow step size, and ACI with a tuned step size reaches a similar balance without spectral information, although spectral ACI has a lower interval score with the ridge forecaster by a small margin (Section~\ref{sec:robust}): the weights behave as a substitute for a faster step. On electricity and weather the kernel weights are close to uniform at the tuned bandwidths of this study, also when the bandwidths are scaled to the feature spread, for the methods with an adaptive level (Section~\ref{sec:robust}). The gate, measured at a fixed level, has a small effect with both signs. The drift score is a weaker break detector than a rolling coverage deficit (Sections~\ref{sec:results-sim} and~\ref{sec:results-weights}).
\item \emph{Robustness to the grids.} The comparison of DASC (projected) with ACI (projected) is repeated with one step-size set of \exNGammaWord{} values for every method and bandwidths scaled to the feature spread. The criterion was written before the runs. At the block length and the bootstrap seed of the main analysis, no real cell shows a DASC advantage over ACI with a paired interval that excludes zero in width or in interval score, and in two weather cells the interval-score limit is within \exNearHiMax{} of zero. The width statement also holds at longer block lengths and with other bootstrap seeds. The interval-score statement depends on the block length: at \exBlkMidElec{} hours the interval lies below zero in every seed in weather with gradient boosting in both extended designs (Section~\ref{sec:robust}).
\item \emph{Sizes of common evaluation pitfalls.} Centring on a different forecast, a frozen pool, a step-size mismatch, the finite-sample rank and the length of the pool are measured one at a time on the same streams (Section~\ref{sec:pitfalls}). The pitfalls are known; the sizes are measured here.
\item \emph{Repaired theory} (Section~\ref{sec:theory}). A measurability lemma shows that every quantity used at time $t$ is a function of the data before $t$. A per-step bound at the adaptive level $1-\alpha_t$ takes its lower bound from Theorem~1 of the sibling paper \citep{opoku2026spectral} and adds the upper bound and a proved concentration step. The free update satisfies the identity of \citet{gibbs2021adaptive}. For the projected update we restate the identity and sign structure of \citet{lai2026online} for $[\alpha_{\min},\alpha_{\max}]$, give a constant-stream counterexample to a two-sided $O(1/T)$ bound, and report the share of steps at the cap in the runs. A refutation diagnostic gives a lower confidence bound on the Lipschitz constant in the condition of \citet{guan2023localized}; it is not in the localised conformal papers we opened. A blocking argument under $\beta$-mixing in the form of \citet{oliveira2024split} is correct and vacuous at the pool sizes used here. For the free update, Section~\ref{sec:avg} gives an averaged-recursion account of why a faster ACI step can reproduce the regime coverage of kernel weights in an idealised model. The block identity and its pathwise antisymmetry are Proposition~4.1 of \citet{gibbs2021adaptive} applied over blocks. We did not find in the sources we opened the periodic orbit of the averaged recursion with its bound and sharp lower bound, or the common-mode shift and matching step of the weighted quantile.
\item \emph{Causality is tested.} A perturbation test of the code ran \caChecks{} checks with \caFail{} failures, and it detects a deliberately leaky forecaster (Section~\ref{sec:design-causality}).
\end{enumerate}
The projected identity, the Lipschitz condition, the combination of kernel weights with an ACI update, the blocking argument, and the block identity and pathwise antisymmetry of the averaged-recursion account are in the literature cited and are not claimed as new. The paper does not claim that DASC is better, narrower or more accurate than any comparator.

Section~\ref{sec:related} places the work in the literature, and Section~\ref{sec:method} defines the \rvNMethodsWord{} methods. Section~\ref{sec:theory} states the theory, Section~\ref{sec:design} gives the design and Section~\ref{sec:results} the results. Section~\ref{sec:discussion} states the limits of the study, and Section~\ref{sec:conclusion} concludes.

%% file: s2_related.tex
\section{Related work}
\label{sec:related}

\paragraph{Online conformal inference.} \citet{gibbs2021adaptive} update the miscoverage level after every step and prove a long-run coverage identity for every data sequence. The performance of the update depends on the step size: a large step adapts quickly and a small step is stable \citep{gibbs2021adaptive,gibbs2022conformal,zaffran2022adaptive}. \citet{angelopoulos2023conformal} view the update as an error integrator within a PID scheme and give a guarantee that also covers a saturated integrator whose saturation levels lie outside the score range. The PID-style baseline of this paper is a different controller, on the level and not on the quantile.

\paragraph{Weighted and localised calibration.} \citet{barber2023beyond} bound the coverage of weighted split conformal prediction, with weights fixed in advance, by a weight-averaged total-variation mismatch. \citet{guan2023localized} and \citet{hore2025local} localise the calibration with test-point-dependent weights for exchangeable data. For time series, \citet{lee2026kernel} use kernel weights under strong mixing and obtain an asymptotic rate, and \citet{xu2023conformal} and \citet{chernozhukov2018exact} give asymptotic coverage results under strong mixing. \citet{barber2026predictive} study why unweighted split conformal prediction works under $\beta$-mixing, and \citet{oliveira2024split} give non-asymptotic bounds by a blocking argument; neither treats weights.

\paragraph{Similarity weights with an adaptive level.} Online localised conformal prediction \citep{lai2026online} places kernel weights on covariates over a rolling window and updates the level by the projected ACI step. Related methods use attention weights with a miscoverage update \citep{zhu2025attention}, retrieval weights with ACI \citep{heurich2026rarecp}, Hopfield-network retrieval weights (combined with ACI in an appendix) \citep{auer2023conformal} and similarity weights with an adaptive level for panel data \citep{tu2026online}. Spectral ACI, in this paper, is an instance of this family with a periodogram feature, and DASC adds a drift gate. We found no third-party conformal method with spectral, periodogram or frequency features in the sources we opened; the spectral feature choice is that of the authors' own preprints. The combination of similarity weights with an ACI update is not claimed as new.

\paragraph{Projection of the level.} Projecting the level onto an interval adds boundary terms to the ACI identity. \citet[Proposition~3.1 and Appendix~B]{lai2026online} state the pathwise identity, the signs of the boundary terms and bounds on them, for projection onto $[0,1]$. The sibling paper \citep[Remark~4]{opoku2026spectral} states that the projection adds boundary terms and reports how often it is active, for the range it uses.

\paragraph{Smoothness of the residual law.} The condition that the conditional distribution function of the score is Lipschitz in the covariates, with respect to a distance, is Assumption~5.1(i) of \citet{guan2023localized}. Later localised methods use variants \citep{lee2026kernel}. We found no test, estimator or confidence bound for the constant in the localised conformal papers we opened.

\paragraph{Evaluation practice.} Comparisons of conformal methods for time series usually fix one point forecaster for all methods. \citet{stocker2025gentle} describe this as a design choice that isolates the effect of the conformal step. \citet{sabashvili2026conformal} uses one base forecaster for all methods and default hyperparameters for fairness. The sensitivity of online methods to the step size is documented in \citet{gibbs2021adaptive} and \citet{zaffran2022adaptive}. None of the sources we opened states the consequence of centring one method on a different forecast, or of tuning one method's step size and fixing another's, as a pitfall for width comparisons. We do not claim that the pitfall has not been noticed; we measure its size in a case.

\paragraph{The sibling paper.} The sibling paper \citep{opoku2026spectral}, version \dcSiblingVersion{}, studies spectral similarity weights with the ACI update (the method called spectral ACI here) on simulated and real series. Its Theorem~1 states a one-sided lower coverage bound at the adaptive level on an assumed concentration event. Its Remark~2 states that the Lipschitz constant is not observable, and its Remark~4 and supplement report the projection. It reports no benefit of spectral weights beyond a recency kernel on its nine real series, and a null effect for its own effective-sample-size gate. Its evaluation practice uses a shared point forecaster, a common projection and one step size for the single-level ACI-type methods. The results of this paper on three other real streams are consistent with its conclusion about real series, although the mechanism differs: here the tuned kernel weights are close to uniform on electricity and weather.

%% file: s3_method.tex
\section{Methods under comparison}
\label{sec:method}

This section defines drift-aware spectral conformal prediction (DASC) exactly as it is implemented in the code that produced every result below, and the \rvNComparatorsWord{} methods it is compared with. All \rvNMethodsWord{} methods take the same point forecast as input. None of them computes a forecast.

\subsection{Setting and notation}
\label{sec:setting}

Let $(Y_t)_{t\ge 1}$ be a real-valued stream and $\cH_t=\sigma(Y_1,\ldots,Y_t)$ the history up to time $t$, with $\cH_0$ trivial. A point forecaster $\widehat f_t$ is an $\cH_{t-1}$-measurable random variable, for example $\widehat f_t=Y_{t-1}$ or $\widehat f_t=Y_{t-\dcLagSeasonal}$. Calendar information known in advance may enter $\widehat f_t$. The residual is $R_t=\abs{Y_t-\widehat f_t}$. For a target miscoverage level $\alpha=0.1$, every method forms the interval
\[
C_t=\bigl[\widehat f_t-q_t,\ \widehat f_t+q_t\bigr],\qquad E_t=\ind\{Y_t\notin C_t\},
\]
where $q_t$ is a quantile of past residuals. The aim is a miss rate near $\alpha$ together with narrow intervals. Coverage is always reported next to width, and the interval score defined in Section~\ref{sec:design-measures} combines the two.

\subsection{Spectral features}
\label{sec:features}

Let $\ell$ be a window length and $K=\dcK$ the number of bands. The spectral window at time $t$ is $W_t=(Y_{t-\ell},\ldots,Y_{t-1})$. It ends at $t-1$ and does not contain $Y_t$. The feature $S_t=\Phi(W_t)\in\R^K$ is computed as follows: subtract the window mean, take the squared modulus of the discrete Fourier transform, drop the zero frequency, split the positive frequencies into $K$ consecutive groups of nearly equal size and average the periodogram within each group, then divide by the Euclidean norm. Normalisation makes the distance between two features depend on the shape of the spectrum and not on its overall level, so that two windows with the same frequency pattern and different amplitude have nearly the same feature. A window with fewer than \dcMinWindow{} observations (the first steps of a stream) gets the zero vector. Alternative feature maps change only the constants in Section~\ref{sec:theory}.

\input{s3_phase_remark}

\subsection{Drift-aware spectral conformal prediction}
\label{sec:dasc}

At time $t$ the calibration pool is $\cP_t=\{\max(1,t-m_{\max}),\ldots,t-1\}$, restricted to indices whose residual is defined. The pool residuals $R_i$, $i\in\cP_t$, are all observed before $t$. DASC computes, in this order:
\begin{enumerate}[label=(\arabic*)]
\item pre-gate weights $w^{\mathrm{pre}}_{i,t}\propto\exp\{-\norm{S_i-S_t}^2/h^2\}$ for $i\in\cP_t$, normalised to sum to one, with bandwidth $h>0$;
\item the drift score $D_t=\bigl\|\bar S^{\mathrm{rec}}_t-\sum_{i\in\cP_t}w^{\mathrm{pre}}_{i,t}S_i\bigr\|_2$, where $\bar S^{\mathrm{rec}}_t$ is the mean of $S_{t-\ell_r},\ldots,S_{t-1}$ and $\ell_r=\dcLr$ for every data set;
\item the gate $m_t=\lfloor m_{\max}-(m_{\max}-m_{\min})\min\{1,D_t/\lambda\}\rfloor$ with threshold $\lambda>0$, and the gated pool $\cP^{g}_t=\{i\in\cP_t:\ i\ge t-m_t\}$, replaced by the $\dcMinGated$ most recent indices of $\cP_t$ when it has fewer than $\dcMinGated$ elements;
\item final weights $w_{i,t}\propto\exp\{-\norm{S_i-S_t}^2/h^2\}$ for $i\in\cP^{g}_t$, normalised, and the two diagnostics $\nefft=1/\sum_{i}w_{i,t}^2$ and $\bar d_t=\sum_{i}w_{i,t}\norm{S_i-S_t}$;
\item the weighted quantile $q_t(a)=\inf\{r:\ \sum_{i\in\cP^{g}_t}w_{i,t}\ind\{R_i\le r\}\ge 1-a\}$ at the current level $a=\alpha_t$, and the interval $C_t=[\widehat f_t-q_t(\alpha_t),\widehat f_t+q_t(\alpha_t)]$;
\item after $Y_t$ is observed, the level update $\alpha_{t+1}=\alpha_t+\gamma(\alpha-E_t)$ with $\alpha_1=\alpha$.
\end{enumerate}
Write $\Fhat_{w,t}(r)=\sum_{i\in\cP^{g}_t}w_{i,t}\ind\{R_i\le r\}$ for the weighted empirical distribution function of the residuals in the gated pool, with the final weights $w_{i,t}$, so that $q_t(a)=\inf\{r:\ \Fhat_{w,t}(r)\ge 1-a\}$. In the projected variant the update is followed by projection onto $[\alpha_{\min},\alpha_{\max}]=[\dcCapLo,\dcCapHi]$. In the free variant there is no projection, and the interval is the whole real line when $\alpha_t\le 0$ and the empty set when $\alpha_t\ge 1$ (Section~\ref{sec:conventions}). The fixed-level variant keeps $\alpha_t\equiv\alpha$. The step size is $\gamma=\dcGammaDefault$ in the main runs, the value used in the earlier preprint, and is tuned in the equal-budget study of Section~\ref{sec:results-equal}. The earlier preprint fixed the bandwidth at $h=\dcBandwidthDefault$; here $h$ is tuned (Section~\ref{sec:design-tuning}).

Two things set DASC apart from a plain adaptive method with an equal-weight pool. The weights $w_{i,t}$ depend on the data through the features, and the gate $m_t$ depends on the data through $D_t$. Both are functions of $Y_1,\ldots,Y_{t-1}$ only (Lemma~\ref{lem:measurable}).

\subsection{The \rvNComparatorsWord{} comparison methods}
\label{sec:comparators}

Table~\ref{tab:methods} lists the \rvNMethodsWord{} methods. They differ in how the pool residuals are weighted (equal, decay or kernel weights, with a gated pool in the DASC rows), in how the level $\alpha_t$ is chosen (fixed, the ACI update, or the controller of the PID-style row) and in the parameters tuned.

\begin{table}[t]
\centering
\caption{The \rvNMethodsWord{} methods. ``Pool'' is the maximum number of past residuals used. The last column gives the number of configurations tried for each forecaster (Section~\ref{sec:design-tuning}); the pool size is part of every grid.}
\label{tab:methods}
\small
\begin{tabular}{p{3.0cm}p{3.7cm}p{3.3cm}p{3.1cm}r}
\toprule
Method & Calibration weights & Level $\alpha_t$ & Tuned parameters & Configs \\
\midrule
Rolling conformal & equal & fixed & pool & \cfgRoll \\
ACI (projected) & equal & ACI update, projected & pool, $\gamma$ & \cfgAci \\
ACI (free) & equal & ACI update, free & pool, $\gamma$ & \cfgAciFree \\
PID-style & equal & controller on $\alpha_t$, projected & pool, $k_p,k_i,k_d$ & \cfgPid \\
Decay-weighted & $\propto \theta^{\,t-i}$ & fixed & pool, $\theta$ & \cfgDecay \\
Spectral (fixed level) & kernel on $S_i-S_t$ & fixed & pool, $h$, $\ell$ & \cfgSpecFix \\
Spectral ACI (no gate) & kernel on $S_i-S_t$ & ACI update, projected & pool, $h$, $\gamma$, $\ell$ & \cfgSpecAci \\
DASC (fixed level) & kernel, gated pool & fixed & pool, $h$, $\lambda$, $m_{\min}$, $\ell$ & \cfgDascFix \\
DASC (projected) & kernel, gated pool & ACI update, projected & pool, $h$, $\lambda$, $m_{\min}$, $\ell$ & \cfgDasc \\
DASC (free) & kernel, gated pool & ACI update, free & pool, $h$, $\lambda$, $m_{\min}$, $\ell$ & \cfgDascFree \\
\bottomrule
\end{tabular}
\end{table}

Rolling conformal is split conformal prediction on a moving window of the last residuals. ACI is the adaptive conformal inference of \citet{gibbs2021adaptive} with the equal-weight pool, in a projected and a free version. The PID-style method is the controller on the miscoverage level that the earlier project code used, $\alpha_{t+1}=\Pi\bigl(\alpha_t+k_p e_t+k_i\sum_{s\le t}e_s+k_d(e_t-e_{t-1})\bigr)$ with $e_t=\alpha-E_t$ and $\Pi$ the projection onto $[\dcCapLo,\dcCapHi]$. It is not the quantile-tracking conformal PID control of \citet{angelopoulos2023conformal}. The decay-weighted method is weighted split conformal prediction with exponentially decaying weights, the time-decay weighting of \citet{barber2023beyond}, normalised over the pool with no point mass on the test point. Spectral conformal prediction with a fixed level uses the kernel weights of Section~\ref{sec:dasc} without drift score or gate. Spectral ACI adds the ACI update to the kernel weights and has no gate. It is the sibling method \citep{opoku2026spectral}, and it follows \citet{lai2026online} and \citet{zhu2025attention} in combining similarity-weighted calibration with an ACI update, using a periodogram feature.

DASC is spectral ACI plus the drift score and the gate. Hence ``DASC projected minus spectral ACI'' differs by the gate and by whatever the tuning selected for each method (pool, bandwidth, feature window, step size). The pair ``DASC with a fixed level minus spectral conformal with a fixed level'' has no step size, so the step-size confound is absent from it; differences in the tuned pool, bandwidth and feature window remain.

\subsection{Conventions that apply to every method}
\label{sec:conventions}

\begin{itemize}
\item \emph{One quantile operator.} Every method uses the weighted quantile of Section~\ref{sec:dasc} with weights normalised over the pool and no mass on the test point. With equal weights this is one order statistic below the quantile used by the original baseline code. The effect is at most one order statistic of the pool.
\item \emph{Undefined residuals} (before a forecaster's first forecast) are excluded from the pool. The original DASC code kept them in the pool for the first $m_{\max}$ steps.
\item \emph{Free update outside $[0,1]$.} When $\alpha_t\le 0$ the interval is the whole real line, so $E_t=0$ and the width is infinite. When $\alpha_t\ge 1$ the interval is empty, so $E_t=1$ and the interval score is that of the degenerate interval $\{\widehat f_t\}$. These are the conventions under which the long-run identity of Theorem~\ref{thm:free} holds. Steps with infinite width are counted and excluded from the means of width and interval score; in a paired comparison they are excluded for both methods. Coverage is computed over every evaluation step, so a step with the whole line counts as covered.
\item \emph{Non-negative targets.} For the finance data, whose target is an absolute return, the lower limit is replaced by $\max(\text{lower limit},0)$ for every method, and the width is measured after this replacement.
\item \emph{Level relaxation.} The earlier code could shift the level used for the quantile by a function of $D_t$ and $\neff$. It is switched off in every reported run.
\end{itemize}

\subsection{Computational cost per step}
\label{sec:cost}

Given the feature vectors, one DASC step costs $O\bigl(m_{\max}(K+\log m_{\max})\bigr)$ operations: $m_{\max}K$ for the distances and weights, $O(\ell_rK)$ for the drift score, and $m_{\max}\log m_{\max}$ for the sort inside the weighted quantile. The feature of the new window costs one Fourier transform of length $\ell$, that is $O(\ell\log\ell)$. Rolling conformal and ACI cost $O(m_{\max}\log m_{\max})$ per step. In the implementation used here the weights are recomputed over the whole pool at every step, with no incremental update. Section~\ref{sec:results-cost} reports measured times.

%% file: s3_phase_remark.tex
\begin{remark}[Band features do not see the phase of a cycle]
\label{rem:phase}
Let the series have a period that divides the window length $\ell$. When the window advances by one step, its discrete Fourier transform satisfies $X_k(s+1)=e^{2\pi i k/\ell}X_k(s)$ for every bin $k$. The moduli $\abs{X_k}$, the periodogram and the band features above therefore do not change, and the features carry no information about the position of the window in the daily cycle. This is the standard shift property of the transform. On the hourly streams, hour of day explains between \phHourRsqMin{} and \phHourRsqMax{} per cent of the variance of the band features, on average over the bands, for windows of \dcHourlyWins{} hours. The rank correlation between the feature distance of two steps and the circular difference of their hours lies between \phMagMin{} and \phMagMax{} for the band features, and between \phPhaseMin{} and \phPhaseMax{} for a distance on the phase of the daily frequency. Table~\ref{tab:v2-hour} adds a phase term, and a term in the clock hour, to the kernel. The phase term lowers the hour-of-day deviation of ACI in \phPhaseBetterNWord{} of the \phNCellsWord{} cells, but it does not improve on the calendar-hour kernel. Where time stamps exist, the calendar-hour kernel gives the largest reduction of the deviation in \phCalLowestNWord{} of the \phNCellsWord{} cells. In the two electricity cells its deviation is smaller than that of the phase kernel, and the interval of the difference excludes zero; in weather with the ridge forecaster the two kernels match within their intervals. The fourth cell (weather, seasonal naive forecaster) is a tie: no variant reduces the deviation there with an interval that excludes zero. Tuned DASC does not differ from ACI in any cell.
\end{remark}

\input{v2_hour_of_day}

%% file: v2_hour_of_day.tex
\begin{table}[t]
\centering
\caption{Coverage by hour of day with magnitude features, a phase kernel and a calendar-hour kernel. The deviation is the mean over the 24 clock hours of the absolute difference between the coverage in that hour and $\dcAlphaPct$ per cent, on the evaluation steps; ``noise floor'' is the value expected for ACI if the coverage in every hour were exactly right, estimated by shuffling hour labels within weeks. DASC is tuned over the equal-budget grid. The phase and calendar kernels use the same grids and the same tuning rule as each other. Panel B gives differences in the deviation (first method minus second) with 95 per cent paired moving-block bootstrap intervals; a negative value means a smaller deviation for the first method.}
\label{tab:v2-hour}
\footnotesize
\setlength{\tabcolsep}{3pt}
\textbf{A. Deviation of the hourly coverage from the target}\par\smallskip
\begin{tabular}{lccccc}
\toprule
Cell & ACI & DASC tuned & Phase kernel & Calendar kernel & Noise floor (ACI) \\
\midrule
Electricity, seasonal naive & 0.0606 & 0.0604 & 0.0422 & 0.0322 & 0.0128 \\
Electricity, rolling ridge & 0.0527 & 0.0531 & 0.0318 & 0.0174 & 0.0121 \\
Weather, seasonal naive & 0.0141 & 0.0137 & 0.0126 & 0.0132 & 0.0086 \\
Weather, rolling ridge & 0.0693 & 0.0693 & 0.0134 & 0.0118 & 0.0085 \\
\bottomrule
\end{tabular}
\par\bigskip
\textbf{B. Paired differences in the deviation}\par\smallskip
\resizebox{\textwidth}{!}{\begin{tabular}{lcccc}
\toprule
Cell & DASC tuned $-$ ACI & Phase $-$ ACI & Calendar $-$ ACI & Phase $-$ calendar \\
\midrule
Electricity, seasonal naive & $-$0.0002 [$-$0.0026, +0.0022] & $-$0.0184 [$-$0.0239, $-$0.0093] & $-$0.0284 [$-$0.0341, $-$0.0186] & +0.0100 [+0.0052, +0.0149] \\
Electricity, rolling ridge & +0.0004 [$-$0.0019, +0.0027] & $-$0.0209 [$-$0.0253, $-$0.0124] & $-$0.0353 [$-$0.0397, $-$0.0257] & +0.0145 [+0.0078, +0.0191] \\
Weather, seasonal naive & $-$0.0004 [$-$0.0013, +0.0007] & $-$0.0016 [$-$0.0041, +0.0027] & $-$0.0009 [$-$0.0041, +0.0032] & $-$0.0007 [$-$0.0026, +0.0020] \\
Weather, rolling ridge & $-$0.0001 [$-$0.0013, +0.0012] & $-$0.0559 [$-$0.0572, $-$0.0476] & $-$0.0575 [$-$0.0586, $-$0.0491] & +0.0016 [$-$0.0015, +0.0048] \\
\bottomrule
\end{tabular}}
\end{table}

%% file: s4_theory.tex
\section{Theory}
\label{sec:theory}

This section states what is proved about the method, under which assumptions, and what is not proved. Full proofs are in Appendix~\ref{app:proofs}. Each statement carries its status in Table~\ref{tab:theory-status} at the end of the section. The statements replace the theory of the earlier preprint. The section prefers weak true statements to strong ones that cannot be proved. In particular, no finite-sample coverage guarantee is claimed without assumptions, because distribution-free conditional coverage is impossible without structure \citep{barber2021limits}. The per-step bound below holds under a model and is not a distribution-free guarantee.

\subsection{Causality}
\label{sec:causal}

Everything used at time $t$ must be determined by the history before $t$. This is the property that makes the coverage probability at time $t$ a function of the conditional law of $R_t$ alone.

\begin{lemma}[Everything used at time $t$ is determined by the past]
\label{lem:measurable}
For every $t$, the quantities $S_t$, $w^{\mathrm{pre}}_{i,t}$, $D_t$, $m_t$, $\cP^{g}_t$, $w_{i,t}$, $\nefft$, $\bar d_t$, $\alpha_t$ and $q_t(\alpha_t)$ are $\cH_{t-1}$-measurable. With $F_t(r)=\Prob\{R_t\le r\mid\cH_{t-1}\}$,
\begin{equation}
\label{eq:cover-is-F}
\Prob\{Y_t\in C_t\mid\cH_{t-1}\}=F_t\bigl(q_t(\alpha_t)\bigr)\qquad\text{almost surely.}
\end{equation}
\end{lemma}

\begin{proof}
$S_j$ is a measurable function of $(Y_{j-\ell},\ldots,Y_{j-1})$, hence $\cH_{t-1}$-measurable for $j\le t$. The pre-gate weights, $D_t$, $m_t$, the gated pool, the final weights, $\nefft$ and $\bar d_t$ are measurable functions of $(S_j)_{j\le t}$. The residuals $R_i$, $i\le t-1$, are $\cH_{t-1}$-measurable because $\widehat f_i$ is $\cH_{i-1}$-measurable. By induction $E_s=\ind\{R_s>q_s(\alpha_s)\}$ is $\cH_s$-measurable for $s\le t-1$, so $\alpha_t$, a function of $E_1,\ldots,E_{t-1}$, is $\cH_{t-1}$-measurable, and so is $q_t(\alpha_t)$. Finally $\{Y_t\in C_t\}=\{R_t\le q_t(\alpha_t)\}$ because the interval is closed, and the substitution property of regular conditional distributions gives \eqref{eq:cover-is-F}.
\end{proof}

The earlier preprint defined the spectral window through $Y_t$. With that definition $q_t$ depends on $Y_t$ and \eqref{eq:cover-is-F} is false, although the code was causal. Section~\ref{sec:design-causality} reports a perturbation test of the code.

\subsection{A per-step bound at the adaptive level}
\label{sec:perstep}

Fix $t$. Let $\cG_t=\sigma(S_1,\ldots,S_t)$ be the feature path and $\cK_{t-1}=\cG_t\vee\sigma(R_1,\ldots,R_{t-1})$ the feature path together with the past residuals. By Lemma~\ref{lem:measurable} the weights, the gated pool, $\nefft$ and $\bar d_t$ are $\cG_t$-measurable. The level $\alpha_t$ and the quantile $q_t(\alpha_t)$ are $\cK_{t-1}$-measurable but not $\cG_t$-measurable, because $\alpha_t$ depends on past residuals through $E_1,\ldots,E_{t-1}$.

\begin{assumption}[Lipschitz regime model]
\label{ass:A}
There is a family of distribution functions $\{F(\cdot\mid s):\ s\in\R^K\}$, measurable in $s$, and constants $L\ge 0$, $\delta\ge 0$, such that for every $t$:
\begin{enumerate}[label=(\roman*)]
\item conditionally on $\cG_t$ the residuals $R_1,\ldots,R_t$ are independent, and $R_i$ has distribution function $F(\cdot\mid S_i)$;
\item $\sup_r\abs{F(r\mid s)-F(r\mid s')}\le L\norm{s-s'}+\delta$ for all $s,s'$ in the range of $\Phi$.
\end{enumerate}
\end{assumption}

Part (i) makes explicit what is needed to treat the pool residuals as independent draws from the regime laws. Part (ii) says that feature vectors that are close have distribution functions that are close in the uniform distance, up to a slack. It is the condition of \citet[Assumption~5.1(i)]{guan2023localized}, with a norm on the feature space and an additive slack $\delta$ for mismatch that the feature does not resolve, and it is Assumption~2 of the sibling paper \citep{opoku2026spectral} without the slack.

\begin{remark}[Status of Assumption~\ref{ass:A}]
\label{rem:A-status}
Part (i) is an idealisation and cannot hold exactly for the forecasters used here. With $\widehat f_i=Y_{i-1}$ the residual $R_i$ is a function of $(Y_{i-1},Y_i)$, and $Y_i$ is a coordinate of every window $W_j$ with $i+1\le j\le i+\ell$. So $R_i$ and $\cG_t$ are functions of the same observations, and (i) asserts that the normalised band powers carry no information about an individual residual beyond $F(\cdot\mid S_i)$. This is a statement about what the feature map discards. It is plausible as an approximation when $\ell$ is large compared with the two coordinates that a single residual involves, and when $\Phi$ is a normalised average over the window; each coordinate then enters $S_j$ with influence of order $1/\ell$. No approximate version is proved. Section~\ref{sec:mixing} gives the only route known to us that avoids (i), and it is too weak to be useful at the pool sizes used here. Appendix~\ref{app:checks} reports numerical checks in simulated models; they are evidence about the configurations tried and not a proof.
\end{remark}

Two lemmas control the weighted pool. The first bounds the bias from using residuals whose features differ from the test feature; the second bounds the sampling error of the weighted empirical distribution function.

\begin{lemma}[Weighted calibration bias]
\label{lem:bias}
Under Assumption~\ref{ass:A}(ii), with $\Fbar_{w,t}(r)=\sum_{i\in\cP^{g}_t}w_{i,t}F(r\mid S_i)$,
\[
\sup_r\abs{\Fbar_{w,t}(r)-F(r\mid S_t)}\le L\bar d_t+\delta .
\]
\end{lemma}

\begin{lemma}[Uniform concentration of a weighted empirical distribution]
\label{lem:grid}
Let $w_1,\ldots,w_n\ge 0$ with $\sum_iw_i=1$ and $\neff=1/\sum_iw_i^2$, and let $R_1,\ldots,R_n$ be independent with distribution functions $F_1,\ldots,F_n$. Put $\Fhat(r)=\sum_iw_i\ind\{R_i\le r\}$ and $\Fbar(r)=\sum_iw_iF_i(r)$. For every integer $K\ge 2$ and every $\eta\in(0,1)$,
\[
\Prob\Bigl\{\sup_{r\in\R}\abs{\Fhat(r)-\Fbar(r)}>\sqrt{\frac{\log(4K/\eta)}{2\neff}}+\frac1K\Bigr\}\le\eta .
\]
\end{lemma}

Lemma~\ref{lem:grid} holds for the weighted empirical distribution function of independent, non-identically distributed variables and is uniform in $r$. Uniformity is what allows it to be applied at the residual-dependent point $q_t(\alpha_t)$. The proof is a grid argument with Hoeffding's inequality \citep{hoeffding1963probability}.

\begin{theorem}[Approximate coverage at the adaptive level]
\label{thm:perstep}
Suppose Assumption~\ref{ass:A} holds. Fix $t$ and $\eta\in(0,1)$, let $K_t\ge 2$ be a $\cG_t$-measurable integer (the default is $K_t=\max\{2,\lceil\nefft\rceil\}$), and put
\[
\varepsilon_t(\eta)=\sqrt{\frac{\log(4K_t/\eta)}{2\nefft}}+\frac1{K_t},\qquad
\pi_t=\sum_{i\in\cP^{g}_t}w_{i,t}\ind\{R_i=q_t(\alpha_t)\}.
\]
There is an event $\mathcal B_t\in\cK_{t-1}$ with $\Prob(\mathcal B_t\mid\cG_t)\ge 1-\eta$ such that on $\mathcal B_t$, whenever $\alpha_t\in[0,1]$,
\begin{align}
\Prob\{Y_t\in C_t\mid\cK_{t-1}\}&\ge 1-\alpha_t-L\bar d_t-\delta-\varepsilon_t(\eta),\label{eq:perstep-lower}\\
\Prob\{Y_t\in C_t\mid\cK_{t-1}\}&\le 1-\alpha_t+L\bar d_t+\delta+\varepsilon_t(\eta)+\pi_t .\label{eq:perstep-upper}
\end{align}
When the pool residuals are distinct, $\pi_t\le\max_iw_{i,t}$; under ties $\pi_t$ can be larger. The event $\mathcal B_t$ does not depend on $\alpha_t$, so the bounds hold for the level produced by either adaptive rule and for a fixed level.
\end{theorem}

In words: on an event whose probability given the feature path is at least the chosen confidence level, the conditional coverage probability lies within the bias bound and the sampling term of the level in force, with the weight that falls on the quantile itself added on the upper side. The bounds hold at the level $1-\alpha_t$ in force at time $t$, and not at $1-\alpha$. The quantities $\bar d_t$, $\nefft$ and $\pi_t$ can be computed at prediction time. The constants $L$ and $\delta$ cannot, and Section~\ref{sec:check} shows that data can refute a pair $(L,\delta)$ but cannot certify one. ``Computable bound'' therefore means computable up to two constants that must be assumed.

\emph{Relation to earlier work.} Barber et al. \citeyearpar[Theorem~2]{barber2023beyond} bound the coverage of weighted split conformal prediction by $1-\alpha$ minus a weight-averaged total-variation mismatch, for weights fixed in advance. They note that the result extends to data-dependent weights when the total-variation terms are conditioned on the weights, and leave the details open. Theorem~\ref{thm:perstep} is obtained by a different route. The total-variation terms conditioned on the weights are not computable from a single stream, and for the DASC weights no bound on them is proved here. The lower bound \eqref{eq:perstep-lower} has the form of Theorem~1 of the sibling paper \citep{opoku2026spectral}, which states it at the adaptive level on an assumed concentration event (its Assumption~3). What Theorem~\ref{thm:perstep} adds is the upper bound \eqref{eq:perstep-upper}, Lemma~\ref{lem:measurable}, and a proved concentration step (Lemma~\ref{lem:grid}). Localisation by weights that depend on the test point, with finite-sample guarantees by other means, is studied for exchangeable data by \citet{guan2023localized} and \citet{hore2025local}.

\begin{remark}[Reading the bound]
\label{rem:perstep-reading}
(a) In the earlier preprint the projected level sat at its cap $\alpha_{\max}=\dcCapHi$ in \rpOldCapShare{} per cent of the steps of the weather run, so \eqref{eq:perstep-lower} promised coverage of at least $1-\alpha_{\max}$ minus losses there. The \rpPubWeaDascCovPct{} per cent coverage observed in that run lies outside the model. The reason is known: the quantile came from the residuals of one forecaster while the interval was centred on another, so the centred residual is not in the family $F(\cdot\mid s)$, and the model does not apply. (b) The term $\varepsilon_t(\eta)$ is larger than the term $\sqrt{\log(2/\eta)/(2\nefft)}$ displayed in the earlier preprint. At $\nefft=\thEpsNeff$ and $\eta=\thEpsEta$ the two are \thEpsNew{} and \thEpsOld. The earlier term has no proof for weighted, non-identically distributed residuals. Write $\Fhat_{w,t}(r)=\sum_{i\in\cP^{g}_t}w_{i,t}\ind\{R_i\le r\}$ for the weighted empirical distribution function of the residuals in the gated pool, with the weights of the method. The proof in the preprint, a union bound over the random jump points of $\Fhat_{w,t}$, was not valid. The grid version is proved, at the price of a factor of \thEpsRatio{} at these values (the ratio of the unrounded terms). (c) Kernel weighting cannot increase the bias term relative to equal weights on the same gated pool: if $w_{i,t}\propto\phi(\norm{S_i-S_t})$ with $\phi\ge 0$ non-increasing, then $\bar d_t\le\abs{\cP^g_t}^{-1}\sum_{i\in\cP^g_t}\norm{S_i-S_t}$, by Chebyshev's sum inequality (Remark~5 of the sibling paper). In the bound, the sampling cost of localisation enters through $\nefft$: in $\varepsilon_t(\eta)$ with the default $K_t$, and, when the pool residuals are distinct, in $\pi_t\le\max_iw_{i,t}\le\bigl(\sum_iw_{i,t}^2\bigr)^{1/2}=\nefft^{-1/2}$.
\end{remark}

\subsection{Long-run calibration of the adaptive level}
\label{sec:longrun}

The statements in this subsection are deterministic. They hold for every sequence of residuals and concern the realised miss indicators $E_t$.

\begin{theorem}[Free update; \citealp{gibbs2021adaptive}, Lemma~4.1 and Proposition~4.1]
\label{thm:free}
Run $\alpha_{t+1}=\alpha_t+\gamma(\alpha-E_t)$ with $\alpha_1\in[0,1]$ and $\gamma>0$, and the conventions of Section~\ref{sec:conventions}. Then for every $T\ge 1$
\begin{equation}
\label{eq:free-identity}
\frac1T\sum_{t=1}^TE_t=\alpha+\frac{\alpha_1-\alpha_{T+1}}{\gamma T},
\end{equation}
$\alpha_t\in[-\gamma,1+\gamma]$ for all $t$, and therefore
\begin{equation}
\label{eq:free-bound}
\abs{\frac1T\sum_{t=1}^TE_t-\alpha}\le\frac{\max\{\alpha_1,1-\alpha_1\}+\gamma}{\gamma T}.
\end{equation}
\end{theorem}

Equation~\eqref{eq:free-identity} does not involve the weights, so equal weights, spectral weights and any other weights satisfy it. The sibling paper makes the same remark. The identity therefore gives no support for spectral weighting.

\begin{theorem}[Projected update: identity and sign structure]
\label{thm:projected}
Run $\alpha_{t+1}=\Pi_{[\alpha_{\min},\alpha_{\max}]}(U_t)$ with $U_t=\alpha_t+\gamma(\alpha-E_t)$, $\alpha_1\in[\alpha_{\min},\alpha_{\max}]$ and $0<\alpha_{\min}\le\alpha\le\alpha_{\max}<1$. Let $\epsilon_t=\alpha_{t+1}-U_t$ be the projection residual, $N^{+}_T=\#\{t\le T:U_t>\alpha_{\max}\}$ and $N^{-}_T=\#\{t\le T:U_t<\alpha_{\min}\}$. Then:
\begin{enumerate}[label=(\alph*)]
\item $\epsilon_t<0$ only if $E_t=0$, and then $\epsilon_t\ge-\gamma\alpha$; $\epsilon_t>0$ only if $E_t=1$, and then $\epsilon_t\le\gamma(1-\alpha)$;
\item for every $T\ge 1$,
\begin{equation}
\label{eq:proj-identity}
\frac1T\sum_{t=1}^TE_t-\alpha=\frac{\alpha_1-\alpha_{T+1}}{\gamma T}+\frac1{\gamma T}\sum_{t=1}^T\epsilon_t ;
\end{equation}
\item consequently
\begin{equation}
\label{eq:proj-bounds}
-\frac{\alpha_{\max}-\alpha_{\min}}{\gamma T}-\alpha\frac{N^{+}_T}{T}\ \le\ \frac1T\sum_{t=1}^TE_t-\alpha\ \le\ \frac{\alpha_{\max}-\alpha_{\min}}{\gamma T}+(1-\alpha)\frac{N^{-}_T}{T}.
\end{equation}
\end{enumerate}
If the floor never binds the long-run miss rate is at most $\alpha+(\alpha_{\max}-\alpha_{\min})/(\gamma T)$, and if the cap never binds it is at least $\alpha-(\alpha_{\max}-\alpha_{\min})/(\gamma T)$. No two-sided bound of order $1/T$ holds in general: for a constant stream with the lag-one forecaster every interval covers, $E_t\equiv 0$, the level reaches $\alpha_{\max}$ after finitely many steps and stays there, and the miss rate is $0$, at distance $\alpha$ from the target for every $T$.
\end{theorem}

\emph{Relation to earlier work.} Projecting the level adds boundary terms to the ACI identity. The identity \eqref{eq:proj-identity}, the sign structure in (a) and the bound on the boundary terms by $\gamma\alpha$ and $\gamma(1-\alpha)$ are in Proposition~3.1 and Appendix~B of \citet{lai2026online}, stated there for projection onto $[0,1]$, with a recommendation to report the average boundary terms. Theorem~\ref{thm:projected} restates them for $[\alpha_{\min},\alpha_{\max}]$ and adds the constant-stream counterexample to a two-sided $O(1/T)$ bound. The sibling paper states that projection adds boundary terms and reports how often the projection is active, without an identity or a bound \citep[Remark~4]{opoku2026spectral}. What this paper adds is the measured size of the terms on the data of Section~\ref{sec:design-data}, reported in Table~\ref{S-tab:projected-identity} and Section~\ref{sec:results-identity}.

\begin{remark}[What the projection costs]
\label{rem:projection}
By Theorem~\ref{thm:projected}, a binding cap can only lower the long-run miss rate, by at most the fraction of binding steps times $\alpha$, and a binding floor can only raise it, by at most the fraction of binding steps times $1-\alpha$. The bounds in \eqref{eq:proj-bounds} are worst cases: cap and floor contributions have opposite signs and can offset each other. The projection keeps the quantile level in a range where the weighted quantile of a finite pool is defined, and it avoids empty and infinite intervals. The free update restores \eqref{eq:free-bound}. Any statement about long-run calibration must either use the free update or report $N^{\pm}_T/T$. The averaged-recursion account of Section~\ref{sec:avg} is for the free update.
\end{remark}

The earlier preprint stated a long-run theorem for the projected update with a boundary term of order $1/T$. The theorem is false, by the counterexample above. The preprint's own weather run is a second instance. In that run the interval was centred on the lag-one forecast, with $T=\rpOldT$ and $\gamma=\dcGammaDefault$, and the cap was binding at $N^+_T/T=\rpOldCapShare$ per cent of the steps. The telescoping term was $\rpOldTele$ and the projection term $\rpOldProj$. They sum to the recorded deviation $\rpOldDev$ of the miss rate from $\alpha$, against the earlier bound of $\rpOldBound$. The \rpOldCapShare{} per cent figure belongs to that earlier configuration only. In the like-for-like run of this paper, projected ACI with the seasonal naive forecaster on the weather data has the cap binding at \rvWeaNaiveAciCap{} per cent of the steps (Table~\ref{S-tab:projected-identity}).

\begin{remark}[The two guarantees do not combine]
\label{rem:no-combine}
Theorem~\ref{thm:perstep} bounds the conditional miss probability $p_t=\Prob\{E_t=1\mid\cK_{t-1}\}$ by $\alpha_t$ plus losses. Theorems~\ref{thm:free} and~\ref{thm:projected} give the realised average of $E_t$ as the target level plus an error term, whatever the per-step probabilities are. The error term is explicit and of order $1/T$ for the free update; for the projected update it contains the boundary terms, and no two-sided bound of order $1/T$ holds in general. Proposition~\ref{prop:combine} in Appendix~\ref{app:proofs} is all that follows from both: on average the adaptive level sits at or above $\alpha$ minus the average per-step loss. The update compensates the losses by raising the level. It does not give coverage at level $1-\alpha$ at any step.
\end{remark}

\subsection{Drift score and gate}
\label{sec:gate}

\begin{lemma}[Drift score]
\label{lem:drift}
For every $t$, with $\bar d^{\mathrm{pre}}_t=\sum_{i\in\cP_t}w^{\mathrm{pre}}_{i,t}\norm{S_i-S_t}$,
\[
D_t\le\norm{\bar S^{\mathrm{rec}}_t-S_t}+\bar d^{\mathrm{pre}}_t .
\]
\end{lemma}

The earlier preprint stated $D_t\le\bar d_t$. That inequality is not correct in general: it needs $\bar S^{\mathrm{rec}}_t=S_t$, and the average distance has to use the pre-gate weights, because the drift score is defined through them. With the final weights it failed in \thDriftSyntheticOld{} of \thDriftSyntheticCases{} steps of a synthetic stream and in \thDriftRandomOld{} of \thDriftRandomCases{} random configurations, by up to \thDriftRandomExcess. Lemma~\ref{lem:drift} was never violated. The drift score $D_t$ is a monitoring statistic and enters no bound.

\begin{remark}[The gate is a heuristic]
\label{rem:gate}
Theorem~\ref{thm:perstep} holds for whatever gated pool is used, because the gate is $\cG_t$-measurable. Gating changes $\bar d_t$ and $\nefft$, and so the value of the bound, but nothing is proved about the effect of the gate on coverage or width. Likewise nothing is proved about the effect on coverage or width of the level relaxation that the code can apply, which is off in every reported run.
\end{remark}

\subsection{Checking the assumption from data}
\label{sec:check}

The earlier preprint proposed to estimate $L$ by regressing $\abs{\Fhat_{w,s}(q_s)-(1-\alpha)}$ on $\bar d_s$. At a fixed level that quantity is at most $\max_iw_{i,s}$ when the pool residuals are distinct. Under the adaptive rule it equals $\abs{\alpha_s-\alpha}$ up to the same slack; this is not proved here. It carries no information about $L$. What data can do is refute pairs $(L,\delta)$ under Assumption~\ref{ass:A}.

\begin{proposition}[A testable consequence of Assumption~\ref{ass:A}]
\label{prop:check}
Let $\cG=\cG_T$ and let $A,B\subseteq\{1,\ldots,T\}$ be finite, disjoint, $\cG$-measurable index sets, for instance clusters of the features. Put $F_A(r)=\abs A^{-1}\sum_{i\in A}F(r\mid S_i)$, $\Fhat_A(r)=\abs A^{-1}\sum_{i\in A}\ind\{R_i\le r\}$ and $D_{AB}=(\abs A\abs B)^{-1}\sum_{i\in A,j\in B}\norm{S_i-S_j}$. Under Assumption~\ref{ass:A}(ii), $\sup_r\abs{F_A(r)-F_B(r)}\le L\,D_{AB}+\delta$. Under Assumption~\ref{ass:A} as a whole, when $\abs A,\abs B\ge 2$ and with $\varepsilon_A(\eta)=\sqrt{\log(4\abs A/\eta)/(2\abs A)}+1/\abs A$,
\[
\Prob\Bigl\{\sup_r\abs{\Fhat_A(r)-\Fhat_B(r)}\le L\,D_{AB}+\delta+\varepsilon_A(\eta)+\varepsilon_B(\eta)\ \Bigm|\ \cG\Bigr\}\ge 1-2\eta .
\]
\end{proposition}

\begin{definition}[Assumption diagnostic]
\label{def:Lmin}
For a collection of index sets and a tolerance $\delta\ge 0$ let
\[
\widehat L_{\min}(\delta,\eta)=\max_{(A,B)}\frac{\bigl(\sup_r\abs{\Fhat_A(r)-\Fhat_B(r)}-\delta-\varepsilon_A(\eta)-\varepsilon_B(\eta)\bigr)^+}{D_{AB}}.
\]
\end{definition}

If Assumption~\ref{ass:A} holds with constants $(L,\delta)$, then $\widehat L_{\min}(\delta,\eta)\le L$ with probability at least $1-J\eta$, where $J$ is the number of index sets. Because $\widehat L_{\min}$ is a maximum over pairs of the $J$ sets, a confidence level of $0.95$ is reached with $\eta=0.05/J$ for each set, and Section~\ref{sec:results-assumption} reads the diagnostic at that level. The diagnostic is a lower confidence bound on every admissible Lipschitz constant. A large value refutes small $L$, that is, it refutes the pair $(L,\delta)$ under parts (i) and (ii) of the assumption together. It cannot say which part fails. A small value certifies nothing, because the inequality can be loose on the available pairs. Substituting $\widehat L_{\min}$ for $L$ in Theorem~\ref{thm:perstep} does not give a valid bound. No estimator of $L$ is given. The diagnostic refutes but cannot certify. A bootstrap variant studied by simulation is in Appendix~\ref{app:exploratory}; it has no proof. Consider the simulated model where the assumption holds with $L$ between \thAdLpathMin{} and \thAdLpathMax{} on the realised feature paths of \thAdSeedsWord{} replications. There $\widehat L_{\min}(0,\eta)$, read at $\eta=0.05/J$ for the $J$ clusters kept (\thAdJMin{} to \thAdJMax{} over the replications, so $\eta$ from \thAdEtaAdjMin{} to \thAdEtaAdjMax), lay at or below the constant in \thAdLeCount{} of \thAdSeeds{} replications, as a lower confidence bound should, with values up to \thAdLminMax{}. The regression slope of the earlier preprint also lay below the constant in every replication, with values up to \thAdSlopeMax{}, so it cannot serve as an estimate of $L$ and a bound computed from it would have been invalid. Neither quantity estimates $L$. A lower confidence bound on the smallest admissible $L$ is not in the localised conformal papers we opened (Section~\ref{sec:related}).

\subsection{Dependence without the independence idealisation}
\label{sec:mixing}

The concentration step of Theorem~\ref{thm:perstep} uses Assumption~\ref{ass:A}(i). The only route known to us that avoids it is a blocking argument under $\beta$-mixing, which has the form of the argument of \citet[Proposition~10 and Lemma~11]{oliveira2024split} for unweighted split conformal prediction. For a strictly stationary sequence $(\xi_i)_{i\in\mathbb Z}$ the $\beta$-mixing coefficient is $\beta(k)=\sup_n\E\bigl[\sup_{A\in\sigma(\xi_j:j\ge n+k)}\abs{\Prob(A\mid\sigma(\xi_j:j\le n))-\Prob(A)}\bigr]$ \citep[Definition~2]{mohri2010stability}.

\begin{lemma}[Block coupling]
\label{lem:yu}
Let $h$ be a measurable function with $\abs h\le 1$ of $\mu$ blocks of consecutive indices separated by gaps of length at least $a$, let $Q$ be the joint law of the blocks and $P$ the product of their marginals. Then $\abs{\E_Qh-\E_Ph}\le 2(\mu-1)\beta(a)$.
\end{lemma}

This is Proposition~10 of \citet{oliveira2024split} with $M=1$, applied to the odd blocks, with the even blocks as gaps. The coupling goes back to \citet{yu1994rates} and is stated in \citet[Lemma~3]{mohri2010stability} with the constant $(\mu-1)$ in place of $2(\mu-1)$. Yu's original paper has not been checked, so the larger constant of Oliveira et al. is used.

\begin{proposition}[Blocked concentration for fixed weights]
\label{prop:blocked}
Let $\xi_1,\ldots,\xi_n$ be a stretch of a strictly stationary $\beta$-mixing sequence with $\abs{\xi_i}\le 1$ and $\E\xi_i=0$, and let $w_1,\ldots,w_n\ge 0$ be fixed weights that sum to one. Split $\{1,\ldots,n\}$ into consecutive blocks $B_1,\ldots,B_{2\mu}$ of length $a$ (the last may be shorter), and put $W_j=\sum_{i\in B_j}w_i$. For every $u>0$,
\[
\Prob\Bigl\{\Bigl|\sum_{i=1}^nw_i\xi_i\Bigr|\ge u\Bigr\}\le 2\exp\Bigl(-\frac{u^2}{8\sum_{j\ \mathrm{odd}}W_j^2}\Bigr)+2\exp\Bigl(-\frac{u^2}{8\sum_{j\ \mathrm{even}}W_j^2}\Bigr)+4(\mu-1)\beta(a).
\]
\end{proposition}

Proposition~\ref{prop:blocked} is the standard blocking argument applied to fixed weights with block masses $W_j$. It is correct, and it is vacuous at the pool sizes used here. Take equal weights, $\beta(k)=\rho^k$ for $\rho=e^{-1/\thMixTau}$, and confidence $\eta=\thEpsEta$. Optimised over the block length, with both exponential terms of the bound replaced by the larger of the two, the smallest deviation $u$ that the right-hand side certifies is \mxSmallFourU{} for $n=\mxSmallN$ residuals, \mxMidFourU{} for $n=\mxMidN$ and \mxLargeFourU{} for $n=\mxLargeN$. With the kernel weights of the method on a pool of $\mxSmallN$ it lies between \mxKernFourMin{} and \mxKernFourMax{}. Every value exceeds $\alpha=0.1$, so the bound certifies nothing useful at these sizes. The values with the constant $2$ in place of $4$ are \mxSmallTwoU, \mxMidTwoU{} and \mxLargeTwoU{}; the larger constant makes the bound weaker, so the conclusion holds for either. Keeping the two exponential terms apart lowers every value and leaves each above $\alpha$. These values are the authors' own arithmetic and not a published finding. The independent-case term $2\exp(-2u^2\neff)$ is not valid under dependence of this kind. For a stationary $\thMixM$-dependent sequence with $n=\mxSmallN$ and equal weights, the deviation $0.10$ has empirical probability \thMixEmp{} against the independent-case value \thMixIid{} (a factor of \thMixRatio, from \thMixReps{} replications).

\begin{remark}[Limits of Proposition~\ref{prop:blocked}]
\label{rem:mixing-gap}
The proposition needs fixed weights and centres the indicators at the stationary marginal law. In the method the weights are functions of $S_t$ and of the $S_i$, and the natural centring is the regime law $F(\cdot\mid S_i)$. To use the proposition one would condition on the feature path and assume the residuals are $\beta$-mixing given it. That assumption is no more natural than Assumption~\ref{ass:A}(i), so (i) stays as the stated idealisation. The earlier preprint used blocks of size one separated by gaps and kept the exponential term at full strength. Its proof did not decouple the gap blocks, and its accompanying arithmetic was wrong: with $\rho=e^{-1/\thMixTau}$, $n=\mxSmallN$ and gap $a=\thOldGap$, the penalty is $4(\mu-1)\rho^{a}$ with $\mu-1=\lfloor n/(a+1)\rfloor$ block pairs, which equals \mxOldValue{} and is not below $\thOldClaim$; falling below it needs $\rho\le\mxOldRho$.
\end{remark}

\input{s4_averaged}

\subsection{Summary of what is proved}
\label{sec:theory-summary}

Table~\ref{tab:theory-status} gives, for each statement of this section, its status, the numerical check that exists for it, and its scope.

\begin{table}[t]
\centering
\caption{Status of each statement. ``Proved'' means a written proof in Appendix~\ref{app:proofs}. ``Numerical check'' refers to Appendix~\ref{app:checks}; it is evidence and never part of a proof. The statements of Section~\ref{sec:avg} carry the labels exact, ODE level and numerical defined there, and their proofs are in Appendix~\ref{S-app:avg-proofs}.}
\label{tab:theory-status}
\small
\begin{tabular}{p{3.6cm}p{2.2cm}p{4.6cm}p{4.2cm}}
\toprule
Statement & Status & Numerical check & Scope and caveat \\
\midrule
Lemma~\ref{lem:measurable} (causality) & proved & perturbation test of the code: \caChecks{} checks, \caFail{} failures & none \\
Lemma~\ref{lem:bias} (bias) & proved & none needed & Assumption~\ref{ass:A}(ii) \\
Lemma~\ref{lem:grid} (concentration) & proved & largest exceedance frequency \thGridMaxExceed{} over \thGridConfigsWord{} configurations & independence given the features; fixed weights \\
Theorem~\ref{thm:perstep} (per-step) & proved & lower bound: no failure in the models checked; upper bound: uninformative, the bound is at least $1$ in \thPsHiOneMin{} to \thPsHiOneMax{} per cent of steps & Assumption~\ref{ass:A}; level $1-\alpha_t$; loose (mean slack \thPsSlackMin{} to \thPsSlackMax{}) \\
Theorem~\ref{thm:free} (free update) & proved; Gibbs and Cand\`es & identity holds to rounding on real paths & conventions of Section~\ref{sec:conventions}; says nothing per step \\
Theorem~\ref{thm:projected} (projected update) & proved; restatement of Lai and Raskutti & identity holds to rounding in every run & two-sided $O(1/T)$ fails in general \\
Proposition~\ref{prop:combine} (combination) & proved & none: the inequality cannot fail in the models used & Assumption~\ref{ass:A}; in expectation; weak \\
Lemma~\ref{lem:drift} (drift score) & proved & no violation in \thDriftSyntheticCases{} and \thDriftRandomCases{} cases & none \\
Gate, level relaxation & heuristic & none & nothing proved \\
Proposition~\ref{prop:check}, Definition~\ref{def:Lmin} & proved & $\widehat L_{\min}\le L$ at $\eta=0.05/J$ in \thAdLeCount{} of \thAdSeeds{} replications & refutes only; cannot separate parts (i) and (ii) \\
Lemma~\ref{lem:yu} (block coupling) & cited & none & Oliveira et al.; Yu (1994) not checked \\
Proposition~\ref{prop:blocked} (blocked) & proved & the bound equals $1$ in \thMixAtOne{} of \thMixCells{} cells of the check & correct, vacuous at the pool sizes used \\
Proposition~\ref{prop:avg-block} (block identity, free update) & exact; Gibbs and Cand\`es & stored runs, \avNCells{} cells &
Proposition~4.1 of Gibbs and Cand\`es (2021) over blocks; free update; idealised pool \\
Proposition~\ref{prop:avg-ode} (averaged recursion) & ODE level & bound held in all \avNCells{} cells & idealised pool; chain gap proved only as $O(\sqrt\gamma)$ \\
Proposition~\ref{prop:avg-match} (weighted quantile at a fixed level) & exact; ODE level for the matching step & \avNCells{} cells and the rows of Table~\ref{tab:v2-avg-match} & closed forms and the bound $\abs c$ exact; uniqueness of the matching step is numerical \\
\bottomrule
\end{tabular}
\end{table}

%% file: s4_averaged.tex
\subsection{Why a faster step can stand in for localisation: an averaged-recursion account}
\label{sec:avg}

In the simulated design built for spectral weighting (Section~\ref{sec:results-sim}), the kernel weights split coverage between the two regimes more evenly than ACI at an equal slow step, and ACI with a faster step does as well. This subsection accounts for that finding for the free update in an idealised model. The identities hold for the free update with the conventions of Section~\ref{sec:conventions} at levels $\alpha_t\le 0$ and $\alpha_t\ge 1$; the projected update, which the experiments use, adds boundary terms (Theorem~\ref{thm:projected}). Each statement carries a label: \emph{exact} means proved for the stochastic chain, \emph{ODE level} means proved for the averaged recursion only, and \emph{numerical} means checked on stored runs and not proved. The proofs are in Appendix~\ref{S-app:avg-proofs}.

\emph{Setting.} Residuals alternate between two regimes in blocks of $m$ steps. In regime one the residual $\abs{R_t}$ is the absolute value of an $N(0,1)$ variable, in regime two of an $N(0,\varsigma^2)$ variable with $\varsigma>1$, independent over $t$. The calibration pool is infinite and holds the two laws in equal shares. With $S_r(q)=\Prob(\abs{R}>q\mid\text{regime }r)$, the quantile $q(a)$ at level $a$ solves $\tfrac12\{S_1(q)+S_2(q)\}=a$, with $q(a)=+\infty$ for $a\le 0$ and $q(a)=0$ for $a\ge 1$. The free update $\alpha_{t+1}=\alpha_t+\gamma(\alpha-E_t)$, $\alpha_1\in[0,1]$, runs on this stream, and $x=\gamma m$. Put $p_r(a)=S_r(q(a))$, the miss probability in regime $r$ at level $a$, and let $a^*_r$ solve $p_r(a^*_r)=\alpha$. Then $a^*_2<\alpha<a^*_1$. The averaged recursion is the differential equation $\dot a=x\{\alpha-p_{r(\tau)}(a)\}$ in the time $\tau=t/m$, with the same regime switching.

\begin{propavg}[Block identity; exact, free update]
\label{prop:avg-block}
For every block of $m$ consecutive steps, almost surely, $\sum_{t\in\text{block}}(E_t-\alpha)=(\alpha_{\mathrm{start}}-\alpha_{\mathrm{end}})/\gamma$, where $\alpha_{\mathrm{start}}$ and $\alpha_{\mathrm{end}}$ are the levels used at the first step of the block and of the next. Over the first $J$ periods (one block of each regime per period), the deviations $d_1(J)$ and $d_2(J)$ of the two regime miss rates from $\alpha$ satisfy $d_1(J)+d_2(J)=(\alpha_1-\alpha_{2Jm+1})/(xJ)$, which has modulus at most $(1+\gamma)/(xJ)$. Under any invariant law of the level at the start of a period, the expected miss rates in the two regimes are $\alpha-D/x$ and $\alpha+D/x$ for one number $D$: the two deviations are exact negatives of each other.
\end{propavg}

Summing the update over a block gives $(\alpha_{\mathrm{start}}-\alpha_{\mathrm{end}})/(\gamma m)$, as in the proof of Proposition~4.1 of \citet{gibbs2021adaptive}, and the sum over $2mJ$ steps is that proposition. The block identity and the pathwise antisymmetry are therefore not new. The expectation form under an invariant law is a one-line consequence and is the only addition. Projection adds boundary terms \citep[Proposition~3.1]{lai2026online}. Interval-wise coverage bounds are known \citep[Theorem~3.1]{gibbs2022conformal}, and so are bounds for any interval \citep[Theorem~C.3]{bhatnagar2023improved}. The invariant law need not be unique: when $\alpha=1/n$ and $n$ divides $2m$, each residue class of the lattice on which the level lives carries its own invariant law.

\begin{propavg}[Averaged recursion; ODE level]
\label{prop:avg-ode}
For every $x>0$ and $\varsigma>1$ the averaged recursion has exactly one periodic solution, and every solution approaches it. Its level $A_*$ at the start of a regime-one block and $B_*$ at the end satisfy $a^*_2<A_*<B_*<a^*_1$ and $A_*\le\alpha\le B_*$. The regime miss rates of the averaged recursion are $\mu_1=\alpha-D/x$ and $\mu_2=\alpha+D/x$ with $D=B_*-A_*$, and
\[
(a^*_1-a^*_2)\bigl\{1-e^{-xs_1(a^*_2)}-e^{-xs_2(a^*_1)}\bigr\}\ \le\ D\ <\ a^*_1-a^*_2,\qquad s_r=p_r'.
\]
Hence $\abs{\mu_r-\alpha}<(a^*_1-a^*_2)/x$ for every $x$, and $x\abs{\mu_r-\alpha}\to a^*_1-a^*_2$ as $x\to\infty$, so the constant is sharp. The miss rates depend on $\gamma$ and $m$ only through $x$.
\end{propavg}

ACI's level is a constant-step stochastic approximation with known stationary analyses for exchangeable scores \citep[Theorems~3.1 and~3.2]{zaffran2022adaptive}. Proposition~\ref{prop:avg-ode} studies its averaged form. We did not find the periodic orbit of the averaged recursion, or its bound with a sharp lower bound, in the sources we opened, and the search was not exhaustive. The bound concerns the differential equation, so it is not comparable to the chain-level bounds of \citet{gibbs2021adaptive}, which concern the state-conditional bias of the chain. An empirical collapse of coverage onto a delay-to-memory ratio has been reported for delayed feedback \citep{elhalabi2026adaptive}.

For the stochastic chain the regime miss rates are within $O(\sqrt\gamma)$ of those of the averaged recursion at fixed $x$ (exact; Proposition D, Appendix~\ref{S-app:avg-proofs}). The constants are not uniform in $x$, the bound is capped through the mass outside a ball around the orbit, and for small $\gamma$ it is about \avChainConst{} times $\sqrt\gamma$ in the cell checked, which is of no use at the step sizes of interest. Numerically the gap grows like $\gamma$, with log-log slopes of \avGapSlopeMin{} to \avGapSlopeMax{} (Appendix~\ref{S-app:avg-details}). Table~\ref{tab:v2-avg-match} gives the equilibria.

\begin{propavg}[Weighted quantile at a fixed level; exact and ODE level]
\label{prop:avg-match}
Let the pool at the fixed level $\alpha$ put mass $\omega\in(\tfrac12,1)$ on the current regime and $1-\omega$ on the other. The regime miss rates of the weighted quantile are $\alpha+c-d_W$ and $\alpha+c+d_W$ with $d_W>0$, in closed form (Appendix~\ref{S-app:avg-proofs}). The shift $c$ is negative when $a^*_1\le a_m$, the level at which $s_1=1$; this holds at the ratios of Table~\ref{tab:v2-avg-match}. (i) Exact: for every $\gamma$ and $m$ the expected regime miss rates of ACI have the form $(\alpha-d,\alpha+d)$ by Proposition~\ref{prop:avg-block}. The sup-norm distance between the two pairs of miss rates is therefore at least $\abs c$ for every step size, with equality if and only if $d=d_W$. No step size removes the common-mode shift $c$. (ii) ODE level: there is a scale $x^*(\omega,\varsigma)$ at which the swing of ACI's averaged recursion equals the weighted quantile's, so that the error of the averaged recursion is $\abs c$, and $x^*(1-\omega)\to1$ as $\omega\to1$. The matching step for blocks of length $m$ is $\gamma^*(m)=x^*/m$, and it must not exceed $\tfrac12$. For the chain, attainment at $x^*$ is not proved; Proposition D adds its gap. Uniqueness of $x^*$ is numerical.
\end{propavg}

\input{v2_avg_matching}

Table~\ref{tab:v2-avg-match} gives $x^*$ and $\abs c$ for two values of $\omega$. The weighted quantile at a fixed level has a regime shift that no ACI step removes: ACI at the matching step reproduces the regime gap of the weighted quantile, and the shift remains. The sign condition for $c$ is sufficient and cannot be dropped for a general target level: for $\varsigma=\avRhoC$ and $\omega=\avWA$, $c$ changes sign at $\alpha\approx\avCZeroAlphaThree$ (at $\alpha\approx\avCZeroAlphaTwo$ for $\varsigma=\avRhoB$). ACI and a sliding-window quantile have been compared by simulation on square-wave streams \citep[Section~6.2]{vaze2026simultaneous}, and \citet{oancea2026dynamic} states in qualitative terms, at the opening of its Section~4, that the weights shrink the gap while the controller regulates only the cumulative rate. The common-mode shift $c$, the swing match and the scale $x^*$ are not in the sources we opened.

\emph{Link to the recurring frequency-and-noise design.} In that design spectral ACI at the steps \avOwnStepA{} and \avOwnStepB{} has the regime coverage that ACI reaches at the steps \avMatchStepA{} and \avMatchStepB{}, with sup-norm errors of \avMatchErrA{} and \avMatchErrB{}, so in this design the kernel weights act like an ACI step \avRatioB{} to \avRatioA{} times faster. This is long-run regime coverage only; it says nothing about width, interval score or coverage just after a switch. Appendix~\ref{S-app:avg-details} gives the residual scales and the checks of the recursion against simulation.

\emph{Limits.} The model has an infinite pool, independent Gaussian residuals and the free update. The experiments use a rolling pool, for which the idealised pool fails when a block is longer than the window, and the projected update, which is not modelled. The extension of Proposition~\ref{prop:avg-ode} to the rolling window is ODE level for the window-aware mean-field recursion, argued and not written out. The explicit formula for the linearised recursion (Appendix~\ref{S-app:avg-proofs}) has an error of order $x^2$ at fixed $\varsigma$, but not uniformly in $\varsigma$. Uniqueness of $x^*$ is not proved, and the invariant law is not unique on the lattice described above. Dependent or non-Gaussian residuals and empirical quantiles are not covered, and one simulated design supports the link. We know of no source that equates a faster step with similarity weighting or gives a step ratio; the search was not exhaustive.

%% file: v2_avg_matching.tex
\begin{table}[t]
\centering
\caption{Averaged-recursion account: the levels $a^{*}_{1}$ and $a^{*}_{2}$ at which each regime alone would have miss rate $\alpha$, and the matching scale $x^{*}=\gamma^{*}m$ and common-mode shift $|c|$ of the weighted quantile at a fixed level for two values of the same-regime mass $\omega$. The scale ratio is $\varsigma$ and the target is $\alpha=\avAlpha$. The matching scale is the ODE solution (Proposition~\ref{prop:avg-match}); the shift $|c|$ is a lower bound on the error of every ACI step.}
\label{tab:v2-avg-match}
\small
\setlength{\tabcolsep}{5pt}
\begin{tabular}{rrrrrrr}
\toprule
 & & & \multicolumn{2}{c}{$\omega=\avWA$} & \multicolumn{2}{c}{$\omega=\avWB$} \\
\cmidrule(lr){4-5}\cmidrule(lr){6-7}
$\varsigma$ & $a^{*}_{1}$ & $a^{*}_{2}$ & $x^{*}$ & $|c|$ & $x^{*}$ & $|c|$ \\
\midrule
1.5 & 0.186 & 0.057 & 3.8 & 0.006 & 10.0 & 0.004 \\
2 & 0.255 & 0.051 & 4.0 & 0.016 & 10.1 & 0.009 \\
3 & 0.342 & 0.050 & 4.5 & 0.029 & 10.1 & 0.018 \\
\bottomrule
\end{tabular}
\end{table}

%% file: s5_design.tex
\section{Simulation design and evaluation protocol}
\label{sec:design}

The study is organised by the ADEMP structure for simulation studies \citep{morris2019using}: aims, data-generating mechanisms, estimands, methods and performance measures. The evaluation on real data uses the same methods, tuning and measures, and serves as a check on the simulations.

\subsection{Aims}
\label{sec:design-aims}

\begin{enumerate}[label=A\arabic*.]
\item Compare DASC with adaptive conformal inference (ACI), spectral ACI and decay-weighted conformal prediction when all methods use one point forecaster per data set, one tuning criterion and the same evaluation steps.
\item Measure how much of the earlier reported width gain is due to a difference in the centring forecast, and how large other common evaluation pitfalls are (Section~\ref{sec:pitfalls}).
\item Separate what the kernel weights, the gate, the drift score and the level update each contribute.
\item Check whether the identities of Section~\ref{sec:longrun} hold in the runs, and whether the smoothness assumption is refuted on the data.
\item Measure the computational cost per step.
\end{enumerate}

The authors of this paper proposed DASC and spectral ACI. The study is a re-evaluation by the same authors and is not a neutral comparison in the sense of an independent benchmark. The design limits the room for choices that favour these methods in three ways: every grid contains the same pool sizes; the tuning criterion and the development data are the same for every method; and the code, the saved results and a causality test (Section~\ref{sec:design-causality}) are public.

\subsection{Data-generating mechanisms}
\label{sec:design-data}

\emph{Simulated streams.} Table~\ref{tab:dgm} in Appendix~\ref{app:tables} lists \smNScenWord{} scenarios. Each is run with \smNReps{} evaluation replications (seeds \dcEvalSeeds) and tuned on \smNTuneSeedsWord{} separate development seeds (\dcTuneSeeds). The scenario in which the frequency and the noise scale both alternate between two recurring regimes is the one built for spectral weighting: a regime is identifiable from the spectrum of the recent window, and the residual scale differs between regimes. The scenario with recurrence in amplitude only is built so that the normalised features cannot separate the regimes. The stress scenarios are those of the earlier preprint. The original three-phase stream keeps the code's length and break points.

\emph{Real streams.} Table~\ref{tab:data} in Appendix~\ref{app:tables} describes the three real data sets. Electricity is the hourly mean of household global active power from the UCI household power consumption data, using the first \dcElecLength{} non-missing hourly rows; two gaps of missing hours lie inside that range, so the seasonal forecast is misaligned for \dcLagSeasonal{} rows after the longer gap. Weather is the hourly two-metre temperature in Dallas from the Open-Meteo archive, using the first \dcWeaLength{} hours. Finance is the absolute daily log return of the \dcSPName{} index in percentage points from the FRED series, with \dcFinLength{} values after differencing. Each stream is one path. The three series are not a sample of application domains, and the bootstrap intervals below describe the sampling variation of one path.

\subsection{Estimands}
\label{sec:design-estimands}

The target is marginal coverage at the nominal level $1-\alpha=\dcAlphaPct$ per cent over the evaluation steps, reached with the narrowest intervals. Conditional coverage given the features or the regime is reported as a diagnostic and is not a target, because distribution-free conditional coverage is impossible without structure \citep{barber2021limits}. The comparisons of interest are paired differences between two methods in coverage, mean width and mean interval score, computed on the same steps.

\subsection{Point forecasters}
\label{sec:design-forecasters}

Each data set has three point forecasters, shared by all \rvNMethodsWord{} methods. Every forecast $\widehat f_t$ depends on $Y_1,\ldots,Y_{t-1}$ and on calendar information for time $t$.
\begin{itemize}
\item \emph{Naive.} $\widehat f_t=Y_{t-\dcLagSeasonal}$ for the hourly data and $\widehat f_t=Y_{t-1}$ for finance and for the simulations.
\item \emph{Rolling ridge.} Ridge regression (penalty \dcRidgeAlpha) on standardised lagged values and calendar terms. Hourly: lags \dcHourlyLags{} and sine and cosine terms for hour of day and day of week; first fit at step \dcHourlyFirstFit, refit every \dcHourlyRefit{} steps on the most recent \dcHourlyTrain{} rows with targets before the refit time. Daily: lags \dcDailyLags{} and the means of the past \dcDailyMeans{} days; first fit at step \dcDailyFirstFit, refit every \dcDailyRefit{} steps on the most recent \dcDailyTrain{} rows. Simulations: lags $1$ to $6$, first fit at step \dcSimRidgeFirst, refit every \dcSimRidgeRefit{} steps on the most recent \dcSimRidgeTrain{} rows.
\item \emph{Rolling gradient boosting.} The same inputs and refit schedule with gradient-boosted trees (\dcXgbTrees{} trees, depth \dcXgbDepth, learning rate \dcXgbRate, subsampling fraction \dcXgbSub), fitted with one thread. Not used in the simulations.
\end{itemize}
Table~\ref{S-tab:forecasters} in Appendix~\ref{S-app:tables} gives the accuracy of the forecasters on the evaluation steps. The naive forecaster is deliberately weak, so that the conformal layer has structure to correct. The strong forecasters have a lower error than the naive one on every data set, by a margin that differs between data sets, and so leave less for the conformal layer to correct.

\subsection{Tuning protocol}
\label{sec:design-tuning}

Every method is tuned with the same criterion, the mean interval score on the development data, and then frozen. On the real streams the development data are the first \dcDevShare{} per cent of the stream, and every configuration is run once over the whole stream with its adaptive state carried across the boundary; its development score is computed on the prefix and its evaluation metrics on the remainder. In the simulations the development data are \smNTuneSeedsWord{} separate seeds and the evaluation seeds are disjoint. Selection uses the development score only. For the real streams the grid files list every configuration with its development score and its evaluation metrics, so a reader can verify the selection. For the simulations the grid files list the development scores only, because the evaluation seeds are disjoint; the evaluation results of the selected configurations are in \path{sim_summary.csv}.

The grids are as follows. The pool size has the same candidates in every method of a data set: \dcElecPools{} for electricity, \dcWeaPools{} for weather, \dcFinPools{} for finance and \dcSimPools{} for the simulations. Feature windows are \dcHourlyWins{} for the hourly data and \dcDailyWins{} for finance and the simulations. The step size $\gamma$ for ACI takes the values \dcGammaGrid. The bandwidth $h$ takes \dcHGrid, the drift threshold $\lambda$ takes \dcLambdaGrid, and $m_{\min}$ is \dcMminGrid{} times the pool size. The decay factor $\theta$ takes \dcDecayGrid. The PID-style gains are $k_p\in\{\dcPidKp\}$, $k_i\in\{\dcPidKi\}$, $k_d\in\{\dcPidKd\}$. For spectral ACI the step size takes \dcSpecAciGamma. The step-size sets differ between methods: \dcSpecAciGamma{} for spectral ACI against \dcGammaGrid{} for ACI and for tuned DASC. Spectral ACI selected the largest value of its set in \exEdgeOrigTopSpecAciWord{} of the \exNCellsWord{} real cells, ACI in \exEdgeOrigTopAciWord{} and DASC with a tuned step size (the equal-budget extension below) in \exEdgeOrigTopDascWord{}. Section~\ref{sec:robust} repeats the comparison with one set of \exNGammaWord{} values for every method.

In the main runs the step size of DASC is fixed at $\dcGammaDefault$, the default of the earlier preprint, while ACI's step size is tuned over \dcGammaGridWord{} values. This asymmetry is the reason for the equal-budget extension below. Table~\ref{tab:methods} lists the number of configurations per method. In total \cfgTotal{} configurations are tried per forecaster, of which DASC receives \cfgDasc{} in each of its three variants.

\emph{Equal-budget extension.} The extension has two stages, run for each data set and forecaster, and for each simulated scenario and forecaster on the development seeds. In stage one the DASC configuration chosen in the main run is re-run with $\gamma$ in the ACI grid (\tgStageOne{} configurations). In stage two, at the chosen $\gamma$, the bandwidth and the threshold are re-tuned over $h\in\{\dcHGridShort\}$ and $\lambda\in\{\dcLambdaGridShort\}$ (\tgStageTwo{} configurations, one of them repeating the stage-one winner). DASC projected and DASC free then receive \tgBudgetDasc{} configurations each, against \cfgAci{} for ACI per variant, \cfgSpecAci{} for spectral ACI and \cfgDecay{} for decay-weighted conformal prediction. The baselines keep their main-run configurations. They are re-run in the same code, and \tgRerunMatch{} of \tgRerunTotal{} re-runs reproduce the main-run coverage and width to \dcRerunTol.

\subsection{Evaluation steps and performance measures}
\label{sec:design-measures}

All methods are run from the same first step, and all measures are computed over the same evaluation steps. The measures are:
\begin{itemize}
\item \emph{Marginal coverage} and \emph{mean width}. Width is never reported without coverage.
\item The \emph{interval score} at level $\alpha$ \citep{gneiting2007strictly}, $S_\alpha(l,u;y)=(u-l)+\tfrac2\alpha(l-y)^++\tfrac2\alpha(y-u)^+$, averaged over steps; lower is better. It is a proper score for the central $1-\alpha$ interval and rewards narrow intervals that cover.
\item \emph{Worst-block coverage}: the minimum coverage over non-overlapping blocks of \dcWorstElec{} hours (hourly data), \dcWorstFin{} trading days (finance) or \dcSimBlock{} steps (simulations).
\item \emph{Local coverage} by hour of day and month (hourly data), by causal volatility tercile (cut points fixed on the development prefix, volatility measured over the past \dcFinVolWin{} days) and by calendar year (finance), and by regime (simulations).
\item The \emph{diagnostics} of Section~\ref{sec:dasc}: $\nefft$, $\bar d_t$, $D_t$, the pool size after gating, the level $\alpha_t$, the fractions of steps at which the projection binds, and the projection residual.
\item \emph{Paired differences} between two methods. In the simulations, uncertainty is measured by the standard error over replications. On the real streams it is measured by a paired moving-block bootstrap: the per-step differences are resampled in blocks of \bootBlockElec{} hours (electricity and weather) or \bootBlockFin{} trading days (finance), the same blocks for both methods, with \bootN{} resamples and percentile \dcBootLevel{} per cent intervals. A width difference is only interpretable next to the coverage difference, so both are shown. A check outside the main design repeated the equal-budget study and the extended designs with two longer block lengths and \exSeedNSeedsWord{} bootstrap seeds (Section~\ref{sec:robust}). The width intervals against ACI keep their pattern in every draw. In some cells the interval-score interval against ACI lies below zero in every seed at \exBlkMidElec{} hours. Against spectral ACI, no draw favours DASC when every method searches the same step sizes, and in the equal-budget study the width interval of weather with gradient boosting against spectral ACI lies below zero in every seed at the main-run block length and in \exBlkSpecWeaBoostWMidWord{} of the \exSeedNSeedsWord{} seeds at \exBlkMidElec{} hours. The draws at \exBlkLongElec{} hours are a sensitivity check only (Section~\ref{sec:robust}). The check did not cover the main runs or the free update.
\end{itemize}
In the simulations the evaluation covers every step from the start of the evaluation window to the end of the stream (\dcSimSteps{} steps; \dcStressSteps{} for the five stress scenarios), and the Monte Carlo standard error is the standard deviation of the per-replication values divided by $\sqrt{\smNReps}$.

\subsection{Causality test}
\label{sec:design-causality}

The test is run for each method and forecaster, at test times spread over the evaluated range and at refit boundaries. At each test time $t$ the interval is computed twice: on the stream, and on a stream in which every value from $t$ onward, including $Y_t$, is replaced by unrelated values, with forecasts and features recomputed. The test passes only if the lower and upper limits, the forecast and the feature vector at $t$ are bit-for-bit identical in the two runs. It covers \caStreamsWord{} stream types (a simulated stream, an hourly stream with calendar terms and refits, and a non-negative daily stream with the clipping rule), \caMethodsWord{} methods and the forecasters of each stream, in total \caChecks{} checks. A deliberately leaky forecaster, $\widehat f_t=Y_t$, serves as a negative control and must fail.

\subsection{Reproduction anchor}
\label{sec:design-anchor}

Before any new result is reported, the new code is run under the settings of the earlier preprint, and its coverage and width are compared with the published numbers. Two differences between the original code and the new framework are reproduced for this purpose and then removed: the original computed raw, un-normalised features for windows of \dcLegacyWins{} points at the start of each stream, and it kept the undefined residuals of the seasonal forecaster in the pool for the first $m_{\max}$ steps. Section~\ref{sec:results-anchor} gives the comparison, and then the same runs with DASC centred on the forecaster that the baselines use.

%% file: s6_results.tex
\section{Results}
\label{sec:results}

\subsection{Reproduction anchor and the centring error}
\label{sec:results-anchor}

The new framework reproduces the numbers published in the earlier preprint when it is run under the original settings (Table~\ref{tab:reproduction}). For DASC on electricity, coverage and width are \rpRepElecDascCovF{} and \rpRepElecDascWidF, against the published \rpPubElecDascCov{} and \rpPubElecDascWid. On weather they are \rpRepWeaDascCovF{} and \rpRepWeaDascWidF{} against \rpPubWeaDascCov{} and \rpPubWeaDascWid, and on finance \rpRepFinDascCovF{} and \rpRepFinDascWidF{} against \rpPubFinDascCov{} and \rpPubFinDascWid. The remaining width differences are at most \rpHourlyMaxWidDiff; they arise because the original code kept undefined residuals in the pool during the first steps. The ACI rows differ by one order statistic through the quantile convention.

{\AtBeginEnvironment{tabular}{\scriptsize\setlength{\tabcolsep}{3pt}}\input{tab_reproduction_paper}}

The published comparison differed in more than the method, and the reported gains came from the centring difference. In the hourly experiments of the earlier preprint the DASC interval was centred on the previous observation ($\widehat f_t=Y_{t-1}$), while its quantile came from the residuals of the lag-\dcLagSeasonal{} forecaster. Rolling conformal, ACI, PID and the other baselines were centred on the lag-\dcLagSeasonal{} forecast. The published widths of DASC, \rpPubElecDascWid{} on electricity and \rpPubWeaDascWid{} on weather, are \rpGainElec{} and \rpGainWea{} per cent below the published widths of ACI, \rpPubElecAciWid{} and \rpPubWeaAciWid. These are the reported gains. Section~\ref{sec:pitfalls} measures the effect of centring next to that of other evaluation choices. In the weather run the adaptive level was pinned at its cap in \rpOldCapShare{} per cent of the steps, and the coverage was \rpPubWeaDascCovPct{} per cent at a nominal $\dcAlphaPct$ per cent. With DASC centred on the forecast that the baselines use, and nothing else changed, coverage and width become \rpRecElecDascCov{} and \rpRecElecDascWid{} on electricity against \rpRepElecAciCov{} and \rpRepElecAciWid{} for ACI (a relative width difference of \rpRecElecRel{} per cent). On weather they become \rpRecWeaDascCov{} and \rpRecWeaDascWid{} against \rpRepWeaAciCov{} and \rpRepWeaAciWid{} for ACI (\rpRecWeaRel{} per cent). Every result below uses one forecaster for all methods.

\subsection{Simulations}
\label{sec:results-sim}

At the selected step sizes DASC (projected) is wider than ACI (projected) in every cell, by \smWidMin{} to \smWidMax{} per cent, with standard errors of \smSeWidMin{} to \smSeWidMax{} per cent. Table~\ref{tab:sim-headline} gives the paired difference in each of the \smNScenWord{} scenarios with the naive forecaster, and Table~\ref{tab:sim-headline-ridge} in Appendix~\ref{app:tables} gives it for the ridge forecaster; the figures in the text cover both forecasters. The coverage differences lie between \smCovDiffMin{} and \smCovDiffMax. The interval-score difference is positive in all \smNCellsWord{} cells and exceeds two standard errors in \smIsWorseTwoWord{}. In the other \smIsWithinOneWord{} cells (the original three-phase stream and the abrupt shift, each with the ridge forecaster) it is within one standard error of zero. Coverage of DASC is between \smDascCovMin{} and \smDascCovMax{} in every cell, and the six adaptive methods lie between \smAdaptCovMin{} and \smAdaptCovMax. Coverage does not separate the methods.

{\AtBeginEnvironment{tabular}{\scriptsize\setlength{\tabcolsep}{3pt}}\input{tab_sim_headline}}

No variant of DASC has the lowest mean interval score in any cell. Among all \rvNMethodsWord{} methods, decay-weighted conformal prediction has the lowest in \smBestDecay{} of the \smNCellsWord{} cells, the spectral method with a fixed level in \smBestSpecFix, ACI (projected) in \smBestAci, rolling conformal in \smBestRoll{} and spectral ACI in \smBestSpecAci. In \smDecayBestNarrow{} of the cells in which decay is best, its margin over the runner-up is below the standard error of its own mean score, so the ranking among the leaders is not resolved. The spectral weights do not help in the scenarios without regime structure (stationary $\mathrm{AR}(1)$, white noise and pure drift): in all six cells DASC is wider than ACI and has the higher interval score.

\emph{The scenario built for spectral weighting.} In this design the spectral layer behaves as a substitute for a faster step size. At an equal step, the gap in mean interval score between ACI and spectral ACI shrinks as the step grows (Section~\ref{sec:robust}). When every method searches the same step sizes, spectral ACI has the lower interval score with the ridge forecaster, by a small margin and at a wider interval, and ties ACI with the naive forecaster. The recurring regimes differ in frequency and noise, and the methods differ in how they split coverage between the two regimes (Table~\ref{tab:sim-gamma} and Figure~\ref{S-fig:regime} in Appendix~\ref{S-app:tables}, with the naive forecaster). At the selected step sizes, ACI (projected, $\gamma=\scRecFNNaiveAciGamma$) covers the two regimes at \rgRecFNNaiveAciROneCov{} and \rgRecFNNaiveAciRTwoCov, spectral ACI ($\gamma=\scRecFNNaiveSpecAciGamma$) at \rgRecFNNaiveSpecAciROneCov{} and \rgRecFNNaiveSpecAciRTwoCov, and DASC (projected, $\gamma=\dcGammaDefault$) at \rgRecFNNaiveDascROneCov{} and \rgRecFNNaiveDascRTwoCov. At the equal step size $\gamma=\dcGammaDefault$ for all three, ACI (projected) covers \gsRecFNNaiveAciProRone{} and \gsRecFNNaiveAciProRtwo{} with interval score \gsRecFNNaiveAciProIs{} (standard error \gsRecFNNaiveAciProIsSe). Spectral ACI covers \gsRecFNNaiveSpecAciRone{} and \gsRecFNNaiveSpecAciRtwo{} with score \gsRecFNNaiveSpecAciIs{} (\gsRecFNNaiveSpecAciIsSe), and DASC \gsRecFNNaiveDascRone{} and \gsRecFNNaiveDascRtwo{} with score \gsRecFNNaiveDascIs{} (\gsRecFNNaiveDascIsSe). So the spectral weights balance the two regimes better than ACI when all three use the same slow step. ACI with its tuned step size reaches a similar balance without any spectral information: coverage \rgRecFNNaiveAciROneCov{} and \rgRecFNNaiveAciRTwoCov{} and score \smRecFNNaiveAciIs. The gate does not help here: at the equal step size the score of DASC is above that of spectral ACI. With the ridge forecaster, which removes most of the regime structure, the ordering is the same at smaller gaps (Table~\ref{tab:sim-gamma}).

{\AtBeginEnvironment{tabular}{\scriptsize\setlength{\tabcolsep}{3pt}}\input{tab_sim_gamma_paper}}

\emph{The gate.} The gate's effect is small and has both signs. The comparison of DASC with a fixed level against the spectral method with a fixed level is free of the step-size confound. Over the \smNCellsWord{} cells the gate lowers the interval score by more than two standard errors in \smGateBetter{}, raises it in \smGateWorse{} and leaves it within two standard errors in \smGateTie. Its effect on width lies between \smGateWidMin{} and \smGateWidMax{} per cent. A possible explanation, not tested here, is that the gate helps after a break, where recency truncation removes stale residuals, and hurts under recurrence, where it discards old residuals from the same regime that the weights were meant to recover. In the original three-phase stream the coverage by regime (the two stationary phases and the drift phase) at a fixed level is \rgThreeNaiveSpecFixACov, \rgThreeNaiveSpecFixBCov{} and \rgThreeNaiveSpecFixDriftCov{} for the spectral method and \rgThreeNaiveDascFixACov, \rgThreeNaiveDascFixBCov{} and \rgThreeNaiveDascFixDriftCov{} for DASC.

\emph{The free update.} DASC (free) is wider than ACI (free) by \smFreeInMin{} to \smFreeInMax{} per cent in \smFreeNIn{} of the \smNCellsWord{} cells. The exception is the frequency-and-noise scenario with the naive forecaster, where free ACI oscillates between the regimes (coverage \rgRecFNNaiveAciFreeROneCov{} and \rgRecFNNaiveAciFreeRTwoCov) and DASC (free) is \spRecFNNaiveFreeFreeWAbs{} per cent narrower with interval score lower by \spRecFNNaiveFreeFreeSAbs. With the ridge forecaster the same scenario also favours DASC (free) in the interval score, which is lower by more than two standard errors, while the width difference is within one standard error of zero.

\subsection{Real data}
\label{sec:results-real}

Over the \rvNCellsWord{} cells, which are three series with three forecasters each and so not \rvNCellsWord{} independent comparisons, the coverage of DASC (projected) is between \rvDascCovMin{} and \rvDascCovMax. Table~\ref{tab:headline-real} brings the two comparisons together for every data set and forecaster. Its left block is ACI with the step size and pool selected on the development prefix. The next block is ACI at DASC's step size $\gamma=\dcGammaDefault$ and DASC's selected pool; this configuration was not selected on the development prefix and is shown only to make the comparison symmetric. The third block is DASC as in the main runs, and the last is DASC with its step size tuned at the larger budget of Section~\ref{sec:design-tuning}. The full tables of all \rvNMethodsWord{} methods are Tables~\ref{S-tab:real-electricity}, \ref{S-tab:real-weather} and~\ref{S-tab:real-finance} in Appendix~\ref{S-app:tables}, and the paired comparisons with bootstrap intervals are in Table~\ref{S-tab:real-paired}. Figure~\ref{S-fig:cov-width} in the same appendix plots coverage against width for all \rvNMethodsWord{} methods.

{\AtBeginEnvironment{tabular}{\scriptsize\setlength{\tabcolsep}{3pt}}\input{tab_headline_real}}

In the paired comparison with ACI (projected), the width interval lies above zero in \prAciWidPos{} cells, includes zero in \prAciWidIncl{} and lies below zero in \prAciWidNeg. The interval-score interval lies above zero in \prAciIsPos{} cells and below zero in \prAciIsNeg. At DASC's step size and pool, projected ACI and DASC have widths that differ by at most \agRelWidMaxNonFin{} per cent on electricity and weather (DASC is narrower in \agNNarrowWord{} of these cells, by \agNarrowMax{} per cent), and DASC is \agRelWidFinMin{} to \agRelWidFinMax{} per cent wider on finance (Table~\ref{tab:headline-real}). The comparison at DASC's step size has no bootstrap interval. The cells that favour DASC in the main runs are described below.

\emph{Electricity.} With the naive forecaster DASC is \prElecNaiveDascAciWAbs{} per cent wider than ACI (\prElecNaiveDascAciWci), with coverage \rvElecNaiveDascCov{} against \rvElecNaiveAciCov. With ridge and gradient boosting it is narrower, by \prElecRidgeDascAciWAbs{} per cent (\prElecRidgeDascAciWci) and \prElecBoostDascAciWAbs{} per cent (\prElecBoostDascAciWci). The interval-score differences, \prElecRidgeDascAciS{} and \prElecBoostDascAciS, have intervals that lie below zero. These two cells are consistent with a step-size effect of the comparison. The development prefix selected $\gamma=\asElecRidgeAciGamma$ for ACI with the ridge forecaster and DASC's step size is $\dcGammaDefault$; ACI run at DASC's step size and DASC's pool gives \agElecRidgeAciCovF{} and \agElecRidgeAciWid{} against DASC's \rvElecRidgeDascCovF{} and \rvElecRidgeDascWid{} (Table~\ref{tab:headline-real}). This matched-step comparison has no interval. With the strong forecasters the projection never binds, so DASC (projected) and DASC (free) are identical step for step. Spectral ACI and ACI are within \prElecRidgeSpecAciWAbs{} per cent of each other in width with the ridge forecaster, because the kernel weights are close to uniform at the tuned bandwidth: the median effective sample size is \rvElecRidgeSpecAciNeff{} for a pool of \rvElecRidgeSpecAciPool. Rolling conformal prediction, with no adaptation, and the spectral method with a fixed level have the lowest interval scores of the \rvNMethodsWord{} methods with ridge and agree to two decimals (\rvElecRidgeRollIs{} for both, against \rvElecRidgeDascIs{} for DASC). With gradient boosting it has the lowest, with the spectral method with a fixed level second (\rvElecBoostRollIs{} against \rvElecBoostDascIs{} for DASC).

\emph{Weather.} In the main runs ACI and the other baselines use the settings selected on the development prefix, while the step size of DASC is fixed. With the naive forecaster, DASC (projected) is \prWeaNaiveDascAciWAbs{} per cent wider than ACI (\prWeaNaiveDascAciWci) at lower coverage (\rvWeaNaiveDascCov{} against \rvWeaNaiveAciCov). The gap is consistent with the step size and not the spectral layer: ACI at DASC's step size $\gamma=\dcGammaDefault$ gives \agWeaNaiveAciCovF{} and \agWeaNaiveAciWid{} against DASC's \rvWeaNaiveDascCovF{} and \rvWeaNaiveDascWid. This matched-step comparison has no interval. The development prefix selected $\gamma=\asWeaNaiveAciGamma$ for projected ACI, which then spends \rvWeaNaiveAciCap{} per cent of its updates at the cap (Table~\ref{S-tab:projected-identity} in Appendix~\ref{S-app:tables}). A possible explanation, not tested here, is that the weather stream rewards a fast step size because the residual scale changes with the season, so that a method whose step size is fixed loses that benefit.

With the ridge forecaster the methods tie (\prWeaRidgeDascAciW{} per cent, \prWeaRidgeDascAciWci). With gradient boosting DASC is \prWeaBoostDascAciWAbs{} per cent wider (\prWeaBoostDascAciWci) at higher coverage (\rvWeaBoostDascCov{} against \rvWeaBoostAciCov, where ACI under-covers). Against spectral ACI with the naive forecaster DASC is \prWeaNaiveDascSpecWAbs{} per cent wider (\prWeaNaiveDascSpecWci) at equal coverage (\rvWeaNaiveDascCov{} against \rvWeaNaiveSpecAciCov). Against decay-weighted conformal prediction it is \prWeaNaiveDascDecayWAbs{} per cent narrower (\prWeaNaiveDascDecayWci), while decay under-covers (\rvWeaNaiveDecayCov).

The free variants favour DASC in the main runs: DASC (free) is \prWeaNaiveFreeFreeWAbs{} per cent narrower than ACI (free) (\prWeaNaiveFreeFreeWci), with interval score lower by \prWeaNaiveFreeFreeSAbs. Free ACI selected $\gamma=\asWeaNaiveAciFreeGamma$ and DASC uses $\dcGammaDefault$; at $\gamma=\dcGammaDefault$ free ACI has coverage \agWeaNaiveAciFreeCovF{} and width \agWeaNaiveAciFreeWid{} against DASC (free) at \rvWeaNaiveDascFreeCovF{} and \rvWeaNaiveDascFreeWid, so ACI is narrower at equal step size. The infinite-score rule makes free-against-free a slow-step comparison, because a configuration with any whole-line step has an infinite development score and is not selected.

\emph{Finance.} DASC (projected) is wider than ACI (projected), spectral ACI and decay-weighted conformal prediction in all three forecaster cells: by \prFinWidMinAci{} to \prFinWidMaxAci{} per cent against ACI, \prFinWidMinSpec{} to \prFinWidMaxSpec{} per cent against spectral ACI and \prFinWidMinDecay{} to \prFinWidMaxDecay{} per cent against decay. The widths come with coverage that is equal or higher for DASC (against ACI with the naive forecaster, \prFinNaiveDascAciCpt{} points; \prFinNaiveDascAciWci{} per cent in width). The interval-score intervals of these differences all include zero; against ACI with the naive forecaster the difference is \prFinNaiveDascAciS{} (\prFinNaiveDascAciSci).

Projected ACI under-covers on finance (\rvFinNaiveAciCov{} to \rvFinRidgeAciCov{} across the forecasters, the lowest value being with the naive forecaster) because its selected step size sends it to the floor in \rvFinNaiveAciFloor{} per cent of the updates with the naive forecaster. The projection term is \rvFinNaiveAciProj; the free update has no such term, and its coverage is \rvFinNaiveAciFreeCov.

Projected ACI run at DASC's step size gives \agFinNaiveAciCovF{} and \agFinNaiveAciWid{} (naive), \agFinRidgeAciCovF{} and \agFinRidgeAciWid{} (ridge) and \agFinBoostAciCovF{} and \agFinBoostAciWid{} (boosting). DASC gives \rvFinNaiveDascCovF{} and \rvFinNaiveDascWid{} (naive), \rvFinRidgeDascCovF{} and \rvFinRidgeDascWid{} (ridge) and \rvFinBoostDascCovF{} and \rvFinBoostDascWid{} (boosting). In all three cells DASC is wider than this projected ACI and, as a point estimate with ACI run at DASC's step size, has the lower interval score: \rvFinNaiveDascIs{} against \agFinNaiveAciIs{} with the naive forecaster, \rvFinRidgeDascIs{} against \agFinRidgeAciIs{} with ridge and \rvFinBoostDascIs{} against \agFinBoostAciIs{} with gradient boosting. This ACI configuration was not selected on the development prefix and has no bootstrap interval. At the step size that ACI selected, the paired interval-score intervals include zero.

\emph{The finance cell with the lowest interval score.} With the naive forecaster on finance, DASC (free) has the lowest mean interval score of all \rvNMethodsWord{} methods over each method's finite steps, \rvFinNaiveDascFreeIs. Against ACI (free) the difference is \prFinNaiveFreeFreeS{} (\prFinNaiveFreeFreeSci), and against DASC (projected) it is \prFinNaiveFreeProjS{} (\prFinNaiveFreeProjSci); both intervals exclude zero. The upper limit of the interval against ACI (free) is close to zero, and the check of the bootstrap block length and seed in Section~\ref{sec:robust} did not cover the free update. The result does not repeat with the ridge or the gradient-boosting forecaster (\prFinRidgeFreeFreeSci{} and \prFinBoostFreeFreeSci{} against ACI free), nor in any simulation cell. It is the one cell out of the \rvNCellsWord{} in which a DASC variant has the lowest mean interval score (a point estimate). It is not the only cell that favours DASC: the other cells are described above and in the next subsection.

\emph{Conditional coverage.} The worst \dcWorstElec-hour block of DASC (projected) on electricity has coverage between \wbElecDascMin{} and \wbElecDascMax{} across the forecasters, against \wbElecAciMin{} to \wbElecAciMax{} for ACI (projected). On weather the worst block is \wbWeaDascMin{} to \wbWeaDascMax{} for DASC (projected) and \wbWeaAciMin{} to \wbWeaAciMax{} for ACI (projected). On electricity the worst-block coverage of DASC is lower than that of ACI with every forecaster; on weather the ranges for DASC and ACI overlap, and with gradient boosting the worst-block coverage is higher for DASC than for ACI (Table~\ref{S-tab:real-weather}).

With the naive forecaster, coverage by hour of day lies between \lcHourElecNaiveMin{} and \lcHourElecNaiveMax{} on electricity and between \lcHourWeaNaiveMin{} and \lcHourWeaNaiveMax{} on weather, over all methods and hours (Figure~\ref{S-fig:hour}). On electricity the methods differ little by hour. On weather the methods with a fixed level are the most uniform across hours: the mean absolute deviation from the target is \lcMadWeaNaiveRoll{} for rolling conformal prediction, \lcMadWeaNaiveSpecFix{} for the spectral method and \lcMadWeaNaiveDascFix{} for DASC with a fixed level, against \lcMadWeaNaiveDasc{} for DASC (projected) and \lcMadWeaNaiveAci{} for ACI (projected). DASC (projected) is not better than ACI (projected) on this measure.

On finance in 2022, the calendar year of the volatility episode, with the naive forecaster, the coverage of rolling conformal prediction is \lcFinNaiveRollYr{}, of ACI (projected) \lcFinNaiveAciYr{}, of ACI (free) \lcFinNaiveAciFreeYr, of spectral ACI \lcFinNaiveSpecAciYr, of decay-weighted \lcFinNaiveDecayYr, of the PID-style controller \lcFinNaivePidYr, of DASC (projected) \lcFinNaiveDascYr{} and of DASC (free) \lcFinNaiveDascFreeYr. The earlier preprint contrasted DASC with rolling conformal prediction alone in that year. Against the adaptive methods the contrast is small. By volatility tercile (low, middle, high), DASC (projected) covers \lcFinNaiveDascVolLow, \lcFinNaiveDascVolMid{} and \lcFinNaiveDascVolHigh{} and ACI (projected) covers \lcFinNaiveAciVolLow, \lcFinNaiveAciVolMid{} and \lcFinNaiveAciVolHigh, with DASC widest in the high-volatility tercile. The low tercile holds far fewer steps than the other two, and the coverages of the two methods in it differ by a single covered step.

\subsection{Equal step-size budget}
\label{sec:results-equal}

Giving DASC a step-size budget at least as large as ACI's (Section~\ref{sec:design-tuning}) does not change the conclusion. In the main runs ACI's step size is tuned over \dcGammaGridWord{} values and DASC's is fixed. Table~\ref{tab:equal-budget} gives the paired width differences of DASC (projected) with tuned $\gamma$ against ACI (projected), spectral ACI and decay-weighted conformal prediction.

{\AtBeginEnvironment{tabular}{\scriptsize\setlength{\tabcolsep}{3pt}}\input{tab_equal_budget}}

Against ACI, the width interval includes zero in all \ptNonFinCellsWord{} cells on electricity and weather, and lies above zero in all \ptAciWidPosFinWord{} finance cells, with DASC wider by \ptAciWidFinMin{} to \ptAciWidFinMax{} per cent. These are \rvNCellsWord{} cells from three series with three forecasters each, not \rvNCellsWord{} independent comparisons. The interval-score interval lies below zero in \ptAciIsNeg{} of the \rvNCellsWord{} cells and above zero in \ptAciIsPos{} (electricity with gradient boosting, only just). These counts are for the block lengths of the paired bootstrap (\bootBlockElec{} hours, \bootBlockFin{} trading days) and the bootstrap seed of the main analysis. At longer block lengths and with other seeds (Section~\ref{sec:robust}), the width intervals against ACI keep this pattern in every draw. The interval-score interval for electricity with the ridge forecaster lies below zero with every seed at \exBlkMidElec{} hours and with one other seed at the main block length. The difference is a small share of the ACI score, with no gain in width or coverage. Whether it is an artefact of the bootstrap was not tested. On weather with the naive forecaster the tuned DASC has coverage \tgWeaNaiveDascTgCovF{} and width \tgWeaNaiveDascTgWid{} against ACI's \rvWeaNaiveAciCovF{} and \rvWeaNaiveAciWid{} (\ptWeaNaiveTgAciW{} per cent, \ptWeaNaiveTgAciWci), so the earlier weather gain is not recovered by tuning. Tuning the step size helps DASC against its own fixed-step version on weather with the naive forecaster (\ptWeaNaiveTgDascW{} per cent, \ptWeaNaiveTgDascWci, at \ptWeaNaiveTgDascCpt{} points more coverage). Tuning hurts it on electricity with the ridge forecaster (\ptElecRidgeTgDascW{} per cent, \ptElecRidgeTgDascWci), where the step selected on the development prefix of \dcElecDevSteps{} steps is not the best step on the evaluation period.

The cells in which DASC (tuned) has a width or interval-score interval that excludes zero in its favour are all against spectral ACI or decay-weighted conformal prediction. In each of them, at the block length and the bootstrap seed of the main analysis, the paired width and interval-score intervals for the difference between DASC and ACI include zero. At longer block lengths this holds for the width intervals; the interval-score interval of electricity with the ridge forecaster lies below zero in every seed at \exBlkMidElec{} hours (Section~\ref{sec:robust}). Against spectral ACI none of these cells remains when spectral ACI searches the same step sizes; Section~\ref{sec:robust} gives the cells and their values. On weather with the naive forecaster, spectral ACI and decay under-cover at their selected settings (\rvWeaNaiveSpecAciCov{} and \rvWeaNaiveDecayCov, against \tgWeaNaiveDascTgCov{} for DASC), so the narrower intervals of DASC (\ptWeaNaiveTgSpecW{} and \ptWeaNaiveTgDecayW{} per cent) come with \ptWeaNaiveTgSpecCpt{} and \ptWeaNaiveTgDecayCpt{} points more coverage. With gradient boosting on weather, DASC is \ptWeaBoostTgDecayWAbs{} per cent narrower than decay (\ptWeaBoostTgDecayWci) at equal coverage and with interval score lower by \ptWeaBoostTgDecaySAbs; ACI differs from DASC there by \ptWeaBoostTgAciW{} per cent (\ptWeaBoostTgAciWci). In the same cell DASC is \ptWeaBoostTgSpecWAbs{} per cent narrower than spectral ACI (\ptWeaBoostTgSpecWci), at \ptWeaBoostTgSpecCpt{} points of coverage. This width interval depends on the block length of the bootstrap: in the check of Section~\ref{sec:design-measures} it lies below zero in every seed at the main-run block length and in \exBlkSpecWeaBoostWMidWord{} of the \exSeedNSeedsWord{} seeds at \exBlkMidElec{} hours. The draws at \exBlkLongElec{} hours are a sensitivity check only (Section~\ref{sec:robust}). On electricity with the naive forecaster, DASC is \ptElecNaiveTgSpecWAbs{} per cent narrower than spectral ACI (\ptElecNaiveTgSpecWci) at \ptElecNaiveTgSpecCpt{} points of coverage, with interval score lower by \ptElecNaiveTgSpecSAbs. On electricity with the ridge forecaster, the interval-score interval against spectral ACI lies below zero, but only just, and the width interval includes zero.

In the simulations, DASC (projected) with tuned $\gamma$ is narrower than ACI (projected) by more than two standard errors in \tgNAciNarrow{} of the \smNCellsWord{} cells and wider in \tgNAciWide{}, with a mean relative width difference of \tgNAciWidMean{} per cent. Its interval score is better in \tgNAciIsBetter{} cell and worse in \tgNAciIsWorse. If the \smNCellsWord{} cells were independent, a criterion of two standard errors would pass about one cell by chance across both directions, about one half in each direction. The cell with the lower interval score (heavy-tail scenario with the ridge forecaster, just over two standard errors) is within what chance gives. In the main runs, the \smNReps{} replications give standard errors of \smSeWidMin{} to \smSeWidMax{} per cent for the relative width, against width differences of \smWidMin{} to \smWidMax{} per cent (Section~\ref{sec:results-sim}). Against spectral ACI it is narrower in \tgNSpecNarrow{} cells and wider in \tgNSpecWide{} (score better in \tgNSpecIsBetter, worse in \tgNSpecIsWorse). Against decay-weighted conformal prediction it is narrower in \tgNDecayNarrow{} and wider in \tgNDecayWide{} (score better in \tgNDecayIsBetter, worse in \tgNDecayIsWorse); the two cells in which it is better are the frequency-and-noise scenarios in which decay under-covers. Against DASC with $\gamma$ fixed it is narrower in \tgNDascNarrow{} cells. In the frequency-and-noise scenario with the naive forecaster the tuned step size for DASC is $\sgRecFNNaiveDascTgGamma$, and the result is a tie with ACI (coverage \sgRecFNNaiveDascTgCov{} (\sgRecFNNaiveDascTgCovSe), width \sgRecFNNaiveDascTgWid{} (\sgRecFNNaiveDascTgWidSe), interval score \sgRecFNNaiveDascTgIs{} (\sgRecFNNaiveDascTgIsSe) against ACI's \smRecFNNaiveAciIsF).

In the simulations with the free update and tuned $\gamma$, DASC (free) is narrower than free ACI by more than two standard errors in \tgNFreeAciNarrow{} cells and wider in \tgNFreeAciWide{}; the interval score is better in \tgNFreeAciIsBetter{} and worse in \tgNFreeAciIsWorse. On real data, against free ACI at tuned step sizes, the width interval of DASC (free) lies above zero in \ptFreeWide{} cells (electricity with ridge, weather with the naive forecaster and weather with gradient boosting), below zero in \ptFreeNarrow{} and includes zero in \ptFreeIncl. On weather with the naive forecaster DASC (free) is \ptWeaNaiveTgFreeAciWAbs{} per cent wider than free ACI (\ptWeaNaiveTgFreeAciWci), in contrast to the main run.

\subsection{The kernel weights, the drift score and the gate on the real streams}
\label{sec:results-weights}

\emph{At the tuned bandwidths of the main-run grid (\dcHGrid), the kernel weights are close to uniform on electricity and weather and are not on finance.} With bandwidths scaled to the feature spread, the weights of spectral ACI and DASC (projected) stay close to uniform on the hourly streams in the primary extended design (median effective sample size \exNeffSpecAciExtHourMin{} to \exNeffSpecAciExtHourMax{} and \exNeffDascExtHourMin{} to \exNeffDascExtHourMax{} of the pool), and the methods with a fixed level choose localised weights in most hourly cells (Section~\ref{sec:robust}). In the main runs, for the spectral methods on electricity, the median effective sample size is \rvNeffRngElec{} for pools of at most \rvPoolMaxElec; on weather it is \rvNeffRngWea{} for pools of at most \rvPoolMaxWea. The ratio of the median effective sample size to the mean pool size is \rvNeffRatioRngElec{} per cent on electricity and \rvNeffRatioRngWea{} per cent on weather. DASC (projected) on weather with the naive forecaster has median effective sample size \rvWeaNaiveDascNeff{} for a pool of at most \dsWeaNaivePool. The weight-averaged spectral distance $\bar d_t$ is small there (\rvWeaNaiveDascDbar, consistent with normalised features that barely vary) and larger on electricity (\rvElecRidgeDascDbar{} with the ridge forecaster). On finance the weights are concentrated: the median effective sample size of the spectral methods is \rvNeffMinFin{} to \rvNeffMaxFin{} for pools of at most \rvPoolMaxFin, and \rvFinBoostDascNeff{} to \rvFinNaiveDascNeff{} for DASC (projected) across the forecasters. At the main-run bandwidths, finance is the one data set on which the kernel weights depart from uniform, and on it DASC is wider than ACI.

The computable part of the per-step lower bound, $1-\alpha_t-\varepsilon_t(0.05)$, has median \rvWeaNaiveDascLbFifty{} for DASC (projected) on weather with the naive forecaster, and its tenth percentile falls to \rvWeaNaiveDascLbTen{} where the level rises towards the cap (Table~\ref{S-tab:perstep-computable} in Appendix~\ref{S-app:tables}).

\emph{The gate.} On the real data the gate's effect at a fixed level, DASC with a fixed level minus the spectral method with a fixed level, is small and has both signs, as in the simulations. The width interval includes zero in all \prGateWidInclWord{} cells, with point estimates from \prGateWidMin{} to \prGateWidMax{} per cent. The interval-score interval excludes zero in \prGateIsExcl{} cells: it is above zero for electricity with the ridge forecaster, where its lower limit rounds to \prElecRidgeGateSLo{}, and below zero for finance with the naive forecaster (Table~\ref{S-tab:real-paired}).

\emph{The drift score.} In simulated streams the drift score is a weaker break detector than a rolling coverage deficit; its advantage is that it needs no outcomes. Table~\ref{S-tab:drift-detection} in Appendix~\ref{S-app:tables} compares the drift score $D_t$ with the coverage deficit (the shortfall of the coverage of rolling conformal prediction over the last \dcDeficitWin{} steps against $\alpha$), with the threshold at the \ddThrPct th percentile of the statistic before the break and detection within \ddHorizon{} steps. Over the \ddNCombosWord{} combinations of scenario and break in Table~\ref{S-tab:drift-detection}, the AUC of the coverage deficit is higher in \ddAucDefHigherWord{} of them. The comparison sets an outcome-free statistic against an outcome-based one, and in two of the combinations the AUC contrast is about one standard error or less.

At the break at step \dcBreakSecond{} of the original three-phase stream, which combines level, variance and frequency drift, $D_t$ detects the break in \ddThreeBrkNineZeroZeroDRate{} of the replications with AUC \ddThreeBrkNineZeroZeroDAuc{} (\ddThreeBrkNineZeroZeroDAucSe), and the coverage deficit in \ddThreeBrkNineZeroZeroDefRate{} with AUC \ddThreeBrkNineZeroZeroDefAuc{} (\ddThreeBrkNineZeroZeroDefAucSe). At the frequency change at step \dcBreakFirst{} both statistics are weak (AUC \ddThreeBrkFourFiveZeroDAuc{} and \ddThreeBrkFourFiveZeroDefAuc). In the mixed-drift stress scenario the detection rate of $D_t$ is \ddStMixedBrkEightFiveZeroDRate{} against \ddStMixedBrkEightFiveZeroDefRate{} for the deficit, but its AUC is lower (\ddStMixedBrkEightFiveZeroDAuc{} against \ddStMixedBrkEightFiveZeroDefAuc) and its median delay longer (\ddStMixedBrkEightFiveZeroDDelay{} against \ddStMixedBrkEightFiveZeroDefDelay{} steps). In the gradual-frequency scenario the median delay of $D_t$ is shorter (\ddStGradualBrkFiveZeroZeroDDelay{} against \ddStGradualBrkFiveZeroZeroDefDelay{} steps) at a higher false-alarm rate (\ddStGradualBrkFiveZeroZeroDFa{} against \ddStGradualBrkFiveZeroZeroDefFa). The false-alarm rate is measured against a threshold set at the \ddThrPct th percentile of the same statistic before the break, so about one per cent is built in and only the excess over it is informative.

\emph{Finance in 2022.} The drift score ranks January, the month with the lowest DASC coverage, among the highest months, and it ranks March among the lowest. Table~\ref{S-tab:finance-2022-drift} and Figure~\ref{S-fig:finance-drift} in Appendix~\ref{S-app:tables} rank the months of 2022 by $D_t$ and by the coverage deficit among the \fyNMonths{} evaluated months. The monthly mean of $D_t$ ranks January, May and June at \fyDRankJan, \fyDRankMay{} and \fyDRankJun, and March at \fyDRankMar. The coverage deficit ranks March to July at \fyDefRankMin{} to \fyDefRankMax{} and January at \fyDefRankJan. DASC's own coverage is \fyCovJan{} in January, \fyCovMar{} in March and \fyCovJun{} in June. This is consistent with a coverage deficit that lags the failures it measures.

\input{s6_robustness}

\subsection{Long-run identities in the runs}
\label{sec:results-identity}

\emph{Free update.} On every real stream the realised gap $\abs{\bar E_T-\alpha}$ of the free update is below the bound $(\max(\alpha_1,1-\alpha_1)+\gamma)/(\gamma T)$ (Table~\ref{S-tab:longrun-free} in Appendix~\ref{S-app:tables}). For DASC (free) with the naive forecaster the gap is \lrWeaNaiveDascFreeGap{} against a bound of \lrWeaNaiveDascFreeBound{} on weather, \lrElecNaiveDascFreeGap{} against \lrElecNaiveDascFreeBound{} on electricity and \lrFinNaiveDascFreeGap{} against \lrFinNaiveDascFreeBound{} on finance. The level reaches \rvWeaNaiveDascFreeAlphaMin{} on weather with the naive forecaster, and the interval is the whole line in \rvWeaNaiveDascFreeNInf{} of \rvNWea{} steps.

\emph{Projected update.} The exact decomposition $\bar E_T-\alpha=(\alpha_1-\alpha_{T+1})/(\gamma T)+\sum_t\epsilon_t/(\gamma T)$ of Theorem~\ref{thm:projected} holds to rounding in every run (largest absolute gap \rvIdGapMax). On electricity with the ridge or boosting forecaster the projection never binds. The binding case is weather with the naive forecaster (Table~\ref{S-tab:projected-identity} in Appendix~\ref{S-app:tables}). For projected ACI the cap binds in \rvWeaNaiveAciCap{} per cent of the steps and the floor in \rvWeaNaiveAciFloor{} per cent; the telescoping term is \rvWeaNaiveAciTele, the projection term \rvWeaNaiveAciProj{} and the realised gap \rvWeaNaiveAciGap. The cap and floor terms offset each other: the worst-case bounds of Theorem~\ref{thm:projected}(c) are $-\alpha N^+_T/T=\ensuremath{-}\rvWeaNaiveAciCapBound$ and $(1-\alpha)N^-_T/T=\rvWeaNaiveAciFloorBound$, while the projection term is \rvWeaNaiveAciProj.

For DASC (projected) on the same stream the cap binds in \rvWeaNaiveDascCap{} per cent of the steps and the floor in \rvWeaNaiveDascFloor{} per cent, the projection term is \rvWeaNaiveDascProj{} and the realised gap \rvWeaNaiveDascGap. The realised gap is the sum of the telescoping term \rvWeaNaiveDascTele{} and the projection term, so the floor accounts for the shortfall of its coverage (\rvWeaNaiveDascCov) below the target. The cap fraction of \rpOldCapShare{} per cent belongs to the earlier lag-one configuration (Section~\ref{sec:longrun}), and the like-for-like run has \rvWeaNaiveAciCap{} per cent for projected ACI.

\begin{figure}[t]
\centering
\includegraphics[width=\textwidth]{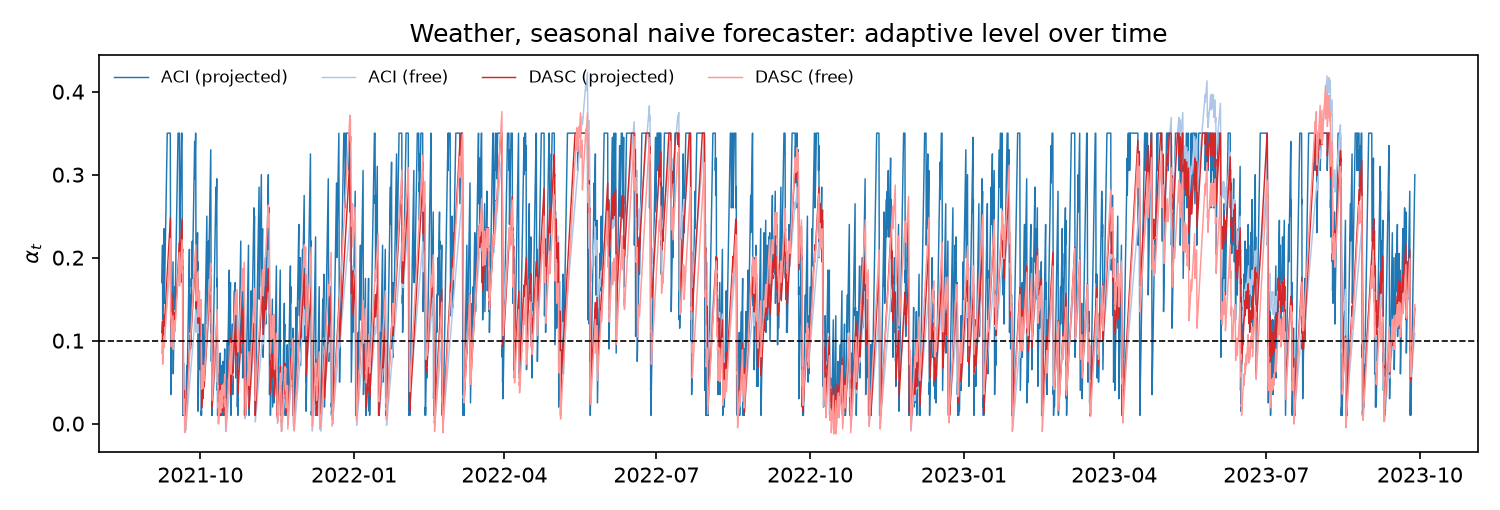}
\caption{Adaptive level $\alpha_t$ over the evaluation steps for weather with the naive forecaster: ACI and DASC, projected and free. The dashed line is $\alpha$. The projected level is flat at its cap in stretches for ACI, and it goes below zero for the free update (whole-line intervals).}
\label{fig:alpha}
\end{figure}

\subsection{The smoothness assumption on the real streams}
\label{sec:results-assumption}

Table~\ref{S-tab:assumption-diagnostic} in Appendix~\ref{S-app:tables} reports $\widehat L_{\min}(\delta,\eta)$ for clusters of the spectral features, with the number of cluster pairs that refute a constant $L=0$. The maximum over pairs is read at the level \adConfPctSim{} per cent, which for the $J$ clusters kept (\adJMin{} to \adJMax{} on the real streams) means $\eta=0.05/J$ for each cluster, that is, \adEtaAdjMin{} to \adEtaAdjMax{} (Section~\ref{sec:check}). On weather with the seasonal naive forecaster the diagnostic is large: $\widehat L_{\min}(0)$ is \adWeaNaiveLMinDzero{} to \adWeaNaiveLMaxDzero, and $\widehat L_{\min}(0.1)$ is \adWeaNaiveLMinDten{} to \adWeaNaiveLMaxDten, with \adWeaNaiveRefMinDten{} to \adWeaNaiveRefMaxDten{} of the cluster pairs (out of \adWeaNaivePairsMax{} at most) refuting even $\delta=0.1$. The median cluster distance is \adWeaNaiveMedDMin{} to \adWeaNaiveMedDMax, consistent with features of the clusters that hardly separate while the residual laws differ by season. Under Assumption~\ref{ass:A}, the pair of a small $L$ and a small $\delta$ is refuted, at the level \adConfPctSim{} per cent, for the weather data with the seasonal naive forecaster. With the ridge forecaster on weather nothing is refuted at any $\delta$ in the table, and $\widehat L_{\min}(0)$ is \adWeaRidgeMaxL. On electricity $\widehat L_{\min}(0)$ is at most \adElecAllMaxL{} across forecasters and nothing is refuted at $\delta=0.1$. On finance $\widehat L_{\min}(0)$ is at most \adFinAllMaxL.

Two cautions apply. The residuals are serially dependent, so the concentration terms are not valid and the test is indicative on real data. And the test cannot say whether part (i) or part (ii) of Assumption~\ref{ass:A} fails (Section~\ref{sec:check}). The test also cannot detect a feature that ignores the regime: in the amplitude-only recurring scenario nothing is refuted because the clusters do not split the regimes.

\input{s6_pitfalls}

\subsection{Computational cost}
\label{sec:results-cost}

DASC (projected) takes \tmSimDascMs{} ms per step at a pool of \tmPoolSim{} and \tmWeaDascMs{} ms at a pool of \tmPoolWea, against \tmUniMinSim{} to \tmUniMaxSim{} microseconds for rolling conformal prediction, ACI and the PID-style controller at the smaller pool and \tmUniMinWea{} to \tmUniMaxWea{} microseconds at the larger. The slowest method in any setting takes \tmAllMax{} ms per step. Table~\ref{tab:timing} gives wall-clock time per step on one core for each method at each data set's maximum pool size, on a synthetic stream with the feature matrix computed beforehand; the Fourier transform of the new window, one per step, is not included. The times were measured on a shared machine while one other process used one core, so they are indicative. At hourly or daily data rates the cost is negligible. The cost grows with the pool size, and the implementation recomputes the weights over the whole pool at every step. Gradient-boosting forecasts were fitted with one thread, because the library's multi-threaded fit ran slowly when several processes shared the machine; tree construction can depend on the thread count in the last digits, so a multi-threaded fit may not reproduce the rows bit for bit.

{\AtBeginEnvironment{tabular}{\scriptsize\setlength{\tabcolsep}{3pt}}\input{tab_timing}}

%% file: tab_reproduction_paper.tex
\begin{table}[t]
\centering
\caption{Regression check of the new framework against the numbers of the earlier preprint, run under the original settings. The hourly DASC runs of the earlier preprint were centred on $y_{t-1}$; the recentred rows use the lag-24 forecast like the baselines. Published values are those of the earlier preprint.}
\label{tab:reproduction}
\small
\begin{tabular}{lllrrrr}
\toprule
Data set & Run & Method & Coverage & Width & Published cov. & Published width \\
\midrule
Synthetic & published settings, legacy features & DASC & 0.8995 & 2.7517 & 0.8995 & 2.7517 \\
 & published settings & ACI & 0.8990 & 2.7348 & 0.8990 & 2.7381 \\
 & published settings, new features & DASC & 0.8995 & 2.7510 & 0.8995 & 2.7517 \\
Electricity & DASC on y[t-1], legacy features & DASC & 0.9033 & 2.5497 & 0.9033 & 2.5508 \\
 & published settings & ACI & 0.9001 & 3.5379 & 0.9001 & 3.5398 \\
 & DASC on y[t-1], new features & DASC & 0.9033 & 2.5497 & 0.9033 & 2.5508 \\
 & DASC recentred on y[t-24] & DASC & 0.9000 & 3.5372 & 0.9033 & 2.5508 \\
Weather & DASC on y[t-1], legacy features & DASC & 0.9655 & 6.4098 & 0.9656 & 6.4129 \\
 & published settings & ACI & 0.8938 & 11.0233 & 0.8938 & 11.0270 \\
 & DASC on y[t-1], new features & DASC & 0.9655 & 6.4098 & 0.9656 & 6.4129 \\
 & DASC recentred on y[t-24] & DASC & 0.8959 & 10.9198 & 0.9656 & 6.4129 \\
Finance & DASC on y[t-1], legacy features & DASC & 0.8988 & 2.2791 & 0.8988 & 2.2791 \\
 & published settings & ACI & 0.8954 & 2.2649 & 0.8954 & 2.2615 \\
 & DASC on y[t-1], new features & DASC & 0.8988 & 2.2791 & 0.8988 & 2.2791 \\
\bottomrule
\end{tabular}
\par\medskip\footnotesize Legacy features: the original simulation code returned raw, un-normalised power for windows of 8 to 12 points. The remaining hourly differences of 0.001 to 0.003 in width arise because the original code kept the 23 undefined residuals of the lag-24 forecaster in the pool for the first $m_{\max}$ steps. ACI rows differ by one order statistic (quantile convention).
\end{table}

%% file: tab_sim_headline.tex
\begin{table}[t]
\centering
\caption{DASC (projected) minus ACI (projected) in each simulated scenario with the naive (lag 1) forecaster: mean (standard error) over replications of the relative width difference in per cent of ACI's width, and of the interval-score difference. Positive values mean that DASC is wider or scores worse. ``Default'': DASC with $\gamma$ fixed and ACI tuned (main runs). ``Tuned'': DASC with $\gamma$ tuned over the ACI grid (equal-budget study).}
\label{tab:sim-headline}
\small
\begin{tabular}{lllll}
\toprule
Scenario & Width, default & Score, default & Width, tuned & Score, tuned \\
\midrule
Original three-phase & +0.81 (0.12) & +0.004 (0.002) & +1.03 (0.14) & +0.013 (0.002) \\
Recurring: frequency only & +0.66 (0.08) & +0.016 (0.001) & +0.15 (0.05) & +0.002 (0.000) \\
Recurring: frequency and noise & +3.33 (0.22) & +0.155 (0.007) & +1.05 (0.15) & $-$0.001 (0.004) \\
Recurring: amplitude only & +0.75 (0.13) & +0.016 (0.001) & +0.16 (0.08) & +0.002 (0.001) \\
Stationary AR(1) & +0.65 (0.11) & +0.013 (0.001) & +0.04 (0.06) & +0.001 (0.000) \\
Pure drift & +1.55 (0.14) & +0.026 (0.002) & +0.26 (0.07) & +0.002 (0.001) \\
White noise & +0.79 (0.12) & +0.032 (0.002) & +0.08 (0.08) & +0.007 (0.002) \\
Stress: abrupt shift & +0.32 (0.08) & +0.004 (0.001) & +0.51 (0.11) & +0.006 (0.002) \\
Stress: gradual frequency & +1.53 (0.10) & +0.012 (0.001) & +0.10 (0.06) & +0.002 (0.001) \\
Stress: heavy tail & +1.71 (0.17) & +0.022 (0.002) & $-$0.02 (0.08) & +0.003 (0.001) \\
Stress: mixed drift & +0.25 (0.09) & +0.010 (0.001) & +0.24 (0.07) & +0.004 (0.001) \\
Stress: weak recurrence & +0.36 (0.09) & +0.004 (0.001) & +0.45 (0.09) & +0.007 (0.001) \\
\bottomrule
\end{tabular}
\end{table}

%% file: tab_sim_gamma_paper.tex
\begin{table}[t]
\centering
\caption{Simulations at equal step size: ACI and spectral ACI re-run with DASC's step size $\gamma=0.015$ and DASC's selected pool (other selected parameters kept); DASC as selected. Mean (standard error) over 100 evaluation seeds. The last column is coverage by regime in order of first appearance.}
\label{tab:sim-gamma}
\small
\begin{tabular}{lllll}
\toprule
Method & Coverage (SE) & Width (SE) & Interval score (SE) & Coverage by regime \\
\midrule
\multicolumn{5}{l}{\emph{Recurring: frequency and noise, naive (lag 1) forecaster}} \\
ACI (projected) & 0.8995 (0.0002) & 3.152 (0.012) & 4.386 (0.015) & 0.984 / 0.812 \\
ACI (free) & 0.8996 (0.0001) & 3.155 (0.012) & 4.388 (0.015) & 0.984 / 0.813 \\
Spectral ACI & 0.8988 (0.0001) & 3.053 (0.009) & 4.079 (0.012) & 0.947 / 0.849 \\
DASC (projected) & 0.8988 (0.0002) & 3.119 (0.011) & 4.169 (0.013) & 0.951 / 0.844 \\
\multicolumn{5}{l}{\emph{Recurring: frequency and noise, rolling ridge forecaster}} \\
ACI (projected) & 0.8993 (0.0001) & 2.701 (0.008) & 3.606 (0.011) & 0.967 / 0.830 \\
ACI (free) & 0.8993 (0.0001) & 2.702 (0.008) & 3.607 (0.011) & 0.967 / 0.830 \\
Spectral ACI & 0.8986 (0.0001) & 2.675 (0.008) & 3.460 (0.010) & 0.935 / 0.861 \\
DASC (projected) & 0.8986 (0.0001) & 2.705 (0.009) & 3.509 (0.011) & 0.939 / 0.857 \\
\multicolumn{5}{l}{\emph{Original three-phase, naive (lag 1) forecaster}} \\
ACI (projected) & 0.8985 (0.0002) & 2.705 (0.008) & 3.399 (0.009) & 0.899 / 0.900 / 0.897 \\
ACI (free) & 0.9000 (0.0002) & 2.724 (0.009) & 3.407 (0.009) & 0.899 / 0.900 / 0.901 \\
Spectral ACI & 0.8987 (0.0002) & 2.714 (0.009) & 3.402 (0.009) & 0.899 / 0.900 / 0.898 \\
DASC (projected) & 0.8992 (0.0002) & 2.721 (0.009) & 3.406 (0.009) & 0.899 / 0.899 / 0.899 \\
\multicolumn{5}{l}{\emph{Original three-phase, rolling ridge forecaster}} \\
ACI (projected) & 0.8983 (0.0002) & 2.037 (0.006) & 2.603 (0.007) & 0.900 / 0.899 / 0.896 \\
ACI (free) & 0.9006 (0.0001) & 2.046 (0.005) & 2.594 (0.008) & 0.900 / 0.899 / 0.902 \\
Spectral ACI & 0.8994 (0.0002) & 2.054 (0.006) & 2.601 (0.007) & 0.900 / 0.899 / 0.899 \\
DASC (projected) & 0.8993 (0.0002) & 2.047 (0.006) & 2.595 (0.007) & 0.900 / 0.899 / 0.899 \\
\bottomrule
\end{tabular}
\end{table}

%% file: tab_headline_real.tex
\begin{table}[t]
\centering
\caption{Coverage / mean width / mean interval score on the evaluation steps, target coverage 0.90. ``ACI, selected'': projected ACI with the step size selected on the development prefix (in brackets). ``ACI at DASC's $\gamma$'': projected ACI at $\gamma=0.015$ with DASC's selected pool; not selected on the development prefix, shown to make the comparison symmetric. ``DASC'': projected DASC with $\gamma=0.015$. ``DASC, tuned'': projected DASC with $\gamma$ tuned over 60 configurations (in brackets).}
\label{tab:headline-real}
\small
\begin{tabular}{lllll}
\toprule
Forecaster & ACI, selected & ACI at DASC's $\gamma$ & DASC & DASC, tuned \\
\midrule
\multicolumn{5}{l}{\emph{Electricity}} \\
Seasonal naive & 0.896/3.42/4.73 (0.05) & 0.898/3.45/4.78 & 0.898/3.46/4.79 & 0.896/3.42/4.72 (0.05) \\
Rolling ridge & 0.900/2.09/3.02 (0.03) & 0.900/2.04/2.99 & 0.900/2.05/2.99 & 0.900/2.09/3.01 (0.02) \\
Gradient boosting & 0.900/2.05/2.97 (0.03) & 0.900/2.04/2.95 & 0.900/2.04/2.95 & 0.900/2.06/2.97 (0.03) \\
\multicolumn{5}{l}{\emph{Weather}} \\
Seasonal naive & 0.900/10.43/12.83 (0.05) & 0.894/11.50/15.37 & 0.893/11.53/15.37 & 0.899/10.41/12.83 (0.05) \\
Rolling ridge & 0.900/2.62/4.06 (0.015) & 0.900/2.62/4.06 & 0.900/2.63/4.06 & 0.900/2.63/4.06 (0.015) \\
Gradient boosting & 0.890/3.19/4.79 (0.05) & 0.895/3.30/4.99 & 0.895/3.30/5.00 & 0.890/3.19/4.78 (0.05) \\
\multicolumn{5}{l}{\emph{Finance}} \\
Naive (lag 1) & 0.888/2.22/3.56 (0.05) & 0.898/2.31/3.74 & 0.901/2.46/3.63 & 0.892/2.37/3.51 (0.05) \\
Rolling ridge & 0.894/1.83/3.01 (0.05) & 0.899/1.91/3.13 & 0.901/2.02/3.09 & 0.893/1.97/2.99 (0.05) \\
Gradient boosting & 0.891/1.93/3.27 (0.05) & 0.899/2.05/3.37 & 0.900/2.17/3.26 & 0.900/2.17/3.26 (0.015) \\
\bottomrule
\end{tabular}
\end{table}

%% file: tab_equal_budget.tex
\begin{table}[t]
\centering
\caption{Equal step-size budget: DASC (projected) with tuned $\gamma$ minus each comparator on the evaluation steps. Width: difference in per cent of the comparator's width with the 95 per cent paired moving-block bootstrap interval; negative values favour DASC. Cov.: coverage difference in percentage points.}
\label{tab:equal-budget}
\small
\begin{tabular}{llcrcrcr}
\toprule
Forecaster & Tuned $\gamma$ & Width vs ACI & Cov. & Width vs spectral ACI & Cov. & Width vs decay & Cov. \\
\midrule
\multicolumn{8}{l}{\emph{Electricity}} \\
Seasonal naive & 0.05 & 0.0 [$-$0.8, +0.7] & $-$0.01 & $-$1.5 [$-$2.9, $-$0.2] & $-$0.22 & $-$0.1 [$-$1.6, +1.4] & +0.04 \\
Rolling ridge & 0.02 & $-$0.1 [$-$1.0, +0.7] & 0.00 & 0.0 [$-$0.9, +0.9] & 0.00 & +2.4 [+1.3, +3.6] & +0.43 \\
Gradient boosting & 0.03 & +0.4 [$-$0.6, +1.2] & $-$0.01 & +0.5 [$-$0.3, +1.3] & $-$0.02 & +2.6 [+1.8, +3.7] & +0.48 \\
\multicolumn{8}{l}{\emph{Weather}} \\
Seasonal naive & 0.05 & $-$0.1 [$-$0.3, +0.1] & $-$0.02 & $-$4.4 [$-$6.0, $-$2.7] & +0.57 & $-$12.6 [$-$15.3, $-$9.5] & +1.33 \\
Rolling ridge & 0.015 & 0.0 [$-$0.6, +0.6] & $-$0.03 & 0.0 [$-$0.6, +0.6] & $-$0.03 & +2.4 [+1.2, +3.4] & +0.23 \\
Gradient boosting & 0.05 & $-$0.2 [$-$0.7, +0.3] & +0.01 & $-$1.4 [$-$2.5, $-$0.2] & $-$0.22 & $-$6.2 [$-$10.2, $-$2.3] & +0.01 \\
\multicolumn{8}{l}{\emph{Finance}} \\
Naive (lag 1) & 0.05 & +7.1 [+2.6, +12.1] & +0.48 & +3.7 [+1.1, +6.6] & $-$0.53 & +3.7 [$-$1.2, +8.8] & $-$0.95 \\
Rolling ridge & 0.05 & +7.7 [+3.7, +12.1] & $-$0.05 & +5.3 [+2.3, +8.2] & $-$0.42 & +4.2 [$-$1.2, +9.2] & $-$0.85 \\
Gradient boosting & 0.015 & +12.4 [+6.3, +17.8] & +0.90 & +7.0 [+2.3, +11.4] & +0.27 & +7.8 [+1.8, +13.2] & +0.42 \\
\bottomrule
\end{tabular}
\end{table}

%% file: s6_robustness.tex
\subsection{Robustness to the grids}
\label{sec:robust}

With every method searching one wider step-size set and the bandwidths scaled to the spread of the feature vectors, the comparison of DASC (projected) with ACI (projected) is robust, with weather cells at the boundary for the interval score. This is the verdict at the block length and the seed fixed in advance. At \exBlkMidElec{} hours the interval-score interval lies below zero in every seed in some of the cells listed in the second check below. The runs answer two concerns about the grids of Section~\ref{sec:design-tuning}: the bandwidths of the main runs were large against the spread of the feature vectors on the hourly streams, which gives nearly equal weights, so tuning could not choose localised weights there; and the step-size sets differed between methods, with the largest value often selected (Appendix~\ref{S-app:robust-details} gives the ratios and the counts).

\emph{Design.} Every method with a step size searched the same \exNGammaWord{} values, \exGammaList{}, and the bandwidth sets added the multiples \exScaleList{} of the median pairwise distance between the feature vectors at the development steps. Spectral ACI met the scaled bandwidths jointly with every other parameter, and DASC met them at one pool, one window and one $m_{\min}$, which errs against DASC, as the equal-budget study did. The wide design adds the main-run search before the two stages of DASC. Appendix~\ref{S-app:robust-details} gives the design as run.

\emph{Criterion.} The criterion was fixed in writing before the runs began, in the design file that accompanies the code. The design file is dated \exDesignDate{}. The statement that DASC (projected) does not beat ACI (projected) at an equal budget is called robust if, in every cell, the paired width interval for DASC minus ACI includes zero or lies above it and no paired interval for the interval-score difference against ACI lies below zero; otherwise it is not robust. The criterion is read at the block length of the main analysis and the pre-stated seed, and the sensitivity to both reported below is in addition. With \exNCellsWord{} cells and two intervals each at the level of the bootstrap interval (\dcBootLevel{} per cent), a near miss is not surprising.

\input{v2_ext_counts}

\emph{Against ACI the criterion holds in both designs} (Table~\ref{tab:v2-ext-counts}; the cell-by-cell values are in Table~\ref{S-tab:ext-real-paired}). In the primary design, at the block length of the main analysis, no width or interval-score interval lies below zero. The width intervals lie above zero in all three finance cells, where DASC is \exPrimAciWFinMin{} to \exPrimAciWFinMax{} per cent wider. In the other \exPrimAciWInclWord{} cells the width differences lie between \exPrimAciWNonFinMin{} and \exPrimAciWNonFinMax{} per cent and every interval includes zero. All \exPrimAciSInclWord{} interval-score intervals include zero. In the wide design no interval lies below zero either, and electricity with boosting becomes wider (\exWideAciElecBoostW{} per cent, \exWideAciElecBoostWci).

In \exNearNWord{} cases the interval-score interval includes zero by a narrow margin: weather with boosting in both designs, and weather with the seasonal naive forecaster in the wide design. Their upper limits are at most \exNearHiMax{}, against point estimates of \exPrimAciWeaBoostIs{} and \exWideAciWeaNaiveIs{}, which favour DASC. These endpoints depend on the bootstrap seed. A check outside the pre-stated design repeated the comparison with \exSeedNSeedsWord{} seeds: the width intervals lie below zero in none of the \exSeedRows{} draws, and the interval-score interval lies below zero in \exSeedSAdvTotal{} of them, all in these cases. The width half of the criterion is therefore stable. The verdict is robust for the procedure fixed in advance, but a different bootstrap seed would have returned not robust on the interval score in these cells, so they are not a clean tie. The effect is about \exNearRelNaive{} to \exNearRelBoost{} per cent of the score, with no gain in width.

A second check outside the pre-stated design used two longer block lengths (\exBlkMidElec{} and \exBlkLongElec{} hours on electricity and weather, \exBlkMidFin{} and \exBlkLongFin{} trading days on finance) with the same seeds, \exBlkNDraws{} draws per design. No width interval lies below zero. At \exBlkMidElec{} hours the interval-score interval lies below zero in every seed in cells where it includes zero at the main-run block length and seed: electricity with the ridge forecaster in the equal-budget study (Section~\ref{sec:results-equal}) and weather with boosting in both extended designs; in the wide design it does so in some of the seeds in weather with the seasonal naive forecaster. At \exBlkLongElec{} hours it adds electricity with the ridge forecaster in both extended designs. The effects are \exBlkEffMin{} to \exNearRelBoost{} per cent of the score, with no width or coverage gain. A resample of electricity at \exBlkLongElec{} hours holds about \exBlkElecBlocksWord{} blocks, so the interval there is not reliable, and the draws at that length are a sensitivity check and not a verdict. Whether the narrower intervals at the longer blocks are an artefact of the bootstrap was not tested.

\emph{Against spectral ACI no cell favours DASC in either design.} In the primary design DASC is wider, with an interval above zero, in \exPrimSpecWDisWord{} cells and has the higher interval score in one. Section~\ref{sec:results-equal} reported \exPriorSpecNarrowNWord{} cells in which tuned DASC was narrower than spectral ACI with an interval below zero; with the common step-size set none of them favours DASC (the values are in Appendix~\ref{S-app:robust-details}). Spectral ACI itself has a lower interval score than ACI, with an interval below zero, in weather with boosting and in weather with the seasonal naive forecaster (in the second the interval includes zero once the supplementary steps below are added). The scaled bandwidths gave spectral ACI a small gain in weather with boosting that DASC does not share.

\emph{Against decay-weighted conformal prediction} the picture is as before: DASC is narrower, with an interval below zero, in \exPrimDecayWAdvWord{} cells and wider, with an interval above zero, in \exPrimDecayWDisWord{}. The weather cell with boosting is a gain in width and interval score over decay at equal coverage (\exPrimDecayWeaBoostW{} per cent, \exPrimDecayWeaBoostWci, coverage difference \exPrimDecayWeaBoostCpt{} points). The criterion concerns ACI, and Section~\ref{sec:intro} already reports that against decay DASC is ahead in some cells and behind in others.

\emph{Step sizes at the edge of the set} (Table~\ref{S-tab:ext-edge}). With the common set the largest value is still selected, by ACI and by DASC (projected) in \exEdgeExtTopAciWord{} cells each. The comparison is symmetric there, because both methods share the set, but the optimum may lie beyond \exGammaTop{}, and the set was not widened for the criterion. A supplementary run outside the pre-stated design added the steps \exTopNewList{} in the \exTopCellsWord{} cells where a step-size method had chosen \exGammaTop{}; DASC (projected) minus ACI (projected) has no interval below zero in width or in interval score in those cells.

\input{v2_ext_sweep}

\emph{Bandwidths and the effective sample size} (Tables~\ref{tab:v2-ext-sweep} and~\ref{tab:v2-ext-neff}). At the main-run bandwidths the median effective sample size is a share of \exSweepHourOrigMin{} to \exSweepHourOrigMax{} of the largest pool on the hourly streams. The scaled set allows localised weights, and with a fixed level the development score chose them: the share is below one half in \exNeffSpecFixExtHourBelowHalfWord{} of the \exNHourCellsWord{} hourly cells for the spectral method and in \exNeffDascFixExtHourBelowHalfWord{} for DASC. When the level adapts, DASC (projected) selects a share of at least \exNearUniformThr{} in \exDascNearUniformNWord{} of the \exNCellsWord{} cells (on the hourly streams, spectral ACI and DASC keep at least \exNeffSpecAciExtHourMin{} and \exNeffDascExtHourMin{}). The study does not test why.

\input{v2_ext_neff}

\emph{The simulated design built for spectral weighting.} The extended step-size set was run in the frequency-and-noise scenario with the lag-one naive and the ridge forecaster. With the ridge forecaster the interval score of spectral ACI minus ACI is lower by more than two standard errors at a wider interval, and that of DASC (projected) minus ACI is higher by more than two standard errors: the spectral weights add a small gain to a tuned step, and DASC adds none. At an equal step the gap in mean score between ACI and spectral ACI shrinks as the step grows with both forecasters, which agrees with the weights acting as a substitute for a faster step (Section~\ref{sec:avg}); a supplementary run beyond \exGammaTop{} brackets the optimum, so the statement holds over the steps that contain it. No paired standard error is reported for these gaps. Appendix~\ref{S-app:robust-details} gives the selected steps and the values.

\emph{Limits.} The criterion is read for the projected updates, on one path per stream, at the block length of the main runs (\bootBlockElec{} hours for electricity and weather, \bootBlockFin{} trading days for finance); the check at longer blocks is outside the pre-stated design and does not settle which length is right. The effective sample sizes are medians over evaluation steps, and the sweep over the bandwidth set uses every fifth evaluation step, the full pool and no gate. The simulated design was re-run in one scenario with two forecasters, with \exSimNReps{} replications against \smNReps{} in the main study; \exSimNotRunWord{} scenarios were not run. No joint search over bandwidth and step size was run for DASC, which meets the scaled bandwidths only in its second stage. The intervals are conditional on the selected configurations.

%% file: v2_ext_counts.tex
\begin{table}[t]
\centering
\caption{Extended step-size set and scaled bandwidths: DASC (projected) minus each comparator in the nine real cells. For the width difference and for the interval-score difference, the table counts the cells in which the 95 per cent paired moving-block bootstrap interval lies below zero (DASC narrower, or lower interval score), includes zero, or lies above zero (DASC wider, or higher interval score). The last column is the range of the coverage difference in percentage points. The primary design is the one the criterion is read from; the wide design adds the main-run search at the default step size before the step-size and bandwidth stages. The counts are for the block lengths of the paired bootstrap (\bootBlockElec{} hours on electricity and weather, \bootBlockFin{} trading days on finance) and the bootstrap seed of the main analysis. At longer block lengths the interval-score interval against ACI lies below zero in some of the cells that are counted as including zero here, by \exBlkEffMin{} to \exBlkEffMax{} per cent of the score and with no gain in width or coverage, as the text reports. At \exBlkMidElec{} hours this holds in every seed for weather with boosting in both designs; electricity with the ridge forecaster appears only at \exBlkLongElec{} hours, a sensitivity check only.}
\label{tab:v2-ext-counts}
\small
\setlength{\tabcolsep}{4pt}
\begin{tabular}{llrrrrrrr}
\toprule
 & & \multicolumn{3}{c}{Width: interval} & \multicolumn{3}{c}{Interval score: interval} & Coverage \\
\cmidrule(lr){3-5}\cmidrule(lr){6-8}
Comparator & Design & below & includes & above & below & includes & above & difference (points) \\
\midrule
ACI (projected) & primary & 0 & 6 & 3 & 0 & 9 & 0 & $-$0.05 to +0.90 \\
 & wide & 0 & 7 & 2 & 0 & 9 & 0 & $-$0.43 to +0.96 \\
\addlinespace[2pt]
Spectral ACI & primary & 0 & 7 & 2 & 0 & 8 & 1 & $-$0.27 to +0.69 \\
 & wide & 0 & 6 & 3 & 0 & 9 & 0 & $-$0.03 to +1.06 \\
\addlinespace[2pt]
Decay-weighted & primary & 3 & 2 & 4 & 2 & 7 & 0 & $-$0.95 to +1.69 \\
 & wide & 3 & 2 & 4 & 3 & 5 & 1 & $-$0.58 to +1.27 \\
\bottomrule
\end{tabular}
\end{table}

%% file: v2_ext_sweep.tex
\begin{table}[t]
\centering
\caption{Share of the pool that the kernel weights use, by bandwidth: median effective sample size over every fifth evaluation step, as a share of the largest pool, computed from the features alone (no gate, no forecaster). The scaled bandwidths are multiples $0.25$, $0.5$, $1$ and $2$ of the median distance $m$ between feature vectors on the development steps; the last column is the three bandwidths of the main runs. Each entry is the range over the feature windows.}
\label{tab:v2-ext-sweep}
\small
\setlength{\tabcolsep}{5pt}
\begin{tabular}{lrrrrr}
\toprule
Data set & $0.25\,m$ & $0.5\,m$ & $1\,m$ & $2\,m$ & Main-run bandwidths \\
\midrule
Electricity & 0.02 to 0.03 & 0.14 to 0.15 & 0.46 to 0.50 & 0.86 to 0.87 & 0.98 to 1.00 \\
Weather & 0.13 to 0.17 & 0.36 to 0.43 & 0.68 to 0.73 & 0.92 to 0.93 & 1.00 \\
Finance & 0.02 to 0.03 & 0.20 to 0.22 & 0.76 to 0.78 & 0.97 to 0.98 & 0.33 to 0.97 \\
\bottomrule
\end{tabular}
\end{table}

%% file: v2_ext_neff.tex
\begin{table}[t]
\centering
\caption{Effective sample size at the selected configuration, as a share of the pool: median over evaluation steps in each cell, then the range over cells. ``Below one half'' counts cells in which the share is below $0.5$. ``Scaled'' is the number of the nine cells in which the selected bandwidth is a multiple of the median feature distance and not one of the three main-run values. Hourly: the six electricity and weather cells. Finance: the three finance cells.}
\label{tab:v2-ext-neff}
\small
\setlength{\tabcolsep}{4pt}
\begin{tabular}{llrrrrr}
\toprule
 & & \multicolumn{2}{c}{Hourly cells} & \multicolumn{2}{c}{Finance cells} & Scaled \\
\cmidrule(lr){3-4}\cmidrule(lr){5-6}
Method & Grids & range & below one half & range & below one half & (of nine) \\
\midrule
Spectral ACI & main grids & 1.00 to 1.00 & 0 & 0.34 to 0.97 & 1 & n/a \\
 & extended & 0.77 to 1.00 & 0 & 0.22 to 0.34 & 3 & 4 \\
\addlinespace[2pt]
DASC (projected) & main grids & 0.98 to 1.00 & 0 & 0.33 to 0.98 & 2 & n/a \\
 & extended & 0.87 to 1.00 & 0 & 0.33 to 0.92 & 2 & 2 \\
 & extended, wide design & 0.14 to 1.00 & 1 & 0.38 to 0.98 & 1 & 5 \\
\addlinespace[2pt]
Spectral, fixed level & main grids & 0.98 to 1.00 & 0 & 0.34 to 0.60 & 1 & n/a \\
 & extended & 0.05 to 0.94 & 5 & 0.04 to 0.23 & 3 & 9 \\
\addlinespace[2pt]
DASC, fixed level & main grids & 0.99 to 1.00 & 0 & 0.33 to 0.33 & 3 & n/a \\
 & extended & 0.04 to 1.00 & 4 & 0.33 to 0.33 & 3 & 4 \\
\bottomrule
\end{tabular}
\end{table}

%% file: s6_pitfalls.tex
\subsection{How large are the common evaluation pitfalls?}
\label{sec:pitfalls}

Of the evaluation choices measured here, centring the interval on a different forecast and freezing the pool move apparent width the most, by tens of per cent in the extreme cells. Centring changes width by \pfAElecW{} per cent on electricity at a coverage change of \pfAElecC{} points, which is close to matched, and by \pfAWeaW{} per cent on weather at \pfAWeaC{} points, which is not. A mismatch of step size changes width by up to \pfBWMax{} per cent. The finite-sample rank and the length of the pool move it by a few per cent or less once the level adapts. Table~\ref{tab:v2-pitfalls} gives the sizes, and Table~\ref{S-tab:pitfalls} the cell-by-cell values with an interval for each.

\emph{Design.} The reference is ACI (projected) at the pool and step size that the tuning rule of Section~\ref{sec:design-tuning} selects for each stream. Each pitfall changes one thing, re-tunes nothing and is measured alone. The effect is the pitfall minus the reference, as a per cent of the reference value. The data are the \pfNStreamsWord{} real cells of Section~\ref{sec:design-data}, with the paired moving-block bootstrap of Section~\ref{sec:design-measures}, and the \pfNScenWord{} simulated scenarios with the lag-one naive and the ridge forecaster, \pfNReps{} replications each (the main study uses \smNReps). The \pfNPitfallsWord{} pitfalls are:
\begin{enumerate}[label=(\alph*)]
\item centring: DASC as published, centred on the previous observation, against DASC centred on the forecast that the baselines use (the reproduction check of Section~\ref{sec:results-anchor}, reused);
\item unequal step size: step size \pfGammaAlt{} in place of the selected one, with the same pool;
\item finite-sample rank: the rank $\lceil (n+1)(1-\alpha_t)\rceil$ in place of $\lceil n(1-\alpha_t)\rceil$ for $n$ pool residuals;
\item frozen pool: the pool receives no residual after the first evaluation step, and the level update continues;
\item window mismatch: a pool half as long or twice as long as the selected one.
\end{enumerate}
Rows (c$'$) and (d$'$) of the table apply (c) and (d) to rolling conformal prediction, whose level does not adapt. Coverage is not forced to match, so the width changes in the cells with a dagger do not compare like with like. A dagger marks a cell, or a simulated median, in which the coverage changes by more than \pfCovTol{} points; those width changes are not at matched coverage and should be read with the interval score.

\input{v2_pitfalls}

\emph{Centring and freezing.} The centring row is one deterministic run, so it has no interval and no interval score. On weather the coverage change of \pfAWeaC{} points is outside the tolerance of \pfCovTol{} points, so the width change of \pfAWeaW{} per cent carries a dagger. It does not apply to finance, where every method uses the lag-one forecast, or to the simulations. Freezing the pool changes the weather widths by \pfDWeaMin{} to \pfDWeaMax{} per cent. With the seasonal naive forecaster the change is \pfDWeaNaiveW{} per cent (\pfDWeaNaiveWci), with a coverage change of \pfDWeaNaiveC{} points and an interval-score change of \pfDWeaNaiveIs{} per cent. On finance the width change is \pfDFinMin{} to \pfDFinMax{} per cent, and on electricity \pfDElecMin{} to \pfDElecMax{} per cent. For rolling conformal prediction the same change reaches \pfDFixWMin{} per cent with a coverage change of \pfDFixCMin{} points. A frozen pool narrows intervals mainly by losing coverage.

\emph{Step size.} Step size \pfGammaAlt{} in place of the selected one changes width by \pfBWMin{} to \pfBWMax{} per cent in the \pfBNAppWord{} cells where the two differ, with both signs. The largest change is weather with the seasonal naive forecaster: \pfBWeaNaiveW{} per cent (\pfBWeaNaiveWci), with an interval-score change of \pfBWeaNaiveIs{} per cent and a coverage change of \pfBWeaNaiveC{} points, which is outside the tolerance of \pfCovTol{} points, so this cell carries a dagger. It is close to the gap between DASC (projected) and ACI (projected) in that cell in the main runs, \prWeaNaiveDascAciWAbs{} per cent (Section~\ref{sec:results-real}). Of the \pfBNAppWord{} cells, \pfBNDagWord{} carry a dagger. In the other \pfBNNonDagWord{} the width changes by \pfBNonDagMin{} to \pfBNonDagMax{} per cent.

\emph{Rank and window.} The finite-sample rank changes width by \pfCWMin{} to \pfCWMax{} per cent in the real cells, with an interval that excludes zero only on finance, and by a median of \pfCSimNaiveWMed{} per cent (lag-one naive) and \pfCSimRidgeWMed{} per cent (ridge) in the simulations. Without the level update the same change is \pfCFixWMin{} to \pfCFixWMax{} per cent in the real cells and a median of \pfCFixSimNaiveWMed{} and \pfCFixSimRidgeWMed{} per cent in the simulations, so in the simulations the level update removes \pfCAbsorbMin{} to \pfCAbsorbMax{} per cent of the effect. A pool half as long changes width by \pfEHalfWMin{} to \pfEHalfWMax{} per cent, and a pool twice as long by \pfEDoubleWMin{} to \pfEDoubleWMax{} per cent.

\emph{Limits.} Interactions between pitfalls are not measured. The real data are one path per stream, and the intervals describe sampling variation within it. The simulations use \pfNReps{} replications, fewer than the main study, and the standard errors of the scenario means for width are at most \pfSimSeMax{} per cent (median \pfSimSeMed{}) in rows (b) to (e) and at most \pfSimSeFixMax{} per cent in rows (c$'$) and (d$'$). The step-size row uses one alternative value, and the freezing time in row (d) is a choice. Rows (c$'$) and (d$'$) use the tuned rolling window of this paper. The study is descriptive and has no simulation with gradient boosting.

%% file: v2_pitfalls.tex
\begin{table}[t]
\centering
\caption{Size of common evaluation pitfalls for adaptive conformal inference (ACI, projected). The reference is ACI at the pool and step size selected by the tuning rule, with the plain quantile rank and a rolling pool. Each row changes one thing and re-tunes nothing; each entry is the pitfall minus the reference, so a positive width change means that the pitfall is wider. Real data: the range over the cells (nine, fewer where the pitfall does not apply), as the minimum to the maximum of the cell values. Simulations: median over scenarios, for the lag-one naive forecaster and for ridge. The coverage change is in percentage points, the interval-score change in per cent of the reference score. $^{\dagger}$ The coverage change of that cell (simulations: the median over scenarios) exceeds $0.5$ points, so its width change is not at matched coverage and should be read with the interval score. The column ``Cells'' counts the real cells with a dagger. Rows (c$'$) and (d$'$) apply the same changes to rolling conformal prediction, whose level does not adapt. Row (a) reuses the reproduction check of DASC (published centring against the shared forecast); it has no interval and no interval score. The full table has an interval for every cell.}
\label{tab:v2-pitfalls}
\scriptsize
\setlength{\tabcolsep}{2pt}
\begin{tabular}{p{3.6cm}rlllll}
\toprule
 & & \multicolumn{3}{c}{Real data} & \multicolumn{2}{c}{Simulated width change (\%)} \\
\cmidrule(lr){3-5}\cmidrule(lr){6-7}
Pitfall & Cells$^{\dagger}$ & Width change (\%) & Coverage (points) & Interval score (\%) & Lag-one naive & Ridge \\
\midrule
(a) Interval centred on another forecast & 1 of 2 & $-$41.3$^{\dagger}$ to $-$27.9 & +0.33 to +6.96 & n/a & n/a & n/a \\
(b) Step size 0.015 in place of the selected one & 4 of 8 & $-$2.3 to +10.3$^{\dagger}$ & $-$0.60 to +1.06 & $-$0.8 to +19.8 & +0.6 & +0.6 \\
(c) Finite-sample rank & 0 of 9 & $-$0.3 to +1.4 & 0.00 to +0.37 & $-$2.4 to +0.3 & +0.2 & +0.1 \\
(d) Pool frozen at the first evaluation step & 3 of 9 & $-$26.8$^{\dagger}$ to $-$0.1 & $-$10.21 to +0.12 & $-$0.4 to +27.5 & $-$4.4$^{\dagger}$ & $-$0.6 \\
(e) Pool halved & 1 of 9 & $-$0.1 to +3.1 & $-$0.16 to +0.53 & $-$2.2 to +0.7 & +0.7 & +0.4 \\
(e) Pool doubled & 0 of 9 & $-$2.9 to +0.7 & $-$0.27 to +0.18 & 0.0 to +4.9 & $-$0.6 & $-$0.1 \\
(c$'$) Finite-sample rank, fixed level & 1 of 9 & +0.3 to +2.0 & +0.06 to +0.69 & $-$0.6 to 0.0 & +1.4 & +1.3 \\
(d$'$) Pool frozen, fixed level & 6 of 9 & $-$54.9$^{\dagger}$ to +16.2$^{\dagger}$ & $-$25.71 to +3.26 & +0.1 to +52.2 & $-$10.9$^{\dagger}$ & $-$5.3$^{\dagger}$ \\
\bottomrule
\end{tabular}
\end{table}

%% file: tab_timing.tex
\begin{table}[t]
\centering
\caption{Computation time per step in microseconds (single core, wall clock) for each method at each dataset's maximum pool size and feature window.}
\label{tab:timing}
\small
\begin{tabular}{lrrrr}
\toprule
Method & simulation (pool 360) & finance (pool 504) & electricity (pool 1440) & weather (pool 2160) \\
\midrule
Rolling conformal & 206 & 152 & 228 & 375 \\
ACI (projected) & 165 & 161 & 221 & 288 \\
ACI (free) & 138 & 155 & 242 & 311 \\
PID-style & 142 & 158 & 244 & 356 \\
Decay-weighted & 401 & 231 & 581 & 456 \\
Spectral (fixed level) & 612 & 441 & 953 & 1207 \\
Spectral ACI (no gate) & 487 & 437 & 1008 & 1770 \\
DASC (fixed level) & 369 & 462 & 669 & 1204 \\
DASC (projected) & 392 & 453 & 716 & 947 \\
DASC (free) & 401 & 470 & 677 & 1012 \\
\bottomrule
\end{tabular}
\end{table}

%% file: s7_discussion.tex
\section{Discussion and limitations}
\label{sec:discussion}

\paragraph{What the evidence shows.} With one forecaster per data set, one tuning criterion and matched evaluation steps, the adaptive methods keep coverage near the target, with some under-coverage at the selected settings. With the projected update, DASC has no width advantage over ACI that survives a matched or a tuned step size; this holds in both step-size designs, the main runs and the equal-budget study. In the equal-budget study and in the extended designs, the width intervals against ACI keep this pattern at longer block lengths and with other bootstrap seeds; that check did not cover the main runs (Section~\ref{sec:robust}). At DASC's step size and pool the widths of the projected variants differ by at most \agRelWidMaxNonFin{} per cent on electricity and weather. At the equal budget the paired width intervals against ACI include zero or lie above it. For the projected update at that budget, and at the block length and the bootstrap seed of the main analysis, no paired interval for the interval-score difference against ACI excludes zero in DASC's favour on the three real streams. At \exBlkMidElec{} hours such an interval lies below zero in every seed in electricity with the ridge forecaster (the equal-budget study) and in weather with gradient boosting (the extended designs), and in the wide design in some of the seeds in weather with the seasonal naive forecaster. The differences are at most \exNearRelBoost{} per cent of the score, with no gain in width or coverage (Section~\ref{sec:robust}). The draws at \exBlkLongElec{} hours are a sensitivity check only. For the free update on finance with the naive forecaster, the interval-score difference against free ACI is \prFinNaiveFreeFreeS{} (\prFinNaiveFreeFreeSci), which excludes zero (Section~\ref{sec:results-real}).

At the equal budget in the simulations, DASC's interval score is better than ACI's by more than two standard errors in \tgNAciIsBetter{} of the \smNCellsWord{} simulated cells. With the free update, DASC (free) is narrower than free ACI by more than two standard errors in \tgNFreeAciNarrow{} simulated cells at tuned step sizes. Section~\ref{sec:results-equal} gives the counts of these cells and of the cells against spectral ACI and decay-weighted conformal prediction; Table~\ref{tab:equal-budget} and Tables~\ref{S-tab:dascgamma-sim-paired-naive1} and~\ref{S-tab:dascgamma-sim-paired-ridge} give the cell-by-cell values. Against spectral ACI, DASC is ahead in some cells with the main-run step-size sets and in none when every method searches the same step sizes (Section~\ref{sec:robust}), also at longer block lengths and with other bootstrap seeds (Section~\ref{sec:design-measures}).

In the cells examined (Sections~\ref{sec:results-real} and~\ref{sec:results-equal}), where DASC is narrower in a real-data cell, the cause is consistent with a difference in step size or in coverage in all cells but one, and the table or section that shows it is named in the text. Against spectral ACI, DASC is narrower on electricity with the naive forecaster and on weather with gradient boosting, at a coverage difference of \ptElecNaiveTgSpecCpt{} points in the first cell and \ptWeaBoostTgSpecCpt{} points in the second. Both are inside the tolerance of \pfCovTol{} points of Section~\ref{sec:pitfalls}, so coverage does not account for the difference in width. The step-size set does: DASC has no width advantage in either cell when both methods search the same step sizes (Section~\ref{sec:robust}). The exception is weather with gradient boosting against decay-weighted conformal prediction, where DASC is \ptWeaBoostTgDecayWAbs{} per cent narrower (\ptWeaBoostTgDecayWci) at equal coverage with a lower interval score (Section~\ref{sec:results-equal}); in the extended design the width difference in the same cell is \exPrimDecayWeaBoostW{} per cent (\exPrimDecayWeaBoostWci) at a coverage difference of \exPrimDecayWeaBoostCpt{} points. The simulated cells are given as counts. In the main runs the kernel weights have an effect in one design only, a simulated one built for them, and there a tuned step size for ACI reaches a similar balance between regimes. With bandwidths scaled to the feature spread and one step-size set for every method, spectral ACI has a lower interval score than ACI in weather with gradient boosting (\exSaciAciWeaBoostIs{}, \exSaciAciWeaBoostIci), and DASC does not share the gain (Section~\ref{sec:robust}).

\paragraph{Limits of the evidence.}
\begin{itemize}
\item \emph{One path per real data set.} The bootstrap intervals describe the sampling variation of one electricity path, one temperature path and one return series. The three series are not a sample of application domains. A different series, or a different stretch of the same one, can give a different ranking.
\item \emph{Bootstrap block length and seed.} The paired intervals use blocks of \bootBlockElec{} hours (electricity and weather) or \bootBlockFin{} trading days (finance) and one bootstrap seed. In the equal-budget study and the extended designs, longer blocks and other seeds move no width interval against ACI below zero. They move the interval-score interval below zero in the cells named above (Section~\ref{sec:robust}). Which block length is right is not settled. The check did not cover the main runs or the free update.
\item \emph{Step-size designs.} In the main runs the step size of ACI is tuned over \dcGammaGridWord{} values while DASC's is fixed at the default of the earlier preprint. The equal-budget study tunes DASC's step size over the ACI grid and re-tunes its bandwidth and threshold, on the same development prefix. A short development prefix selects the step size with noise. On electricity with the ridge forecaster the tuned step is worse on the evaluation period than the fixed one. The tuning rule for the free update, which gives an infinite development score to any configuration with a whole-line step, favours small step sizes, so free-against-free is a slow-step comparison. Neither design is the correct one, and both are reported. With one step-size set of \exNGammaWord{} values for every method, and with a second design that adds the main-run search before the two stages of DASC, the width comparison with ACI does not change and the interval-score comparison has weather cells at the boundary (Section~\ref{sec:robust}). The largest value of the set is selected in \exEdgeExtTopAciWord{} to \exEdgeExtTopSpecAciWord{} of the \exNCellsWord{} cells by the step-size methods, and the set was not widened.
\item \emph{Bandwidth grid.} The bandwidths of the main runs are \exOrigRatioElecMin{} to \exOrigRatioElecMax{} median feature distances on electricity and \exOrigRatioWeaMin{} to \exOrigRatioWeaMax{} on weather, so tuning could not choose localised weights there. The scaled grid of Section~\ref{sec:robust} allows localised weights on the hourly streams; the methods with a fixed level chose them in most hourly cells, and the methods with an adaptive level mostly kept weights close to uniform.
\item \emph{The tuning criterion.} The mean interval score on the development data rewards narrow intervals that cover and penalises misses in proportion to their distance. Methods end up with different coverage at their selected settings (for example projected ACI on finance, decay-weighted conformal prediction on weather, and spectral ACI on weather with the naive forecaster), so a width difference is read together with the coverage difference. A different criterion, such as the narrowest width subject to a coverage constraint, can select other configurations.
\item \emph{Simulation design.} The scenarios are chosen by the authors. One is built for spectral weighting, and in it the weights help relative to ACI at an equal slow step size. A design with stronger spectral structure, a forecaster that leaves the structure in the residuals, and a step size that cannot be tuned could favour DASC. The study does not cover such a design.
\item \emph{Forecasters.} Three forecasters per real stream is a small set. Gradient-boosting forecasts were fitted with one thread, so a multi-threaded fit may not reproduce those rows bit for bit.
\item \emph{Theory.} The per-step bound holds under Assumption~\ref{ass:A}, whose part (i) cannot hold exactly for the forecasters used here. The bound is at level $1-\alpha_t$ and is loose. The refutation diagnostic uses concentration terms that assume independence given the features. Serially dependent residuals violate this, so on real data the diagnostic is indicative; it refutes a pair of constants without saying which part of the assumption fails. The blocked bound is vacuous at the pool sizes used, and its numerical values are the authors' arithmetic. The averaged-recursion account is for the free update in an idealised model and carries the status labels of Section~\ref{sec:avg}.
\item \emph{Cost.} The timings come from a shared machine, exclude the feature transform, and use an implementation that recomputes the weights over the whole pool at every step.
\item \emph{Data.} The electricity rows are consecutive non-missing hourly rows, with two gaps inside the series. The seasonal forecast is misaligned for a day after the longer gap.
\item \emph{Authorship of the method.} The authors proposed DASC and spectral ACI. The design limits favourable choices (Section~\ref{sec:design-aims}), and an independent replication would reduce the remaining concern.
\end{itemize}

\paragraph{What the components are and are not.} The kernel weights are a similarity-based localisation of the calibration pool. The band features do not see the position in the daily cycle (Remark~\ref{rem:phase}); a calendar-hour kernel gave the largest reduction of the hour-of-day deviation in \phCalLowestNWord{} of the \phNCellsWord{} cells. On electricity and weather the weights are close to uniform at the tuned bandwidths of the main runs, and with bandwidths scaled to the feature spread they mostly stay close to uniform for the methods with an adaptive level on the hourly streams (Section~\ref{sec:robust}). On finance the weights are concentrated and DASC is wider. The gate is a recency truncation with a small effect of both signs. A possible explanation, not tested here, is that it helps after a break and discards same-regime residuals under recurrence. The drift score needs no outcomes and is a weaker break detector than a rolling coverage deficit. The free level update carries the long-run guarantee, and it carries it equally for equal weights and for spectral weights (Theorem~\ref{thm:free}). For the projected update, if the floor never binds the long-run miss rate is at most $\alpha+(\alpha_{\max}-\alpha_{\min})/(\gamma T)$, and if the cap never binds it is at least $\alpha-(\alpha_{\max}-\alpha_{\min})/(\gamma T)$; no two-sided bound of order $1/T$ holds in general (Theorem~\ref{thm:projected}).

\paragraph{What would change the conclusion.} A real or simulated stream in which regimes recur with distinct spectra, residual scales differ between them, and ACI's step size cannot adapt fast enough; or a run in which DASC has an interval-score advantage over ACI whose paired interval excludes zero at equal tuning budget. The first does not occur in the evidence reported here. The second does not occur on the three real streams at the block length and the bootstrap seed of the main analysis. At \exBlkMidElec{} hours it occurs in every seed in electricity with the ridge forecaster (the equal-budget study) and in weather with gradient boosting (the extended designs), with differences of at most \exNearRelBoost{} per cent of the score and no gain in width or coverage (Section~\ref{sec:robust}). These cells qualify the interval-score statement for the real streams with a block-length condition, because the criterion of Section~\ref{sec:robust} is read at the block length and the bootstrap seed fixed in advance, and they do not change the width statement. A bandwidth scaled to the spread of the feature vectors, in place of the \exNHOrigWord{} values of the main runs, did not produce such a run on the real streams at the pre-stated block length and bootstrap seed. In two weather cells the interval-score limit is within \exNearHiMax{} of zero, and other seeds give an interval below zero there (Section~\ref{sec:robust}). In the simulations the second occurs in \tgNAciIsBetter{} of the \smNCellsWord{} cells (Section~\ref{sec:results-equal}), the heavy-tail stress scenario with the ridge forecaster, where the paired interval-score difference is just over two standard errors (Table~\ref{S-tab:dascgamma-sim-paired-ridge}). If the \smNCellsWord{} cells were independent, chance alone would pass about one cell across both directions at that criterion (Section~\ref{sec:results-equal}). In the same comparison, with the same criterion of a paired difference of more than two standard errors, DASC is narrower than ACI in \tgNAciNarrow{} cells and wider in \tgNAciWide{}. The cell of the heavy-tail stress scenario with the ridge forecaster meets the criterion of more than two standard errors for an interval-score advantage of DASC over ACI, and it does not change the width statement.

%% file: s8_conclusion.tex
\section{Conclusion}
\label{sec:conclusion}

When all methods share one forecaster, one tuning criterion and the same evaluation steps, DASC with the projected update has no width advantage over adaptive conformal inference (ACI) that survives a matched or a tuned step size, and the adaptive methods reach coverage near the target (Sections~\ref{sec:intro} and~\ref{sec:results}). The width gains of \rpGainElec{} and \rpGainWea{} per cent reported for DASC in an earlier preprint are the effect of centring DASC on a different point forecast from the baselines. In the cells examined (Section~\ref{sec:results-equal}), where DASC is narrower in a real-data cell, the cause is consistent with a difference in step size or in coverage in all cells but one, which is weather with gradient boosting against decay-weighted conformal prediction; in the two cells against spectral ACI, where coverage is within \pfCovTol{} points, the step-size set accounts for the difference in width (Sections~\ref{sec:robust} and~\ref{sec:pitfalls}). In the one simulated design built for spectral weighting, the weights act as a substitute for a faster step size. On electricity and weather the kernel weights are close to uniform at the tuned bandwidths of this study, also when the bandwidths are scaled to the feature spread, for the methods with an adaptive level (Section~\ref{sec:robust}). The gate has a small effect with both signs, and the drift score is a weaker break detector than a rolling coverage deficit. The width comparison with ACI is unchanged when every method searches the same wider step-size set, when the bandwidths are scaled to the feature spread, and at longer block lengths and with other bootstrap seeds. The interval-score comparison has weather cells at the boundary under the extended grids, and at \exBlkMidElec{} hours the interval-score interval lies below zero in every seed in electricity with the ridge forecaster (the equal-budget study) and in weather with gradient boosting (the extended designs), with small differences and no gain in width or coverage; the draws at \exBlkLongElec{} hours are a sensitivity check only (Section~\ref{sec:robust}). Measured evaluation pitfalls move width by tens of per cent in the extreme cells, at a different coverage (Sections~\ref{sec:robust} and~\ref{sec:pitfalls}). Section~\ref{sec:discussion} states the limits of the evidence, including examples of under-coverage, and what would change the conclusion.

The theory now rests on a measurability lemma, a per-step bound at the adaptive level under an explicit regime model, the identity of the free update, the identity and sign structure of the projected update with a counterexample to a two-sided $O(1/T)$ bound, a refutation diagnostic for the smoothness assumption, an averaged-recursion account of why a faster step can stand in for localisation (free update, idealised pool; statements labelled exact, ODE level or numerical), and a blocked concentration bound that is correct and vacuous at the pool sizes used. A perturbation test of the code ran \caChecks{} checks with \caFail{} failures.

The code and the saved results are public and are described in the Reproducibility section.

%% file: reproducibility.tex
The code and the saved results are public in the repository \url{https://github.com/jopoku16/dasc-controlled-evaluation}, in the directories \path{code} and \path{results}, and version v1.0.0 is archived on Zenodo with the version DOI \url{https://doi.org/10.5281/zenodo.24011169}. The archive holds no article text, tables, figures or raw data. The code runs under Python \dcPython{} with numpy \dcNumpy, pandas \dcPandas, SciPy \dcScipy, scikit-learn \dcSklearn{}, xgboost \dcXgboost{} and pyarrow; the exact package versions are pinned in \path{requirements.txt}.

The data are not in the archive. The three public data sets of Section~\ref{sec:design-data} are fetched with the download scripts in \path{code} (\path{download_real_data.py}, \path{download_weather_data.py} and \path{download_finance_data.py}), and the environment variable \path{DASC_DATA_ROOT} gives the data directory. The checksums of the files behind the stored results are in \path{DATA_SHA256.txt}; a fresh download can differ (the FRED series serves only a recent window), so exact reproduction needs the identical files.

The scripts run from \path{code} in this order. \path{test_causality.py} runs the perturbation test. \path{check_reproduce_original.py} produces the reproduction anchor. \path{run_real_data.py} tunes and evaluates every method on the real streams. \path{run_simulations.py} does the same for the simulations and runs the drift-detection check. \path{run_sim_gamma_sensitivity.py} runs the equal-step-size simulations. \path{run_assumption_diagnostic.py} runs the refutation test. \path{run_assumption_diagnostic_bonf.py} evaluates the same diagnostic at the level $\eta/J$ for each cluster (Section~\ref{sec:results-assumption}); its output \path{assumption_diagnostic_bonf_xgb1thread.csv} holds the values reported in the paper. \path{run_timing.py} measures the cost per step. \path{run_all.py} runs this sequence; the recorded full run took \dcRunTime. Outside that sequence, the equal-budget study is \path{run_dascgamma.py}, the extended grids are \path{ext_run_real.py}, \path{ext_run_sim.py} and \path{ext_combine.py}, the block length and bootstrap seed study is \path{blockseed_run.py}, and the evaluation pitfalls are \path{pitfalls_run.py}. The README of the archive lists, for each study including the exploratory ones, the scripts to run and the result files they write.

All seeds are fixed in the scripts: simulation evaluation seeds \dcEvalSeeds, development seeds \dcTuneSeeds, bootstrap seed zero and gradient-boosting seed zero. Gradient boosting runs with one thread in the shipped code. Two stored files, \path{causality_test_report.csv} and \path{assumption_diagnostic.csv}, were written before that setting; a rerun reproduces them except in the gradient-boosting rows.

The numerical checks of the theory run with \path{code/theory_checks/run_all_checks.py}; those on the level paths of the earlier preprint need that preprint's result files, and the outputs of the recorded runs are included. In Section~\ref{sec:mixing}, \path{check_mixing_bound.py} computes the arithmetic of the earlier preprint's remark, the comparison on an $m$-dependent sequence and the values for the constant $2$; the values for the constant $4$ come from the same calculation with the larger constant and are not in the archive.

The numbers, tables and figures of the paper are generated by scripts that read the saved result files and the recorded outputs of the theory checks; these scripts are not part of the archive. A reader can check the reported results of the experiments against the saved result files in \path{results}. The per-step diagnostics (level, quantile, width, score, effective sample size, projection residual and further quantities) are saved in \path{results/diagnostics} for every run on the real data and for the first seeds of each simulated run.

%% file: appendix_short.tex
\section{Proofs}
\label{app:proofs}

This appendix gives the proofs of the statements of Section~\ref{sec:theory}. Lemma~\ref{lem:measurable} is proved in the text. The proofs follow the order of the statements there, so each uses only earlier results: the theorem on the per-step bound uses the two lemmas on the weighted pool, and the proposition below that combines the per-step and long-run statements uses that theorem and the two long-run theorems.

\begin{proof}[Proof of Lemma~\ref{lem:bias}]
The weights are non-negative and sum to one, and part (ii) of the assumption bounds each term, so $\abs{\Fbar_{w,t}(r)-F(r\mid S_t)}=\abs{\sum_iw_{i,t}\{F(r\mid S_i)-F(r\mid S_t)\}}\le\sum_iw_{i,t}(L\norm{S_i-S_t}+\delta)=L\bar d_t+\delta$.
\end{proof}

\begin{proof}[Proof of Lemma~\ref{lem:grid}]
For $k=1,\ldots,K-1$ let $r_k=\inf\{r:\Fbar(r)\ge k/K\}$. Each $r_k$ is finite because $\Fbar$ increases from $0$ to $1$. Write $g(r^-)$ for a left limit. By right-continuity, $\Fbar(r_k^-)\le k/K\le\Fbar(r_k)$. Let
\[
\Delta=\max_{1\le k\le K-1}\max\bigl\{\abs{\Fhat(r_k)-\Fbar(r_k)},\ \abs{\Fhat(r_k^-)-\Fbar(r_k^-)}\bigr\}.
\]
This is the largest error of the weighted empirical distribution function at the grid points, counting left limits as well. The proof has two parts: a deterministic bound on the supremum in terms of this grid error, and a probability bound on the grid error. We show $\sup_r\abs{\Fhat(r)-\Fbar(r)}\le\Delta+1/K$ by splitting $\R$ into $I_0=(-\infty,r_1)$, $I_k=[r_k,r_{k+1})$ for $1\le k\le K-2$, and $I_{K-1}=[r_{K-1},\infty)$. Both $\Fhat$ and $\Fbar$ are non-decreasing. On $I_k$ with $1\le k\le K-2$ we have $\Fbar(r_{k+1}^-)-\Fbar(r_k)\le (k+1)/K-k/K=1/K$, so
\[
\Fhat(r)-\Fbar(r)\le\Fhat(r_{k+1}^-)-\Fbar(r_k)\le\Delta+\tfrac1K,\qquad
\Fhat(r)-\Fbar(r)\ge\Fhat(r_k)-\Fbar(r_{k+1}^-)\ge-\Delta-\tfrac1K .
\]
On $I_0$, $0\le\Fbar(r)\le\Fbar(r_1^-)\le 1/K$ and $0\le\Fhat(r)\le\Fhat(r_1^-)\le\Fbar(r_1^-)+\Delta$, so $-1/K\le\Fhat(r)-\Fbar(r)\le\Delta+1/K$. On $I_{K-1}$, $\Fbar(r)\ge\Fbar(r_{K-1})\ge (K-1)/K$, so $\Fhat(r)-\Fbar(r)\le 1-\Fbar(r)\le 1/K$ and $\Fhat(r)-\Fbar(r)\ge\Fhat(r_{K-1})-1\ge-\Delta-1/K$. This proves the deterministic bound on the supremum.

It remains to bound the grid error in probability. The grid is deterministic, because it depends only on the $F_i$ and the weights, so a bound for a fixed point applies at every grid point. For a fixed $r$, $\Fhat(r)-\Fbar(r)=\sum_iw_i(\ind\{R_i\le r\}-F_i(r))$ is a sum of independent terms, the $i$th of which lies in an interval of length $w_i$. By Hoeffding's inequality, $\Prob\{\abs{\Fhat(r)-\Fbar(r)}\ge u\}\le 2\exp(-2u^2/\sum_iw_i^2)=2\exp(-2u^2\neff)$, and the same holds for the left limits with $\ind\{R_i<r\}$ and $F_i(r^-)$. A union bound over the $2(K-1)$ statistics in $\Delta$ gives $\Prob\{\Delta\ge u\}\le 4(K-1)\exp(-2u^2\neff)$. With $u=\sqrt{\log(4K/\eta)/(2\neff)}$ the right side is $(K-1)\eta/K<\eta$.
\end{proof}

\begin{proof}[Proof of Theorem~\ref{thm:perstep}]
Both bounds combine the uniform bound of the grid lemma, applied at a quantile that depends on the residuals, with the bias lemma. Condition on $\cG_t$. The gated pool and the weights $w_{i,t}$ are then fixed, and by Assumption~\ref{ass:A}(i) the pool residuals are independent with distribution functions $F(\cdot\mid S_i)$. Lemma~\ref{lem:grid} with these weights gives the event
\[
\mathcal B_t=\Bigl\{\sup_r\abs{\Fhat_{w,t}(r)-\Fbar_{w,t}(r)}\le\varepsilon_t(\eta)\Bigr\},\qquad\Prob(\mathcal B_t\mid\cG_t)\ge 1-\eta,
\]
which belongs to $\cG_t\vee\sigma(R_i:i\in\cP^{g}_t)\subseteq\cK_{t-1}$. Write $q_t=q_t(\alpha_t)$. By Lemma~\ref{lem:measurable} and the definition of $\cK_{t-1}$, $q_t$ is $\cK_{t-1}$-measurable. By Assumption~\ref{ass:A}(i), $R_t$ is independent of $(R_1,\ldots,R_{t-1})$ given $\cG_t$ with law $F(\cdot\mid S_t)$, hence $\Prob\{R_t\le q_t\mid\cK_{t-1}\}=F(q_t\mid S_t)$.

\emph{Lower bound.} At the endpoints $\alpha_t=0$ and $\alpha_t=1$ the conventions give $q_t=\pm\infty$, the coverage probability is $1$ or $0$, and both bounds hold trivially, so take $\alpha_t\in(0,1)$. The definition of the quantile gives $\Fhat_{w,t}(q_t)\ge 1-\alpha_t$. On $\mathcal B_t$, $\Fbar_{w,t}(q_t)\ge\Fhat_{w,t}(q_t)-\varepsilon_t(\eta)$, and by Lemma~\ref{lem:bias}, $F(q_t\mid S_t)\ge\Fbar_{w,t}(q_t)-L\bar d_t-\delta$. Combining the three inequalities gives \eqref{eq:perstep-lower}. The bound of Lemma~\ref{lem:grid} is uniform in $r$, so it applies at the residual-dependent point $q_t$ at no extra cost; this is where the adaptive level enters.

\emph{Upper bound.} For $r<q_t$ we have $\Fhat_{w,t}(r)<1-\alpha_t$, so $\Fhat_{w,t}(q_t^-)\le 1-\alpha_t$ and $\Fhat_{w,t}(q_t)=\Fhat_{w,t}(q_t^-)+\pi_t\le 1-\alpha_t+\pi_t$. On $\mathcal B_t$ and by Lemma~\ref{lem:bias}, $F(q_t\mid S_t)\le\Fbar_{w,t}(q_t)+L\bar d_t+\delta\le\Fhat_{w,t}(q_t)+\varepsilon_t(\eta)+L\bar d_t+\delta$, which gives \eqref{eq:perstep-upper}. When the pool residuals are distinct, at most one residual equals $q_t$, so $\pi_t\le\max_iw_{i,t}$.
\end{proof}

\begin{proof}[Proof of Theorem~\ref{thm:free}]
Summing $\alpha_{s+1}-\alpha_s=\gamma(\alpha-E_s)$ over $s=1,\ldots,T$ gives $\alpha_{T+1}-\alpha_1=\gamma T\alpha-\gamma\sum_{t\le T}E_t$, which is \eqref{eq:free-identity}. For the range, if $\alpha_t<0$ then $E_t=0$ and $\alpha_{t+1}=\alpha_t+\gamma\alpha>\alpha_t$; if $\alpha_t\ge 0$ then $\alpha_{t+1}\ge\alpha_t-\gamma(1-\alpha)\ge-\gamma$. Hence $\alpha_t\ge-\gamma$ for all $t$ by induction. Symmetrically, if $\alpha_t>1$ then $E_t=1$ and $\alpha_{t+1}=\alpha_t-\gamma(1-\alpha)<\alpha_t$, and if $\alpha_t\le 1$ then $\alpha_{t+1}\le 1+\gamma\alpha\le 1+\gamma$. With this range, $\abs{\alpha_1-\alpha_{T+1}}\le\max\{\alpha_1+\gamma,\ 1+\gamma-\alpha_1\}$, and \eqref{eq:free-identity} then gives \eqref{eq:free-bound}. This is the argument of \citet{gibbs2021adaptive}. Their convention differs from ours only at $\alpha_t=0$: they use the finite quantile at level $1$ there, and here the interval is the whole line. The argument does not use the convention at that point. The sibling paper restates it with the bound $(1+\gamma)/(\gamma T)$.
\end{proof}

\begin{proof}[Proof of Theorem~\ref{thm:projected}]
(a) If $\epsilon_t<0$ then $U_t>\alpha_{\max}\ge\alpha_t$, so $\gamma(\alpha-E_t)>0$ and $E_t=0$; then $U_t=\alpha_t+\gamma\alpha\le\alpha_{\max}+\gamma\alpha$ and $\epsilon_t=\alpha_{\max}-U_t\ge-\gamma\alpha$. If $\epsilon_t>0$ then $U_t<\alpha_{\min}\le\alpha_t$, so $E_t=1$, $U_t=\alpha_t-\gamma(1-\alpha)\ge\alpha_{\min}-\gamma(1-\alpha)$ and $\epsilon_t=\alpha_{\min}-U_t\le\gamma(1-\alpha)$. (b) $\alpha_{s+1}-\alpha_s=\gamma(\alpha-E_s)+\epsilon_s$; sum over $s\le T$ and rearrange. (c) $\abs{\alpha_1-\alpha_{T+1}}\le\alpha_{\max}-\alpha_{\min}$ because both lie in the interval, and by (a) $-\gamma\alpha N^+_T\le\sum_t\epsilon_t\le\gamma(1-\alpha)N^-_T$. For the counterexample in the last sentence of the statement, on a constant stream with the lag-one forecaster $R_t=0=q_t$ for all $t$ after the first step, so every interval covers: $E_t\equiv 0$, $\alpha_t$ increases by $\gamma\alpha$ per step until it reaches $\alpha_{\max}$ and stays there, $N^+_T=T-O(1)$, and the miss rate is $0$.
\end{proof}

\begin{proposition}[What the per-step and long-run statements give together]
\label{prop:combine}
Let $p_t=\Prob\{E_t=1\mid\cK_{t-1}\}$ and $\Lambda_t(\eta)=L\bar d_t+\delta+\varepsilon_t(\eta)$. Under Assumption~\ref{ass:A}, for every $T$ and $\eta\in(0,1)$:
\begin{enumerate}[label=(\alph*)]
\item with the free update,
\[
\E\Bigl[\frac1T\sum_{t\le T}\alpha_t\Bigr]\ge\alpha-\E\Bigl[\frac1T\sum_{t\le T}\Lambda_t(\eta)\Bigr]-\eta-\gamma\,\E\Bigl[\frac{N^{\mathrm{neg}}_T}{T}\Bigr]-\frac{\max\{\alpha_1,1-\alpha_1\}+\gamma}{\gamma T},
\]
where $N^{\mathrm{neg}}_T=\#\{t\le T:\alpha_t<0\}$;
\item with the projected update the same holds with the last two terms replaced by $(\alpha_{\max}-\alpha_{\min})/(\gamma T)+\alpha\,\E[N^+_T/T]$.
\end{enumerate}
\end{proposition}

\begin{proof}
The proof bounds the conditional miss probability at each step by the level plus the losses, then averages over the steps and applies the long-run statement of each update. For $\alpha_t\in[0,1]$, \eqref{eq:perstep-lower} gives $p_t\le\alpha_t+\Lambda_t(\eta)$ on $\mathcal B_t$, and off it $p_t\le 1\le\alpha_t+\Lambda_t(\eta)+\ind\{\mathcal B_t^c\}$. For $\alpha_t<0$, $p_t=0\le\alpha_t+\gamma$ by Theorem~\ref{thm:free}. For $\alpha_t>1$, $p_t=1\le\alpha_t$. Hence $p_t\le\alpha_t+\Lambda_t(\eta)+\ind\{\mathcal B_t^c\}+\gamma\ind\{\alpha_t<0\}$ for every $t$. Take expectations and average over $t\le T$, using $\E[\ind\{\mathcal B_t^c\}]=\E[\Prob(\mathcal B_t^c\mid\cG_t)]\le\eta$ and $\E[p_t]=\E[E_t]$. For (a), \eqref{eq:free-identity} and the range of $\alpha_{T+1}$ give $\E[T^{-1}\sum_tE_t]\ge\alpha-(\max\{\alpha_1,1-\alpha_1\}+\gamma)/(\gamma T)$. For (b), the left inequality in \eqref{eq:proj-bounds} gives $\E[T^{-1}\sum_tE_t]\ge\alpha-(\alpha_{\max}-\alpha_{\min})/(\gamma T)-\alpha\E[N^+_T/T]$, and $N^{\mathrm{neg}}_T=0$.
\end{proof}

\begin{proof}[Proof of Lemma~\ref{lem:drift}]
$D_t\le\norm{\bar S^{\mathrm{rec}}_t-S_t}+\norm{S_t-\bar S^{\mathrm{cal}}_t}$, where $\bar S^{\mathrm{cal}}_t=\sum_iw^{\mathrm{pre}}_{i,t}S_i$, and $\norm{S_t-\bar S^{\mathrm{cal}}_t}=\norm{\sum_iw^{\mathrm{pre}}_{i,t}(S_t-S_i)}\le\sum_iw^{\mathrm{pre}}_{i,t}\norm{S_t-S_i}$ by the triangle inequality and convexity of the weights.
\end{proof}

\begin{proof}[Proof of Proposition~\ref{prop:check}]
$\abs{F_A(r)-F_B(r)}=\abs{(\abs A\abs B)^{-1}\sum_{i\in A}\sum_{j\in B}\{F(r\mid S_i)-F(r\mid S_j)\}}\le(\abs A\abs B)^{-1}\sum_{i,j}(L\norm{S_i-S_j}+\delta)$, and the right side is the bound in the statement, which proves the deterministic claim. For the probability claim apply Lemma~\ref{lem:grid} with equal weights to $A$ and to $B$, with the number of grid cells equal to the size of the set (given $\cG$ the residuals in $A\cup B$ are independent), take the union of the two failure events, and combine the result with the deterministic claim by the triangle inequality. The statement after Definition~\ref{def:Lmin} follows by a union bound over the index sets. Let $A_1,\ldots,A_J$ be the $J$ index sets of the definition. As in the proposition, they are $\cG$-measurable, they have at least two elements each, and the two sets of a pair are disjoint. Given $\cG$, the lemma applied to $A_j$ as above shows that the event $\{\sup_r\abs{\Fhat_{A_j}(r)-F_{A_j}(r)}>\varepsilon_{A_j}(\eta)\}$ has probability at most $\eta$, so the union of these $J$ events has probability at most $J\eta$. Outside the union, every pair $(A,B)$ satisfies the inequality of the proposition by the argument above, so every term of the maximum in Definition~\ref{def:Lmin} is at most $L$, and $\widehat L_{\min}(\delta,\eta)\le L$ with probability at least $1-J\eta$.
\end{proof}

\begin{proof}[Proof of Proposition~\ref{prop:blocked}]
The proof splits the weighted sum into the sums over odd and even blocks, decouples the blocks of each parity with the lemma on block coupling, and applies Hoeffding's inequality to the independent copies. Write $V_j=\sum_{i\in B_j}w_i\xi_i$, so $\sum_iw_i\xi_i=\sum_{j\ \mathrm{odd}}V_j+\sum_{j\ \mathrm{even}}V_j$ and $\Prob\{\abs{\sum_iw_i\xi_i}\ge u\}\le\Prob\{\abs{\sum_{\mathrm{odd}}V_j}\ge u/2\}+\Prob\{\abs{\sum_{\mathrm{even}}V_j}\ge u/2\}$. The odd blocks are $\mu$ blocks separated by gaps of length $a$ (the even blocks). Apply Lemma~\ref{lem:yu} to $h=\ind\{\abs{\sum_{\mathrm{odd}}V_j}\ge u/2\}$, which is bounded by $1$: the probability under the true law is at most the probability under independent blocks plus $2(\mu-1)\beta(a)$. Under independent blocks with the same marginals the $V_j$ are independent, $\E V_j=0$, and $V_j\in[-W_j,W_j]$, so Hoeffding's inequality gives $\Prob\{\abs{\sum_{\mathrm{odd}}V_j}\ge u/2\}\le 2\exp\{-2(u/2)^2/\sum_{\mathrm{odd}}(2W_j)^2\}=2\exp\{-u^2/(8\sum_{\mathrm{odd}}W_j^2)\}$. The even blocks are treated in the same way with the odd blocks as gaps. Adding the two parities gives the penalty $4(\mu-1)\beta(a)$.
\end{proof}

A uniform-in-$r$ version follows as in Lemma~\ref{lem:grid}: apply Lemma~\ref{lem:yu} once to the indicator of the grid event and the union bound under independent blocks, so the factor $4K$ enters the logarithm and the mixing penalty is not multiplied by $K$.

\section{Numerical checks of the theory}
\label{app:checks}

The scripts in \texttt{code/theory\_checks} check statements of Section~\ref{sec:theory} numerically, with fixed seeds. A check is evidence about the configurations tried. It is not a proof, and some checks cannot fail.

\begin{itemize}
\item \emph{Lemma~\ref{lem:grid}.} For \thGridConfigsWord{} configurations (continuous and discrete laws, sample sizes \thGridNs, kernel and equal weights) and four values of $\eta$, with \thGridReps{} replications each, the largest frequency with which the sup-distance exceeded the bound is \thGridMaxExceed.
\item \emph{Theorem~\ref{thm:perstep}.} The exact conditional coverage is available in closed form in simulated models. In model A the features come from an independent driver series, so Assumption~\ref{ass:A} holds exactly. In models B1 to B4 the features are computed from the observed series, so part (i) of the assumption fails and part (ii) holds. Each model has \thPsSteps{} steps (\thPsSeeds{} seeds), with median effective sample size \thPsNeffMin{} to \thPsNeffMax. The lower bound \eqref{eq:perstep-lower} failed in no step of models A and B1 to B4 at $\eta\in\{0.05,0.1,0.2\}$, and its mean slack is \thPsSlackMin{} to \thPsSlackMax. Model B5 is model B1 with the lag-one forecaster, so part (ii) of the assumption fails there as well, and the lower bound failed in \thPsBFiveFailMin{} to \thPsBFiveFailMax{} per cent of its steps at the same values of $\eta$. The check of the upper bound \eqref{eq:perstep-upper} is uninformative: the bound is at least $1$ in \thPsHiOneMin{} to \thPsHiOneMax{} per cent of steps, so no failure can occur there.
\item \emph{Theorem~\ref{thm:free}.} On the level paths of the earlier preprint and on re-runs, the identity \eqref{eq:free-identity} holds to a gap of at most \thLrIdGapMax, and the range bound holds. In the runs of this paper the identity gap of the projected decomposition is at most \rvIdGapMax.
\item \emph{Proposition~\ref{prop:combine}.} In \thCombRuns{} runs of models A and B1 the margin between the two sides is \thCombMarginMin{} to \thCombMarginMax. The check cannot fail in these models, because the right-hand side is negative.
\item \emph{Lemma~\ref{lem:drift}.} The lemma was not violated in \thDriftSyntheticCases{} steps of a synthetic stream or in \thDriftRandomCases{} random configurations. The earlier inequality $D_t\le\bar d_t$ failed in \thDriftSyntheticOld{} of the steps and in \thDriftRandomOld{} of the configurations.
\item \emph{Proposition~\ref{prop:check}.} In model A, with $L$ between \thAdLpathMin{} and \thAdLpathMax{} on the realised feature paths, $\widehat L_{\min}(0,\eta)\le L$ in \thAdLeCount{} of \thAdSeeds{} replications, with $\eta=0.05/J$ for the $J$ clusters kept ($J$ from \thAdJMin{} to \thAdJMax), the level at which the maximum over pairs has confidence level $0.95$.
\item \emph{Proposition~\ref{prop:blocked}.} On $m$-dependent sequences, where $\beta(a)=0$ for $a>m$, the empirical tail was never above the blocked bound. The bound equals $1$ in \thMixAtOne{} of \thMixCells{} cells, so the check is informative in \thMixInfo{} cells only. The empirical tail exceeds the independent-case term (Section~\ref{sec:mixing}).
\end{itemize}

\section{Design tables}
\label{app:tables}

\FloatBarrier
{\AtBeginEnvironment{tabular}{\scriptsize\setlength{\tabcolsep}{3pt}}\input{tab_dgm}}
{\AtBeginEnvironment{tabular}{\scriptsize\setlength{\tabcolsep}{3pt}}\input{tab_data}}
{\AtBeginEnvironment{tabular}{\scriptsize\setlength{\tabcolsep}{3pt}}\input{tab_sim_headline_ridge}}
\FloatBarrier

\section{Statements of the earlier preprint that this paper replaces}
\label{app:replaced}

\begin{table}[h]
\centering
\caption{Statements of the earlier preprint and what replaces them.}
\label{tab:replaced}
\small
\begin{tabular}{p{5.6cm}p{2.6cm}p{6.6cm}}
\toprule
Earlier statement & Status & Replaced by \\
\midrule
Reported width gains of \rpGainElec{} and \rpGainWea{} per cent on electricity and weather & does not hold & Section~\ref{sec:results-anchor}: centring difference; like-for-like results in Sections~\ref{sec:results-real} and~\ref{sec:results-equal} \\
Spectral window defined through $Y_t$ & incorrect as written; the code was causal & Lemma~\ref{lem:measurable}: the window ends at $t-1$ \\
Long-run calibration theorem for the projected update with an $O(1/T)$ term & false & Theorem~\ref{thm:projected} and its counterexample \\
Per-step bound at the fixed level $\alpha$ with a computable bias term & stated at the wrong level & Theorem~\ref{thm:perstep} at the level $1-\alpha_t$; $L$ and $\delta$ cannot be computed \\
Concentration term proved by a union bound over random jump points & proof invalid & Lemma~\ref{lem:grid}, proved by a grid argument \\
Remark that the mixing penalty is below $\thOldClaim$ at the pool sizes used & arithmetic wrong & Section~\ref{sec:mixing}: the penalty is \mxOldValue; the blocked bound is vacuous \\
$D_t\le\bar d_t$ & false in general & Lemma~\ref{lem:drift} \\
Regression slope as an estimate of $L$; predicted and realised coverage gaps & not an estimate; dropped & Proposition~\ref{prop:check}: a refutation diagnostic \\
Claim that no earlier method has both a per-step computable bound and a long-run guarantee under a testable assumption & dropped & Section~\ref{sec:related}: each guarantee is in earlier work \\
DASC helps most when the forecaster leaves regime structure in the residuals & not supported on the real data & Sections~\ref{sec:results-real} and~\ref{sec:results-weights} \\
Contrast of DASC with rolling conformal alone in the year of the volatility episode & uninformative & Section~\ref{sec:results-real}: adaptive methods have similar coverage \\
\bottomrule
\end{tabular}
\end{table}

\input{appendix_exploratory}

%% file: tab_dgm.tex
\begin{table}[t]
\centering
\caption{Simulated scenarios. Every stream is a sinusoid plus noise except the stationary and white-noise streams; ``regimes'' alternate in blocks of the stated length. The evaluation covers every step from step 220.}
\label{tab:dgm}
\small
\begin{tabular}{llp{10.8cm}}
\toprule
Scenario & Length & Design \\
\midrule
Original three-phase & 1400 & Sinusoid; frequency 0.035 until step 450, 0.055 until step 900, then rising by 0.00006 per step; amplitude 1.0 and 1.7 in alternating blocks of 180 steps; from step 900 level $+1.2$ and noise sd 0.75 (0.45 before). \\
Recurring: frequency only & 1400 & Sinusoid of amplitude 1.0 and noise sd 0.5; frequency 0.03 and 0.09 in alternating blocks of 150 steps. \\
Recurring: frequency and noise & 1400 & As above with noise sd 0.3 in the first regime and 0.9 in the second. Built for spectral weighting. \\
Recurring: amplitude only & 1400 & Frequency 0.05, noise sd 0.5; amplitude 0.5 and 2.0 in alternating blocks of 150 steps. The normalised features cannot separate the regimes. \\
Stationary AR(1) & 1400 & AR(1) with coefficient 0.7 and innovation sd 0.5. \\
Pure drift & 1400 & Sinusoid of frequency 0.05; level rising by 0.002 per step; noise sd rising linearly from 0.3 to 1.0. \\
White noise & 1400 & Independent $N(0,1)$ values. \\
Stress: abrupt shift & 1500 & Sinusoid of frequency 0.04, amplitude 1.2 and 1.65 in blocks of 180 steps, noise sd 0.45; from step 900 level $+1.4$ and noise sd 0.8. \\
Stress: gradual frequency & 1500 & As above with frequency $0.04+0.00006\,(t-500)^+$ and no shift in level or noise. \\
Stress: heavy tail & 1500 & As above with Student $t_3$ noise scaled by 0.35 and level $+0.8$ from step 900. \\
Stress: mixed drift & 1500 & As above with level $+1.0$ and noise sd 0.75 from step 850 and frequency $0.04+0.00004\,(t-700)^+$. \\
Stress: weak recurrence & 1500 & Amplitude $0.25+0.05\,\varepsilon_t$ with $\varepsilon_t\sim N(0,1)$, frequency $0.02+0.00003\,t$, noise sd 0.45; from step 850 level $+0.8$ and noise sd 0.75. \\
\bottomrule
\end{tabular}
\end{table}

%% file: tab_data.tex
\begin{table}[t]
\centering
\caption{Real streams. Warm-up: steps before the first evaluated step, used to fill the calibration pool. Development steps: steps used for tuning (the first 25 per cent of the stream, after warm-up). Evaluation steps: the remainder. Block: block length of the paired bootstrap.}
\label{tab:data}
\small
\begin{tabular}{llrrrrr}
\toprule
Data set & Series & Length & Warm-up & Development & Evaluation & Block \\
\midrule
Electricity & hourly mean of household global active power & 12{,}000 & 1080 & 1{,}920 & 9{,}000 & 168 hours \\
Weather & hourly two-metre temperature, Dallas & 24{,}000 & 1440 & 4{,}560 & 18{,}000 & 168 hours \\
Finance & absolute daily log return, S\&P~500, percentage points & 2{,}514 & 252 & 376 & 1{,}886 & 21 days \\
\bottomrule
\end{tabular}
\end{table}

%% file: tab_sim_headline_ridge.tex
\begin{table}[t]
\centering
\caption{DASC (projected) minus ACI (projected) in each simulated scenario with the rolling ridge forecaster: mean (standard error) over replications of the relative width difference in per cent of ACI's width, and of the interval-score difference. Positive values mean that DASC is wider or scores worse. ``Default'': DASC with $\gamma$ fixed and ACI tuned (main runs). ``Tuned'': DASC with $\gamma$ tuned over the ACI grid (equal-budget study).}
\label{tab:sim-headline-ridge}
\small
\begin{tabular}{lllll}
\toprule
Scenario & Width, default & Score, default & Width, tuned & Score, tuned \\
\midrule
Original three-phase & +0.40 (0.11) & 0.000 (0.002) & +0.33 (0.12) & $-$0.001 (0.002) \\
Recurring: frequency only & +0.67 (0.09) & +0.015 (0.001) & +0.10 (0.05) & +0.001 (0.000) \\
Recurring: frequency and noise & +1.14 (0.19) & +0.059 (0.004) & $-$0.01 (0.10) & +0.009 (0.002) \\
Recurring: amplitude only & +0.74 (0.12) & +0.014 (0.001) & +0.06 (0.08) & +0.004 (0.001) \\
Stationary AR(1) & +0.63 (0.09) & +0.013 (0.001) & +0.08 (0.06) & +0.001 (0.000) \\
Pure drift & +1.30 (0.12) & +0.026 (0.001) & +0.24 (0.07) & +0.005 (0.001) \\
White noise & +0.71 (0.11) & +0.031 (0.001) & +0.07 (0.07) & +0.006 (0.001) \\
Stress: abrupt shift & +0.54 (0.10) & 0.000 (0.001) & +0.54 (0.10) & 0.000 (0.001) \\
Stress: gradual frequency & +0.42 (0.09) & +0.009 (0.001) & +0.42 (0.09) & +0.009 (0.001) \\
Stress: heavy tail & +1.54 (0.23) & +0.011 (0.002) & +0.35 (0.18) & $-$0.003 (0.001) \\
Stress: mixed drift & +0.36 (0.08) & +0.006 (0.001) & +0.36 (0.08) & +0.006 (0.001) \\
Stress: weak recurrence & +0.28 (0.10) & +0.002 (0.001) & +0.49 (0.10) & +0.003 (0.001) \\
\bottomrule
\end{tabular}
\end{table}

%% file: appendix_exploratory.tex
\section{An exploratory bootstrap refinement of the refutation diagnostic}
\label{app:exploratory}

This appendix reports a variant of the diagnostic of Definition~\ref{def:Lmin} that was studied by simulation. It has no proof, and it is not a test with valid size.

\emph{The variant.} For each pair of clusters it takes the Kolmogorov distance between the two empirical distribution functions of the residuals on a grid of \loGridPts{} pooled quantiles, subtracts a margin, divides by the mean feature distance of the pair and maximises over pairs. The margin is the \loMarginPct{} per cent quantile, over \loBootDraws{} circular block-bootstrap resamples of the time-ordered residuals, of the largest difference between the cluster-centred indicator processes. All clusters share the blocks, which keeps their dependence, and the block length is estimated from the within-cluster rank indicators. The diagnostic of Definition~\ref{def:Lmin} uses concentration terms that assume independence in place of this margin.

\emph{Simulation.} The feature path has dimension one or two, each coordinate an $\mathrm{AR}(1)$ of coefficient \loPhiS{}. The residual scale depends on the feature, with \loNScalesWord{} slopes that give true Lipschitz constants of \loLLarge{}, \loLMid{} and \loLSmall{} on the realised paths. The residuals are Gaussian $\mathrm{AR}(1)$ with correlation \loRhoList{}, at sample sizes \loNList{}. There are \loCellsPerScale{} cells per slope and \loGridCells{} in all, with \loRepsLarge{} replications per cell at the two larger slopes and \loRepsSmall{} at the smallest.

\emph{Results.} The lower bound of the variant was at or below the true constant with frequency at least \loBTwoMinCover{} in all \loGridCells{} cells. This holds for the refined variant only: two other bootstrap variants were run, one with a block length set by the sample size alone and one with the margin standardised by the bootstrap scale of each pair, and the coverage of each falls below \loCoverThr{} in some cells at the smallest slope. At the largest slope the power of the variant to refute a constant of one third of the true value was at least \loPowerThr{} in \loBTwoLargePowerPath{} of the \loCellsPerScale{} cells, against \loCZeroLargePowerPath{} for the independence-based bound. At the middle slope the counts are \loBTwoMidPowerPath{} and \loCZeroMidPowerPath{}. At the smallest slope no method reaches that power in any cell, so the power vanishes for small constants.

\emph{Limits.} The independence-based bound also kept its size below \loCZeroMaxSize{} in every cell, so the grid does not show that the blocking makes the variant valid. The margins are not calibrated under strong dependence: at the largest slope they covered the true sup-norm differences of all pairs with frequency at least \loCoverThr{} in only \loBTwoMarginOkLarge{} of the \loCellsPerScale{} cells (nominal \loMarginNominal{}; lowest value \loBTwoMarginMin{} over the \loNScalesWord{} slopes). Outside the grid, at residual correlation \loExploreRho{} and sample size \loExploreN{}, the size reached \loExploreSize{}. A proof would need a bootstrap consistency result for the simultaneous sup-norm quantile with $\beta$-mixing residuals whose law changes with time, data-dependent clusters and a growing number of pairs. None is given. Real data, features that depend on past residuals, non-Gaussian dependence and more than two feature dimensions are not covered.